\documentclass[paper=a4, fontsize=11pt]{scrartcl}
\usepackage[T1]{fontenc}
\usepackage{fourier}
\usepackage[english]{babel} 
\usepackage{amssymb}
\usepackage{amsmath, amsfonts, amsthm} 
\usepackage{mathrsfs}
\usepackage{graphicx}
\usepackage{url}
\usepackage[title, titletoc, page]{appendix} 
\usepackage{listings} 
\usepackage[table,xcdraw]{xcolor}
\usepackage{comment}
\usepackage{caption}
\usepackage{subcaption}
\usepackage[ruled,vlined]{algorithm2e}
\usepackage{todonotes}
\usepackage{sectsty}
\allsectionsfont{\centering \normalfont\scshape}
\usepackage{float}
\usepackage{placeins}
\usepackage{tikz}
\usetikzlibrary{decorations.pathreplacing,calc}
\usepackage{cite} 
\usepackage{bm}
\usepackage{thmtools}

\usepackage{fancyhdr}
\numberwithin{equation}{section}		
\numberwithin{figure}{section}			
\numberwithin{table}{section}				

\newtheorem{thm}{Theorem}[section]
\newtheorem{defi}[thm]{Definition}
\newtheorem{lem}[thm]{Lemma}

\newtheorem{cor}[thm]{Corollary}
\newtheorem{remark}{Remark}[section]

\newtheorem{prop}[thm]{Proposition}

\newcommand{\vect}[1]{\boldsymbol{\mathbf{#1}}}

\newcommand{\leqs}{\leqslant}

\newcommand{\bluebox}[1]{{\setlength{\fboxsep}{1pt}\setlength{\fboxrule}{0.5pt}\fcolorbox{blue}{white}{\ensuremath{#1}}}}

\title{Generalized Hankel/Toeplitz matrix for array signal processing}

\author{Wenchong Huang\footnote{
  School of Mathematical Sciences, Zhejiang University, Hangzhou, 310027, China. (vongvuncung@zju.edu.cn, kpli@zju.edu.cn, pingliu@zju.edu.cn)} \and Kunpeng Li \footnotemark[1]
\and Ping Liu \footnotemark[1]  \thanks{
Institute of
Fundamental and Transdisciplinary Research, Zhejiang University, Hangzhou, 310027, China.}
}

\begin{document}
\maketitle
\begin{abstract}

In this paper, we introduce generalized Hankel/Toeplitz matrices (GHM/GTM) and the associated generalized Vandermonde decomposition for nonuniform array signal processing and multi-dimensional super-resolution. The proposed framework was discovered from the study of resolution limit theory and extends the classical Hankel/Toeplitz structure by allowing substantially more flexible sampling geometries while preserving the underlying low-rank Vandermonde factorization. Through devising an optimal algorithm based on this GHM framework, we derive the state-of-the-art upper bound estimate for the computational resolution limit (CRL) of source-number detection in general $d$-dimensional super-resolution problems. For segmented sampling sets, whose geometry is closely related to sparse and distributed arrays, we establish deterministic lower bounds for the minimum singular values of the associated generalized Vandermonde matrices and derive corresponding stability and number-detection guarantees for multi-clump source configurations. To address the computational bottleneck of conventional multi-level Hankel constructions in high dimensions, we further introduce randomized GHM constructions whose matrix dimensions scale with the effective degrees of freedom rather than with the full tensor-product grid, together with deterministic recovery guarantees conditional on the realized Vandermonde factors. We also extend the framework to source localization by developing GHM-based MUSIC algorithms for nonuniform measurements, with stability characterized through the conditioning of the generalized Vandermonde factors. Numerical experiments on synthetic data demonstrate that the proposed GHM-based methods achieve competitive resolution and recovery accuracy while substantially reducing matrix size and computational cost, especially in high-dimensional settings.

\end{abstract}

    

\section{Introduction}
Resolving point sources from band-limited measurements is a fundamental problem arising in imaging, line spectral estimation, array processing, radar, sonar, seismology, and astronomy
\cite{donoho1992superresolution,candes2014towards,tang2013compressed,kay1988modern}.
In these problems, the unknown object is commonly modeled as a discrete measure
\[
    \mu=\sum_{j=1}^{n}a_j\delta_{\mathbf y_j},
\]
and the goal is to recover the number of sources, their locations, and their amplitudes from noisy band-limited measurements.
Over the past several decades, significant progress has been made in the mathematical theory and computational methods of super-resolution
\cite{liu2021mathematical}.
Nevertheless, much of the existing research in applied mathematics, both on resolution-limit theory and on super-resolution algorithms, has focused on uniformly sampled Fourier measurements; see Section~\ref{sec:srlimitandalgorithm}.
In comparison, super-resolution from incomplete or nonuniform Fourier measurements, which naturally arises from sparse, nonuniform, and distributed array configurations, has received considerably less attention.
This paper is concerned with the recovery of the number and locations of point sources from such incomplete and nonuniform Fourier measurements, with particular emphasis on multi-dimensional and high-resolution regimes; the connection with nonuniform array signal processing is discussed in Section~\ref{sec:worknonuniformarray}.
At the center of our study is a generalized Hankel/Toeplitz framework, which arises naturally from our investigation of the computational resolution limits of multi-dimensional super-resolution and provides the structural foundation for the theoretical analysis and algorithms developed in this work.

\subsection{resolution limit and super-resolution algorithms}\label{sec:srlimitandalgorithm}

Rayleigh's classical resolution criterion compares source separation with the width of the diffraction pattern \cite{rayleigh1879xxxi}. Under the Fourier normalization in \eqref{equ:modelsetting1}, the corresponding Rayleigh length is of order \(\pi/\Omega\), where \(\Omega\) is the cutoff frequency. Although the Rayleigh criterion plays an important role in imaging and signal processing, it is a heuristic, model-dependent criterion rather than a theorem about computational recovery. In particular, it does not quantify what can be recovered after algorithmic processing
\cite{papoulis1979improvement,den1997resolution,liu2021theory}.
Therefore, a rigorous mathematical definition of resolution should take into account not only the cutoff frequency, but also the noise level, the signal strength and the sparsity of the sources.

This observation motivates the concept of the \emph{computational resolution limit} (CRL). Roughly speaking, the CRL is the minimum separation distance between point sources such that a given recovery task is possible under a prescribed noise level \cite{liu2021theory,liu2022mathematical,liu2021mathematical}. In super-resolution, two related but different tasks should be distinguished. The first is \emph{number detection}, namely recovering the exact number of sources. The second is \emph{location recovery}, or support recovery, namely stably reconstructing the source locations. These two tasks have different mathematical nature and usually require different separation distances. In the one-dimensional theory, the CRL for number detection is of order
\(   \frac{1}{\Omega}
    \left(\frac{\sigma}{m_{\min}}\right)^{\frac{1}{2n-2}}\),
whereas the CRL for support recovery is of order
\(\frac{1}{\Omega}
    \left(\frac{\sigma}{m_{\min}}\right)^{\frac{1}{2n-1}}\). Here \(\sigma\) denotes the noise level and \(m_{\min}\) denotes the minimum source amplitude. These estimates were later extended to multi-dimensional spaces \cite{liu2021mathematicalhighd, liu2024improved}.


The mathematical theory of super-resolution has been developed from several different perspectives. Donoho first studied the resolution problem from the viewpoint of optimal recovery and showed that the instability of super-resolution is governed by noise amplification and sparsity
\cite{donoho1992superresolution}. Subsequent works obtained sharper minimax estimates by analyzing the conditioning of the associated measurement matrices
\cite{demanet2015recoverability,batenkov2020conditioning,li2021stable}.
In the off-the-grid setting, sharp phase-transition phenomena and minimax rates have been derived using extremal functions, Prony-type systems and quantitative singularity theory
\cite{Moitra:2015:SEF:2746539.2746561,akinshin2015accuracy,batenkov2019super}.
These works reveal that the super-resolution factor, the signal-to-noise ratio and the number of clustered sources jointly determine the difficulty of the inverse problem.

There is also a large body of work on super-resolution algorithms. One important class consists of sparsity-promoting convex optimization methods, including total-variation minimization, BLASSO and atomic norm minimization
\cite{candes2014towards,candes2013super,azais2015spike,duval2015exact,tang2013compressed,tang2014near,chi2020harnessing}.
These methods can provably recover off-the-grid sources under suitable minimum separation or non-degeneracy conditions. However, for general signed or complex sources, the required separation is often several Rayleigh lengths, which limits their applicability in the strongly clustered super-resolution regime. A convex algorithm was proposed to resolve tightly clustered sources \cite{yang2024separation}, but the stability guarantee remains unavailable.


Another important class consists of parametric and subspace methods, including Prony's method, MUSIC, ESPRIT and the Matrix Pencil method
\cite{Prony-1795,schmidt1986multiple,stoica1989music,roy1989esprit,hua1990matrix,hua1991svd}.
These methods often exhibit favourable numerical performance in the sub-Rayleigh regime. Nevertheless, their stability properties in the non-asymptotic regime are still not fully understood. Moreover, most of these algorithms require a priori information about the model order, namely the number of sources, and their performance may depend sensitively on this number
\cite{hansen2018superfast,liao2016music,li2021stable,li2020super}.
This makes the number detection problem not merely a preliminary step, but a central theoretical and computational problem in super-resolution.


\subsection{Non-uniform and distributed arrays}\label{sec:worknonuniformarray}

The standard array model in line spectral estimation and direction-of-arrival estimation is the uniform linear array (ULA). In this setting, the sensor locations are equally spaced, and the corresponding sensing matrix has the classical Vandermonde structure with consecutive row indices. This consecutive-index structure is one of the main reasons why classical subspace methods, such as MUSIC, ESPRIT and the Matrix Pencil method, admit efficient implementations and clean algebraic interpretations
\cite{schmidt1986multiple,roy1989esprit,hua1990matrix}.

However, ULAs are not always desirable or feasible in practical sensing systems. Increasing the aperture of a fully populated ULA usually requires a proportional increase in the number of sensors, which leads to higher hardware cost, stronger mutual coupling, and larger computational burden. This motivates the use of sparse and non-uniform arrays, which aim to achieve comparable aperture or resolution with significantly fewer physical sensors
\cite{pal2010nested,vaidyanathan2011sparsesamplers}.
Typical examples include nested arrays \cite{pal2010nested}, coprime arrays \cite{vaidyanathan2011sparsesamplers}, minimum-redundancy arrays \cite{moffet1968minimum}, and their variants. A central mechanism behind many of these designs is the \emph{coarray}, namely a virtual array generated by sums or differences of physical sensor locations. The coarray determines the effective aperture, the spatial degrees of freedom, and the number of sources that can be resolved by the array. Under suitable assumptions, sparse arrays can identify up to \(O(P^2)\) uncorrelated sources using only \(P\) physical sensors
\cite{pal2010nested,vaidyanathan2011sparsesamplers,wang2017,koochakzadeh2016cramerrao,liu2017cramerrao}.

Another important class of non-uniform measurements is provided by distributed arrays
\cite{heimiller1983distributed,1982Distributed}.
A typical distributed array consists of several locally uniform subarrays separated by relatively large gaps. The measurements are chosen in the following form:
\begin{equation}\label{eq:distributed-array-measurements}
\begin{array}{llcl}
\vect{Y}(D_1),&\vect{Y}(D_1+d),&\cdots,&\vect{Y}(D_1+m_1d),\\
\vect{Y}(D_2),&\vect{Y}(D_2+d),&\cdots,&\vect{Y}(D_2+m_2d),\\
& & \vdots &\\
\vect{Y}(D_k),&\vect{Y}(D_k+d),&\cdots,&\vect{Y}(D_k+m_kd).
\end{array}
\end{equation}
Here \(d\) is the inner spacing of each local uniform subarray, \(D_i\) is the initial position of the \(i\)-th subarray, and \(m_i+1\) is the number of measurements in that subarray. Thus, the measurements are locally uniform inside each block, but the full measurement set is globally non-uniform.

To recover source locations from distributed-array measurements, a collection of methods based on the \emph{virtual interpolated array} (VIA) has been developed
\cite{1992Performance,1992Direction,gavish1993viaesprit,Weiss1993Performance,Gershman1997A,Weiss1995Direction,li_preestimation-based_2009}.
The main idea is to choose a list of virtual uniform measurement positions
\[
    \widehat{\eta}_i=(i-1)d,\qquad i=1,\ldots,h_V,
\]
and then compute the virtual measurements from the actual non-uniform measurements. After this interpolation step, one can apply algorithms originally developed for uniform arrays to the virtual measurements. In this way, the distributed-array problem is transformed into a ULA-type problem. This strategy is useful and has been widely studied, but the interpolation step may introduce modelling errors, noise amplification and additional computational cost.

Multiple-invariance ESPRIT combines several array invariances within one subspace formulation \cite{swindlehurst_multiple_1992}. Related multiscale methods resolve spatial ambiguity by combining a high-resolution ambiguous estimate with a coarse unambiguous estimate \cite{WongKT1998Direction,Vasylyshyn2004Closed}. A typical example is the double-resolution ESPRIT-type method, where the invariance property between widely spaced subarrays gives a high-resolution estimate with periodic ambiguity, while the invariance property inside a short-spacing subarray gives a coarse but unambiguous estimate. The coarse estimate is then used to determine the correct ambiguity branch.

Many of the distributed-array methods reviewed above form covariance or correlation matrices and therefore require multiple snapshots or assumptions that permit covariance estimation. This limits their direct use in single-measurement or coherent-source settings. With one measurement vector, standard data-matrix formulations of MUSIC and ESPRIT often use Hankel or Toeplitz liftings, whose usual construction requires uniform measurement positions. Thus, the reviewed constructions commonly rely either on covariance or correlation data or on uniform positions for these liftings.

The present paper takes a different viewpoint within the band-query model of \eqref{equ:modelsetting1}. We prescribe non-uniform Fourier frequencies directly and organize the queried values into matrices whose Vandermonde factors have rows selected from separated consecutive blocks or random indices. This enables us to develop efficient super-resolution algorithms for single-measurement or coherent-source settings.




\subsection{Model setting and our contribution}
Let us introduce the mathematical model for the multi-dimensional super-resolution problem. Fix $d,n\in\mathbb N$, a cutoff frequency $\Omega>0$, and a noise level $\sigma>0$. We model the collection of point sources as a discrete measure $\mu = \sum_{j=1}^{n} a_j \delta_{\mathbf{y}_j}$, where $\mathbf{y}_j \in\mathbb{R}^d$ ($j=1, \ldots, n$) represents the location of the $j$-th point source and $a_j \in\mathbb{C}$ denotes its complex amplitude. Throughout the paper, every atomic representation is assumed to be reduced: its nodes are pairwise distinct and all its amplitudes are nonzero. Such a representation is unique up to ordering, so $n=|\operatorname{supp}\mu|$ and the minimum amplitude $m_{\min} = \min_{1 \le j \le n} |a_j|$ are intrinsic to $\mu$. The same convention applies to every discrete measure introduced below, including admissible measures. We will use support recovery as a substitute for location reconstruction throughout the paper.

The measurement consists of the noisy Fourier data of $\mu$ within a bounded frequency band:
\begin{equation}\label{equ:modelsetting1}
\mathbf{Y}(\boldsymbol{\omega}) = \mathcal{F}[\mu](\boldsymbol{\omega}) + \mathbf{W}(\boldsymbol{\omega}) = \sum_{j=1}^{n} a_j e^{i \mathbf{y}_j \cdot \boldsymbol{\omega}} + \mathbf{W}(\boldsymbol{\omega}), \quad \boldsymbol{\omega} \in [-\Omega, \Omega]^d,
\end{equation}
where $\mathbf W:[-\Omega,\Omega]^d\to\mathbb C$ is the noise bounded by $|\mathbf{W}(\boldsymbol{\omega})| < \sigma, \ \boldsymbol{\omega} \in [-\Omega, \Omega]^d$. The inverse problem is then to reover the source number, locations and amplitudes from sampled measurments $\vect Y(\omega_j)$'s. 




Based on this model, we have several contributions in this paper summarised as follows. 

\begin{enumerate}
    \item \textbf{A generalized Hankel/Toeplitz framework for nonuniform Fourier sampling.}
    We introduce \emph{generalized Hankel/Toeplitz matrices} and \emph{generalized Vandermonde decomposition} for array signal processing and super-resolution problems. This framework organizes nonuniform Fourier sampling constructions in a unified way, and reduces the matrix dimensions to scale with the effective degrees of freedom, i.e., the number of sources and the spatial dimensionality.

\item \textbf{A state-of-the-art multi-dimensional computational resolution limit for source-number detection.}
Exploiting the flexibility of the GHM framework, we design a deterministic source-number detection scheme and derive the currently sharpest explicit upper bound for the computational resolution limit (CRL) in general $d$-dimensional super-resolution problems.

  \item \textbf{Deterministic minimum-singular-value and thresholding guarantees for segmented sampling sets.}
For segmented frequency sets, whose index geometry is closely related to that of sparse and distributed arrays, we establish deterministic lower bounds for the minimum singular values of the associated generalized Vandermonde matrices (Theorem~\ref{thm:segmented-vandermonde}) and translate these bounds into explicit singular-value-thresholding guarantees for source-number detection (Theorem~\ref{thm:segmented_threshold}). Despite the incomplete and nonuniform nature of the measurements, our resolution analysis shows that the resulting GHM-based algorithms attain the optimal resolution order for sources distributed in several well-separated clusters. In particular, the resolution exponent is governed by the size of the largest local cluster rather than by the total number of sources, demonstrating that the proposed segmented constructions preserve the essential super-resolution capability of fully sampled schemes while using substantially fewer measurements.

    \item \textbf{Randomized GHM constructions that break the curse of dimensionality.}
We develop randomized GHM constructions whose matrix dimensions are governed by the intrinsic problem size, essentially the number of sources and the spatial dimension, rather than by the exponentially growing full tensor grid. This substantially reduces both storage and SVD complexity in high-dimensional super-resolution while preserving the generalized Vandermonde structure required for source-number detection and localization.


    \item \textbf{GHM-based MUSIC for location recovery.}
    Based on the GHM factorization, we extend MUSIC to nonuniform measurements using the row-frequency steering map and a partial SVD, and show that its stability under noise is controlled by the minimum singular values of the Vandermonde factors. We also discuss a coarse-to-local peak-search strategy and its computational cost, and evaluate segmented and random-array constructions numerically.

\end{enumerate}

    Taken together, these results provide a unified theoretical and computational framework connecting sampling geometry, Vandermonde conditioning, computational resolution, source-number detection, and subspace localization. They further reveal that the fundamental structure behind classical Hankel/Toeplitz-based methods is the underlying Vandermonde decomposition, rather than the specific Hankel/Toeplitz arrangement. By decoupling the sampling geometry from this rigid matrix structure, the proposed framework enables more flexible deterministic and randomized sensing designs, reduces the computational burden in high dimensions, and offers a general pathway for extending classical structured-matrix algorithms to sparse, segmented, nonuniform, and distributed measurement settings.

\subsection{Organization}

The remainder of the paper is organized as follows. Section \ref{sec:generalizedhankeltoeplitz} introduces the generalized Hankel/Toeplitz matrices and the generalized Vandermonde decomposition. Section~\ref{sec:number-detection} studies number detection in multi-dimensional super-resolution problems. It first defines the CRL and establishes an explicit upper bound under a sufficient separation condition (Subsection~\ref{subsec:crl-number}), and then develops GHM-based number detection algorithms (Subsection~\ref{subsec:uniform array}). Section~\ref{sec:nonuniform-arrays} treats nonuniform arrays and faster number detection algorithms. Subsection~\ref{subsec:equally_distributed} introduces segmented sampling sets, establishes minimum-singular-value lower bounds and thresholding guarantees for the associated Vandermonde matrices, and derives fast number detection algorithms; Subsection~\ref{subsec:rand-numdetect} presents randomized GHM constructions; numerical experiments are reported at the end of the section. Section~\ref{sec:location-recovery} turns to location recovery. It presents a unified MUSIC framework for classical and generalized Hankel matrices and discusses its implementation and cost (Subsection~\ref{subsec:music-classical}), establishes its noise stability (Subsection~\ref{subsec:music-stability}), compares random arrays with full tensor-grid arrays (Subsection~\ref{subsec:random-array-music}), reports numerical experiments including dimension scaling up to $d=8$ (Subsection~\ref{subsec:music-numerics}). Section \ref{sec:conclusionandfuture} concludes the paper and proposes several future works. The appendices estimate the minimum singular value of generalized Vandermonde matrices. 

\section{Generalized Hankel/Toeplitz matrix and Vandermonde decomposition}\label{sec:generalizedhankeltoeplitz}

Hankel and Toeplitz matrices constitute the algebraic backbone of many
classical subspace-based methods in array signal processing, spectral
estimation, and harmonic retrieval, including MUSIC, ESPRIT, matrix-pencil,
and Prony-type methods. Their effectiveness originates from the fact that,
for exponential measurements, the corresponding structured matrices admit
low-rank Vandermonde factorizations. This
section first recalls several structural limitations of the classical
constructions, and then introduces a generalized Hankel/Toeplitz matrix
together with the associated generalized Vandermonde decomposition. This framework
decouples the sampling geometry from the sum--difference constraint of the
classical Hankel/Toeplitz structure and serves as the unified matrix framework of this paper.


\subsection{Structural Limitations of Hankel/Toeplitz Matrices and Vandermonde Decompositions}
\label{subsec:bottlenecks}

Although the Hankel/Toeplitz matrix is prevalent in the field of array signal processing, they have several structural limitations. The first inefficiency is the strong redundancy of a
Hankel/Toeplitz embedding. A dense $M_1\times M_2$ Hankel matrix, for instance,
contains $M_1M_2$ matrix entries but is determined by only $M_1+M_2-1$
distinct  samples.  Thus, increasing the matrix dimension does not create a comparable number of independent measurements.



The second difficulty concerns the amount of data and the computational work
required by the classical Vandermonde structure. In the model
\eqref{equ:modelsetting1}, the Rayleigh resolution limit is of order
$1/\Omega$, so a finer resolution requires a wider band
$[-\Omega,\Omega]^d$, that is, a larger cutoff frequency
$\Omega$. A classical Vandermonde matrix, however, is built from
consecutive sampling indices. Enlarging the band therefore forces a
larger and denser uniform sampling grid, where the grid spacing is constrained by the Nyquist rate. In $d$ dimensions, a multi-level
Hankel/Toeplitz construction with $L$ indices per coordinate is an
$L^d\times L^d$ matrix and already requires about $(2L)^d$ distinct Fourier
samples. The computational cost is equally severe. A dense SVD of such a matrix costs $\mathcal O(L^{3d})$ operations and
$\mathcal O(L^{2d})$ memory; this exponential growth in $d$ is known as the
curse of dimensionality. 

The third limitation is that the uniform structure leads to a substantial loss of flexibility in array
design. For example, measurements collected from irregular, sparse, or separately optimized nonuniform arrays cannot in general
be  embedded directly into a standard Hankel or Toeplitz matrix.  Consequently,
potentially useful measurements may be discarded or only indirectly utilized.

\subsection{Generalization}

To address the above limitations, we introduce a generalized Hankel/Toeplitz matrix (GHM/GTM) framework together with the associated generalized Vandermonde decomposition. The key idea is to retain the low-rank Vandermonde structure that underlies classical Hankel/Toeplitz-based methods, while relaxing the rigid requirement of consecutive and uniformly spaced sampling indices. In the generalized framework, the row and column index sets can be selected much more flexibly and, in particular, may be nonuniform, sparse, segmented, or randomized. This additional freedom enables a substantially more efficient use of the available Fourier measurements and allows the matrix dimensions to be adapted to the intrinsic degrees of freedom of the source configuration rather than to the size of a full tensor-product sampling grid. From this perspective, the essential structure is no longer the
Hankel/Toeplitz pattern itself, but the generalized Vandermonde decomposition
that it induces.



\begin{defi}[Generalized Hankel and Toeplitz matrices]\label{def:generalized-hankel-toeplitz}
Let $\mathcal A=\{\boldsymbol\alpha_1,\ldots,\boldsymbol\alpha_{M_1}\}$ and $\mathcal B=\{\boldsymbol\beta_1,\ldots,\boldsymbol\beta_{M_2}\}$ be finite subsets of $\mathbb R^d$ with fixed orderings, and let $f:\mathcal D\to\mathbb C$ for some $\mathcal D\subset\mathbb R^d$. If $\mathcal A+\mathcal B\subset\mathcal D$, the generalized Hankel matrix (GHM) generated by $f$ is
\[
\mathbf{GH}_f(\mathcal A,\mathcal B)
:=\bigl[f(\boldsymbol\alpha_p+\boldsymbol\beta_q)\bigr]_{1\le p\le M_1,\,1\le q\le M_2}.
\]
If $\mathcal A-\mathcal B\subset\mathcal D$, the generalized Toeplitz matrix (GTM) generated by $f$ is
\[
\mathbf{GT}_f(\mathcal A,\mathcal B)
:=\bigl[f(\boldsymbol\alpha_p-\boldsymbol\beta_q)\bigr]_{1\le p\le M_1,\,1\le q\le M_2}.
\]
Here $\mathcal A\pm\mathcal B:=\{\boldsymbol\alpha\pm\boldsymbol\beta:\boldsymbol\alpha\in\mathcal A,\ \boldsymbol\beta\in\mathcal B\}$, and both matrices belong to $\mathbb C^{M_1\times M_2}$.
\end{defi}

In the observation model \eqref{equ:modelsetting1}, for the noiseless Fourier signal $f(\boldsymbol\omega)=\sum_{j=1}^{n} a_j e^{i\boldsymbol\omega\cdot\mathbf y_j}$,  the generalized Hankel/Toeplitz matrix admits the following generalized Vandermonde decomposition, which plays a key role in tackling signal processing problems in this paper.  

\begin{defi}[Generalized Vandermonde matrix]\label{generalized vandermonde}
Let $\Gamma=\{\boldsymbol\gamma_1,\ldots,\boldsymbol\gamma_M\}\subset\mathbb R^d$ be a finite frequency set with a fixed ordering, and let $\mathcal X=\{\mathbf y_1,\ldots,\mathbf y_n\}\subset\mathbb R^d$ be a finite node set with a fixed ordering. The generalized Vandermonde matrix associated with $\Gamma$ and $\mathcal X$ is
\begin{equation}\label{equ:gen_vander_matrix0}
\mathcal V_\Gamma(\mathcal X)
:=\bigl[e^{i\boldsymbol\gamma_\ell\cdot\mathbf y_j}\bigr]_{1\le \ell\le M,\,1\le j\le n}
=\bigl[\boldsymbol\phi_\Gamma(\mathbf y_1),\ldots,\boldsymbol\phi_\Gamma(\mathbf y_n)\bigr]
\in\mathbb C^{M\times n},
\end{equation}
where $\boldsymbol\phi_\Gamma(\mathbf y):=(e^{i\boldsymbol\gamma_\ell\cdot\mathbf y})_{\ell=1}^M$ is the corresponding steering vector.
\end{defi}


\begin{defi}[Generalized Vandermonde decomposition]\label{def:generalized-vandermonde-decomposition}
Let $A\in\mathbb C^{M_1\times M_2}$, and let $\mathcal A,\mathcal B\subset\mathbb R^d$ be finite frequency sets with fixed orderings, $|\mathcal A|=M_1$, and $|\mathcal B|=M_2$. We say that $A$ admits a generalized Vandermonde decomposition if there exist a common node set $\mathcal X=\{\mathbf y_1,\ldots,\mathbf y_n\}\subset\mathbb R^d$ and a diagonal matrix $D=\operatorname{diag}(c_1,\ldots,c_n)$ with $c_j\in\mathbb C\setminus\{0\}$ for $j=1,\ldots,n$ such that either
\[
A=\mathcal V_{\mathcal A}(\mathcal X)D\mathcal V_{\mathcal B}(\mathcal X)^\top
\qquad\text{or}\qquad
A=\mathcal V_{\mathcal A}(\mathcal X)D\mathcal V_{\mathcal B}(\mathcal X)^*.
\]
\end{defi}

Classical Hankel and Toeplitz matrices and their multi-level analogs are special cases of Definition~\ref{def:generalized-hankel-toeplitz}. To connect the abstract definition with the observation model \eqref{equ:modelsetting1}, let
\[
\mathcal A=\{\boldsymbol\alpha_1,\ldots,\boldsymbol\alpha_{M_1}\},
\qquad
\mathcal B=\{\boldsymbol\beta_1,\ldots,\boldsymbol\beta_{M_2}\}
\subset\mathbb R^d
\]
and assume $\mathcal A+\mathcal B\subset[-\Omega,\Omega]^d$. The observed matrix
\begin{equation}\label{eq:general-observed-vdm}
A:=\bigl[\mathbf Y(\boldsymbol\alpha_p+\boldsymbol\beta_q)\bigr]_{p,q}
=A_0+\Delta
\end{equation}
has noiseless part and perturbation
\[
A_0
=\mathbf{GH}_{f}(\mathcal A,\mathcal B)
=\mathcal V_{\mathcal A}(\mathcal X)
\operatorname{diag}(a_1,\ldots,a_n)
\mathcal V_{\mathcal B}(\mathcal X)^\top,
\qquad
\Delta:=\bigl[\mathbf W(\boldsymbol\alpha_p+\boldsymbol\beta_q)\bigr]_{p,q},
\]
where $f(\boldsymbol\omega)=\sum_{j=1}^{n} a_j e^{i\boldsymbol\omega\cdot\mathbf y_j}$ and $\mathcal X=\{\mathbf y_1,\ldots,\mathbf y_n\}$. Thus $A_0$ is a  GHM and admits the generalized Vandermonde decomposition of Definition~\ref{def:generalized-vandermonde-decomposition}, while $A$ is its perturbation by the noise matrix $\Delta$. The case of GTM is analogous.

\section{Number detection in multi-dimensional super-resolution problems}\label{sec:number-detection}

The computational resolution limit (CRL) provides a framework for quantifying the resolution of computational super-resolution under noisy measurements. Sharp estimates are available for several one-dimensional problems \cite{liu2022mathematicaloned,liu2021theory}, whereas multi-dimensional bounds generally contain additional dimension-dependent factors. In this section, we use a generalized Hankel-based number detection algorithm to derive a sharp estimate on the multi-dimensional upper bound. 

\subsection{Computational resolution limit for number detection }\label{subsec:crl-number}


We consider the imaging model (\ref{equ:modelsetting1}) in the multi-dimensional space $d \ge 1$ with the point sources tightly spaced and form a cluster. To be more specific, for $p \in [1,\infty]$, we define the $\delta$-neighborhood in $\mathbb{R}^d$ by
\[
B_{\delta,p}^d(\mathbf{x}) := \{\mathbf{y} \in \mathbb{R}^d : \|\mathbf{y}-\mathbf{x}\|_p < \delta\}.
\]
For cubes, we write
\[
Q_\delta^d(\mathbf{x}) = B_{\delta,\infty}^d(\mathbf{x}) := \{\mathbf{y} \in \mathbb{R}^d : \|\mathbf{y}-\mathbf{x}\|_\infty < \delta\}.
\]
We assume that the clustered sources satisfy $\mathbf{y}_j \in B_{\frac{\pi n}{\Omega},1}^d(\mathbf{0})$ for $j=1, \ldots, n$. Confining the sources within a local neighborhood of scale $\mathcal{O}(n/\Omega)$ is a standard framework for studying the super-resolution of clustered sources, and it is essential for the subsequent theoretical analysis. Since we are interested in resolving closely-spaced sources, it is also highly reasonable. We remark that our results for sources in $B_{\frac{\pi n}{\Omega},1}^d(\mathbf{0})$ can be directly generalized to sources in $B_{\frac{\pi n}{\Omega},1}^d(\mathbf{x}), \mathbf{x} \in \mathbb{R}^d$.

The reconstruction process usually targets specific solutions in a so-called $\sigma$-admissible set \cite{liu2021theory}, which comprises discrete measures whose Fourier data are sufficiently close to the measurement $\mathbf{Y}$ in (\ref{equ:modelsetting1}). 


\begin{defi}
Given the measurement $\mathbf{Y}$ generated by the model (\ref{equ:modelsetting1}), let
\[
\widehat{\mu}=\sum_{j=1}^{k} \widehat{a}_j \delta_{\widehat{\mathbf{y}}_j}
\]
be written in reduced form, where $k\in\mathbb Z_{\ge0}$,
$\widehat{\mathbf{y}}_j \in \mathbb{R}^d$, and
$\widehat a_j\in\mathbb C\setminus\{0\}$ for $j=1,\ldots,k$. We say that $\widehat{\mu}$ is a
$\sigma$-admissible discrete measure of $\mathbf{Y}$ if
$$
\bigl|\mathcal{F}[\widehat{\mu}](\boldsymbol{\omega})-\mathbf{Y}(\boldsymbol{\omega})\bigr| < \sigma, \qquad \forall\,\boldsymbol{\omega} \in [-\Omega, \Omega]^d.
$$
If further $\widehat a_j\in\mathbb R_{>0}$ for $j=1,\ldots,k$, then $\widehat{\mu}$ is said to be a positive $\sigma$-admissible discrete measure of $\mathbf{Y}$.
\end{defi}

Note that the set of $\sigma$-admissible measures of $\mathbf{Y}$ characterizes all possible solutions to our super-resolution problem with the given measurement. Detecting the correct source number $n$ is possible only if all admissible measures have at least $n$ supports; otherwise, it is impossible to detect the true number without additional a priori information. Therefore, following definitions similar to those in \cite{liu2021mathematicalhighd,liu2021theory}, we define the CRL for the multi-dimensional number detection problem as follows.

\begin{defi}\label{def:crl-number}
Fix $d,n\in\mathbb N$ with $n\ge2$, $\Omega>0$, $\sigma>0$, and $m_{\min}>0$. The computational resolution limit (CRL) for number detection,
\(
\mathscr{D}_{\mathrm{num},p}(d,n),
\)
is the smallest nonnegative number such that for all sparse measures $\mu=\sum_{j=1}^n a_j\delta_{\mathbf y_j}, \mathbf y_j\in B_{\frac{\pi n}{\Omega},1}^d(\mathbf0)$ and the associated measurement $\vect Y$ in (\ref{equ:modelsetting1}), if 
\[
\min_{i\neq j}\|\vect y_i -\vect y_j\|_p\geq \mathscr{D}_{\mathrm{num},p}(d,n),
\]
then every $\sigma$-admissible discrete measure of $\mathbf Y$ has at least $n$ supports. Restricting both the source class and the admissible measures to positive amplitudes defines
$\mathscr{D}^{+}_{\mathrm{num},p}(d,n)$.
\end{defi}

This definition of the CRL emphasizes the essential impossibility of correctly detecting the number of very close sources by any means. It depends crucially on the signal-to-noise ratio and the sparsity of the sources, structurally differentiating it from classical resolution limits that rely solely on the cutoff frequency.

The following result gives an explicit upper bound for this CRL.

\begin{restatable}{thm}{liresolution}\label{thm:li-resolution}
Let \(\mathbf Y\) be a measurement generated by \(\mu=\sum_{j=1}^n a_j \delta_{\mathbf y_j}\) which is supported on \(B_{\frac{\pi n}{\Omega},1}^d(\mathbf{0})\), and write $m_{\min}:=\min_j|a_j|$. Let \(n\ge 2\), assume $0<\sigma<m_{\min}$, and suppose that the following separation condition is satisfied:
\begin{equation}\label{equ:sepa_condi_1}
\min_{j\neq p}\|\mathbf y_j-\mathbf y_p\|_{1}>
\frac{4\sqrt2\,n\pi}{\Omega}\left(\frac{3}{\sqrt{5}}\right)^d
\left(\frac{\sigma}{m_{\min}}\right)^{\frac{1}{2n-2}}.
\end{equation}
Then there does not exist any \(\sigma\)-admissible discrete measure of \(\mathbf Y\) with fewer than \(n\) supports.
\end{restatable}





\begin{proof}
Apply the GHM threshold bound in Theorem \ref{liuthm5.1v2} with $s=4n$. Its sufficient separation is
\[
\frac{2\sqrt2\,n\pi}{\Omega}
\left(\frac{2n(4n+1)^d}{5^d}\right)^{\frac{1}{2n-2}}
\left(\frac{\sigma}{m_{\min}}\right)^{\frac{1}{2n-2}}.
\]
For $n\ge2$, the elementary estimates $(2n)^{\frac{1}{2n-2}}\le2,
\left(\frac{4n+1}{5}\right)^{\frac{1}{2n-2}}
\le\frac{3}{\sqrt5}$ follow by induction, with equality at $n=2$. Hence $\left(\frac{2n(4n+1)^d}{5^d}\right)^{\frac{1}{2n-2}}
\le
2\left(\frac{3}{\sqrt5}\right)^d$, and
\[
\frac{2\sqrt2\,n\pi}{\Omega}
\left(\frac{2n(4n+1)^d}{5^d}\frac{\sigma}{m_{\min}}\right)^{\frac{1}{2n-2}}
\le
\frac{4\sqrt2\,n\pi}{\Omega}\left(\frac{3}{\sqrt5}\right)^d
\left(\frac{\sigma}{m_{\min}}\right)^{\frac{1}{2n-2}}.
\]
Therefore, Theorem~\ref{liuthm5.1v2} and \eqref{equ:sepa_condi_1} give
\[
\widehat\sigma_n\bigl(\mathbf {GH}(s)\bigr)>(s+1)^d\sigma.
\]
Suppose, to the contrary, that
\(\widehat\mu=\sum_{j=1}^{k}\widehat a_j\delta_{\widehat{\mathbf y}_j}\)
is a $\sigma$-admissible measure with $k<n$. Using the same index ordering as in \eqref{hankel}, define
\[
\widehat {\mathbf{GH}}_0(s)
:=
\left[
\mathcal F[\widehat\mu]\!
\left(\tfrac{\Omega}{s}
(\boldsymbol\alpha^{(p)}+\boldsymbol\alpha^{(q)}-s\boldsymbol 1)\right)
\right]_{p,q}.
\]
This matrix factors through the $k$ nodes of $\widehat\mu$, so
\(\operatorname{rank}(\widehat {\mathbf{GH}}_0(s))\le k<n\).
Let \(\widehat\Delta:=\mathbf {GH}(s)-\widehat {\mathbf{GH}}_0(s)\). By admissibility, every entry of $\widehat\Delta$ has modulus strictly smaller than $\sigma$, and hence
\[
\|\widehat\Delta\|_2
\le \|\widehat\Delta\|_F
<(s+1)^d\sigma.
\]
Weyl's inequality now gives
\[
\widehat\sigma_n\bigl(\mathbf {GH}(s)\bigr)
\le \sigma_n(\widehat {\mathbf{GH}}_0(s))+\|\widehat\Delta\|_2
<(s+1)^d\sigma,
\]
which is a contradiction.
\end{proof}

Consequently, for $n\ge2$ and $0<\sigma<m_{\min}$, Definition~\ref{def:crl-number} and Theorem~\ref{thm:li-resolution} give
\begin{equation}\label{eq:crl-number-upper}
\mathscr D_{\mathrm{num},1}(d,n)
\le
\frac{4\sqrt2\,n\pi}{\Omega}\left(\frac{3}{\sqrt5}\right)^d
\left(\frac{\sigma}{m_{\min}}\right)^{\frac{1}{2n-2}}.
\end{equation}

It is instructive to compare \eqref{eq:crl-number-upper} with the existing estimates in the literature, which cover the two-dimensional case and the general-dimensional case separately. For the two-dimensional case, Liu and Ammari \cite{liu2024improved} proved, in the notation of the present paper, the upper bound
\[
\mathscr D_{\mathrm{num},1}(2,n)
\le
\frac{16.6\,\pi(n-1)}{\Omega}
\left(\frac{\sigma}{m_{\min}}\right)^{\frac{1}{2n-2}},
\]
under the source assumption $\mathbf y_j\in B_{\frac{(n-1)\pi}{6\Omega},\infty}^2(\mathbf0)$, a domain which for $d=2$ is contained in our source ball $B_{\frac{\pi n}{\Omega},1}^2(\mathbf0)$. Specializing \eqref{eq:crl-number-upper} to the two-dimensional case $d=2$, the constant in our bound equals $4\sqrt2\,(3/\sqrt5)^2n=\frac{36\sqrt2}{5}n\approx10.18\,n$, which is smaller than $16.6(n-1)$ for every $n\ge3$.

For a general dimension, through one-dimensional projection trick Liu and Zhang \cite{liu2021mathematicalhighd} characterized the CRL $\mathcal D_{\mathrm{num}, 2}$ for number detection in $d$ dimensions, with sources confined to the $\ell^2$-ball $B_{\frac{(n-1)\pi}{2\Omega},2}^d(\mathbf0)$, by the upper-bound estimate
\[
\mathcal D_{\mathrm{num}, 2}(d, n)
\le
\frac{4.4\pi e\,(\pi/2)^{d-1}\bigl(n(n-1)/\pi\bigr)^{\xi(d-1)}}{\Omega}
\left(\frac{\sigma}{m_{\min}}\right)^{\frac{1}{2n-2}},
\]
where $\xi(d):=\sum_{j=1}^{d}1/j$ for \(d \ge 1\) and \(\xi(0) := 0\). Compared with this estimate, \eqref{eq:crl-number-upper} removes the factor $(n(n-1))^{\xi(d-1)}$, whose exponent of $n$ grows with the dimension, and replaces it by a linear dependence on $n$; moreover, the dimension factor in \eqref{eq:crl-number-upper} is $(3/\sqrt5)^d$, which is smaller than $(\pi/2)^{d-1}$ for every $d\ge3$. This substantial improvement stems from the fact that the analysis in
\cite{liu2021mathematicalhighd} reduces the multi-dimensional problem to a
collection of one-dimensional problems via projection, which inevitably
introduces a loss of source separation and hence deteriorates the resulting
conditioning estimates. In contrast, our analysis works directly with the
multi-variate Vandermonde matrix and therefore exploits the full
multi-dimensional geometry of the source configuration and the information
contained in the measurements.




On the other hand, combining \eqref{eq:crl-number-upper} with the lower bound
established in \cite{liu2021mathematicalhighd}, we obtain the following
two-sided estimate:
\begin{equation}\label{eq:crl-number-twosided}
\frac{2}{\Omega}
\left(\frac{\sigma}{m_{\min}}\right)^{\frac{1}{2n-2}}
\le
\mathscr D_{\mathrm{num},1}(d,n)
\le
\frac{4\sqrt2\,n\pi}{\Omega}
\left(\frac{3}{\sqrt5}\right)^d
\left(\frac{\sigma}{m_{\min}}\right)^{\frac{1}{2n-2}},
\end{equation}
for both the general and positive super-resolution problems.
In particular, the dependence on the noise-to-signal ratio,
\[
\left(\frac{\sigma}{m_{\min}}\right)^{\frac{1}{2n-2}},
\]
is sharp.

Our analysis further suggests that the upper bound can potentially be improved
to the form
\begin{equation}\label{eq:crl-number-improved}
\mathscr D_{\mathrm{num},1}(d,n)
\le
\frac{c(d)}{\Omega}
\left(\frac{\sigma}{m_{\min}}\right)^{\frac{1}{2n-2}},
\end{equation}
where $c(d)$ depends only on the ambient dimension and, in particular, the
additional linear dependence on $n$ in
\eqref{eq:crl-number-twosided} can be removed.
Such an estimate may be viewed as a noise-dependent generalization of the
classical Rayleigh limit, whose characteristic scale is
$c(d)/\Omega$.

The underlying reason is geometric. The worst-case exponent $n-1$ in the
smallest singular value estimate of multivariate Vandermonde matrix is attained when the sources become
essentially collinear \cite{li2025nonharmonic}, so that the high-dimensional configuration degenerates
to the most ill-conditioned one-dimensional geometry. Once the sources are
ordered along this line, like the one-dimensional super-resolution \cite{liu2021theory}, the $n$ can be removed by considering the source order.

\subsection{Number detection algorithms based on singular value thresholding}\label{subsec:uniform array}

This section details the number detection algorithm underpinning the sharp upper bound derived previously. It is a multi-dimensional version of singular-value-thresholding algorithms in \cite{liu2024improved} and enables us to obtain the desired CRL constant. 


Let \(s\ge1\) be an integer, set \(\Lambda:=\{0,1,\dots,s\}\subset \mathbb Z\), and let \(\Lambda^{d}:=\Lambda\times\cdots\times\Lambda\subset \mathbb Z^{d}\). We assume that the measurement is given on the set
\[
\Gamma_s:=\frac{\Omega}{s}\bigl(\Lambda^{d}+\Lambda^{d}-s\boldsymbol 1\bigr)
=
[-\Omega,\Omega]^d\cap \frac{\Omega}{s}\mathbb Z^d
\]
by
\[
\mathbf Y(\boldsymbol\omega)=\mathcal F[\mu](\boldsymbol\omega)+\mathbf W(\boldsymbol\omega)
=\sum_{j=1}^{n} a_j e^{i\boldsymbol\omega\cdot \mathbf y_j}+\mathbf W(\boldsymbol\omega),
\qquad \boldsymbol\omega\in \Gamma_s,
\]
where  \(\boldsymbol 1:=(1,\dots,1)^\top\in\mathbb R^d\), \(\mu=\sum_{j=1}^{n} a_j\delta_{\mathbf y_j}\) and \(\|\mathbf W\|_{\infty}<\sigma\).

The construction queries exactly the values indexed by $\Gamma_s$; hence it uses $(2s+1)^d$ distinct Fourier samples. We fix the lexicographic ordering on \(\Lambda^{d}\) and write
\[
\Lambda^{d}=\{\boldsymbol\alpha^{(1)},\dots,\boldsymbol\alpha^{((s+1)^d)}\}.
\]
We then define the GHM \(\mathbf{GH}(s)\in\mathbb C^{(s+1)^d\times (s+1)^d}\) by
\begin{equation}\label{hankel}
\mathbf{GH}(s):=\bigl[\mathbf Y\bigl(\tfrac{\Omega}{s}(\boldsymbol\alpha^{(p)}+\boldsymbol\alpha^{(q)}-s\boldsymbol 1)\bigr)\bigr]_{1\le p,q\le (s+1)^d},
\end{equation}
and the corresponding noise matrix by
\begin{equation}\label{delta}
\mathbf\Delta:=\bigl[\mathbf W\bigl(\tfrac{\Omega}{s}(\boldsymbol\alpha^{(p)}+\boldsymbol\alpha^{(q)}-s\boldsymbol 1)\bigr)\bigr]_{1\le p,q\le (s+1)^d}.
\end{equation}

Next, define \(\psi_s(t):=(1,t,\dots,t^s)^\top\in\mathbb C^{s+1}\). For each node \(\mathbf y_j=((\mathbf y_j)_1,\dots,(\mathbf y_j)_d)^\top \in\mathbb R^d\), let
\[
(\mathbf V_1)_j:=\psi_s\!\left(e^{i(\mathbf y_j)_1\Omega/s}\right)\otimes\cdots\otimes \psi_s\!\left(e^{i(\mathbf y_j)_d\Omega/s}\right),
\]
and set
\[
\mathbf V_1=[(\mathbf V_1)_1,\dots,(\mathbf V_1)_n],\qquad
\mathbf\Sigma:=\operatorname{diag}\!\left(a_1 e^{-i\Omega\boldsymbol 1^\top\mathbf y_1},\dots,a_n e^{-i\Omega\boldsymbol 1^\top\mathbf y_n}\right).
\]
With the above ordering of \(\Lambda^{d}\), the \(\ell\)-th entry of \((\mathbf V_1)_j\) is \(e^{\,i(\Omega/s)\boldsymbol\alpha^{(\ell)}\cdot \mathbf y_j}\). The noiseless matrix
\[
\mathbf {GH}_0(s):=\mathbf V_1\mathbf\Sigma\mathbf V_1^\top
\]
admits a generalized Vandermonde decomposition in the sense of Definition~\ref{def:generalized-vandermonde-decomposition}. The observed matrix is its perturbation:
\begin{equation}\label{eq:factorization_H}
\mathbf {GH}(s)=\mathbf {GH}_0(s)+\mathbf\Delta
=\mathbf V_1\mathbf\Sigma \mathbf V_1^\top+\mathbf\Delta.
\end{equation}


To analyze the singular values of $\mathbf {GH}(s)$, we first need a lower bound on the minimum singular value of the contiguous Vandermonde matrix $\mathbf V_1$. The following lemma provides this crucial estimate.

\begin{lem}\label{lem:uniform-Vandermonde}
Let $d\geq 1$, $n\ge2$, and let $\mathbf y_j\in B_{\frac{\pi n}{\Omega},1}^d(\mathbf{0})$ for $j=1,\dots,n$. If $s\geq 4n$ is even, then
\begin{equation}\label{eq:uniform-Vandermonde}
\sigma_{\min}(\mathbf V_1)\geq
\sqrt{\frac{1}{n2^{n-1}}\left(2\left\lfloor\frac{s}{2n}\right\rfloor+1\right)^d}\,
\bigl(\theta_{\min}(\Omega,n)\bigr)^{n-1},
\end{equation}
where
\[
\theta_{\min}(\Omega,n):=\min_{j\neq k}\frac{\Omega}{2n\pi}\|\mathbf y_j-\mathbf y_k\|_1.
\]
\end{lem}

\begin{proof}
The detailed proof is deferred to Appendix~\ref{sec:appendix-li}.
\end{proof}

With the above lemma and notation in hand, we can now formulate the multi-dimensional singular value threshold theorem.

\begin{restatable}{thm}{uniform-Vandermonde-threshold}\label{liuthm5.1v2}
Let $d\geq 1$ and $n\ge2$, let $\mu=\sum_{j=1}^{n}a_j\delta_{\mathbf y_j}$ with $\mathbf y_j\in B_{\frac{\pi n}{\Omega},1}^d(\mathbf{0})$ for $j=1,\dots,n$, and assume that $s\geq 4n$ is even and $\sigma<m_{\min}$. Let $\hat{\sigma}_1\geq \hat{\sigma}_2\geq \cdots \geq \hat{\sigma}_{(s+1)^d}$ be the singular values of $\mathbf {GH}(s)$. Then
\begin{equation}\label{liueq5.3v2}
\hat{\sigma}_j\leq (s+1)^d\sigma,\qquad j=n+1,\dots,(s+1)^d.
\end{equation}
Moreover, if
\begin{equation}\label{liueq5.4v2}
\min_{j\neq k}\|\mathbf y_j-\mathbf y_k\|_1>
\frac{2n\pi}{\Omega}
\left(
\frac{n2^n}{m_{\min}\left(2\left\lfloor\frac{s}{2n}\right\rfloor+1\right)^d}(s+1)^d\sigma
\right)^{\frac{1}{2n-2}},
\end{equation}
then
\begin{equation}\label{liueq5.5v2}
\hat{\sigma}_n>(s+1)^d\sigma. 
\end{equation}
\end{restatable}

\begin{proof}
By \eqref{eq:factorization_H}, we have $\mathbf {GH}(s)=\mathbf {GH}_0(s)+\mathbf\Delta$ with $\mathbf {GH}_0(s)=\mathbf V_1\mathbf\Sigma \mathbf V_1^\top$. Since each entry of $\mathbf\Delta$ is bounded by $\sigma$ in modulus, $\|\mathbf\Delta\|_2\leq \|\mathbf\Delta\|_F\leq (s+1)^d\sigma$. Weyl's theorem therefore gives
\[
|\hat{\sigma}_j-\sigma_j(\mathbf {GH}_0(s))|\leq \|\mathbf\Delta\|_2,\qquad j=1,\dots,(s+1)^d.
\]
By decomposition $\mathbf {GH}_0(s)=\mathbf V_1\mathbf\Sigma \mathbf V_1^\top$, we have $\operatorname{rank}(\mathbf {GH}_0(s))\leq n$ and $\sigma_j(\mathbf {GH}_0)=0$ for $j\geq n+1$. Together with $\|\mathbf\Delta\|_2\leqs (s+1)^d \sigma$, this gives \eqref{liueq5.3v2}.

Next, \eqref{liueq5.4v2} implies
\[
\theta_{\min}(\Omega,n)>
\left(
\frac{n2^n}{m_{\min}\left(2\left\lfloor\frac{s}{2n}\right\rfloor+1\right)^d}(s+1)^d\sigma
\right)^{\frac{1}{2n-2}}.
\]
Applying it and (\ref{eq:uniform-Vandermonde}) to $\sigma_n(\mathbf {GH}_0(s))\geq \sigma_{\min}(\mathbf\Sigma)\sigma_{\min}(\mathbf V_1)^2$ yields \[
\sigma_n(\mathbf {GH}_0(s))>2(s+1)^d\sigma.
\]
Therefore,
\[
\hat{\sigma}_n\geq \sigma_n(\mathbf {GH}_0(s))-\|\mathbf\Delta\|_2>(s+1)^d\sigma,
\]
which proves \eqref{liueq5.5v2}.
\end{proof}

Based on Theorem \ref{liuthm5.1v2}, we can propose a simple thresholding algorithm, \textbf{Algorithm \ref{alg:fixed-svt}}, for the number detection.



\begin{algorithm}[H]
\caption{Singular-Value-Thresholding Number Detection Algorithm}\label{alg:fixed-svt}
\small
\KwIn{Measurement data \(\mathcal Y\), positive integer \(s\), dimension \(d\), noise level \(\sigma\).}
\KwOut{The detected number of signals \(\hat n_s\).}

Form the measurement matrix \(\mathbf {GH}(s)\) from \(\mathcal Y\)\;
\(\varepsilon_s \gets \sigma (s+1)^d\)\;
\(\hat\sigma_1\ge \hat\sigma_2\ge \cdots \ge \hat\sigma_{(s+1)^d}\gets \textbf{svd}(\mathbf {GH}(s))\)\;

\For{$i=(s+1)^d$ \textbf{downto} $1$}{
    \If{$\hat\sigma_i>\varepsilon_s$}{
        \Return \(\hat n_s=i\)\;
    }
}
\Return \(\hat n_s=0\)\;
\end{algorithm}





\section{Nonuniform arrays and faster number detection algorithms}\label{sec:nonuniform-arrays}

While Subsection~\ref{subsec:uniform array} addresses number detection for contiguous uniform frequency sets, this section considers non-contiguous frequency queries. Their segmented index geometry is analogous to that of sparse and distributed arrays \cite{pal2010nested,vaidyanathan2011sparsesamplers,heimiller1983distributed,1982Distributed,gavish1993viaesprit,Gershman1997A}: each coordinate set consists of uniform blocks separated by gaps. The larger maximum queried frequency can improve the conditioning of the associated Vandermonde factors, while the consecutive indices within each block permit explicit estimates.

To analyze non-contiguous sampling sets, we employ the GHM framework, which directly accommodates missing frequency indices. For segmented sets, this formulation yields explicit singular-value lower bounds without virtual-array interpolation \cite{li_preestimation-based_2009} or multiple-invariance ESPRIT \cite{swindlehurst_multiple_1992}. We also consider a randomized construction whose matrix size can be chosen below that of a full tensor grid; its recovery statement is conditional on the realized Vandermonde factors.

The remainder of this section is organized as follows. In Subsection \ref{subsec:equally_distributed}, we propose a detection algorithm for segmented frequency sets. In Subsection \ref{subsec:rand-numdetect}, we introduce a randomized sampling heuristic and state its deterministic conditional guarantee. Finally, numerical results are presented for the proposed algorithms.

\subsection{Fast number detection algorithms for equally distributed arrays}\label{subsec:equally_distributed}

We now consider a segmented frequency-query pattern whose blocks have the same lengths and uniform spacings. Let $\mathcal{X} = \{\mathbf{y}_1, \ldots, \mathbf{y}_{n}\} \subset (-\pi, \pi]^d$ be an ordered set of distinct nodes. Let $m,r\in\mathbb Z_{\ge0}$ and $D\in\mathbb N$ satisfy $m\ge1$ and $D\ge m+1$. We define the segmented index set $\Lambda$ as
\[
\Lambda = \{0,1,\dots,m\} \cup \{D,D+1,\dots,D+m\} \cup \cdots \cup \{rD,rD+1,\dots,rD+m\}.
\]
Let $L = |\Lambda| = (r+1)(m+1)$ and denote the maximum index by $\Omega := rD+m$. The multi-dimensional sampling set is the Cartesian product $\Lambda^d = \Lambda \times \dots \times \Lambda \subset \mathbb{Z}^d$ and the measurement is given on the set $\Gamma_{\mathrm{seg}} := \Lambda^d + \Lambda^d-\Omega\boldsymbol 1\subset[-\Omega,\Omega]^d\cap\mathbb Z^d$, i.e., 
\[
\mathbf Y(\boldsymbol\omega)=\mathcal F[\mu](\boldsymbol\omega)+\mathbf W(\boldsymbol\omega)
=\sum_{j=1}^{n} a_j e^{i\boldsymbol\omega\cdot \mathbf y_j}+\mathbf W(\boldsymbol\omega),
\qquad \boldsymbol\omega\in \Gamma_{\mathrm{seg}},
\]
where $\boldsymbol 1:=(1,\dots,1)^\top\in\mathbb R^d$, $\mu=\sum_{j=1}^{n} a_j\delta_{\mathbf y_j}$ and $\|\mathbf W\|_{\infty}<\sigma$. 

Figure~\ref{fig:GHMmeasure} illustrates the resulting one-dimensional measurement set.

\begin{figure}[htbp]
    \centering
    \includegraphics[width=\linewidth]{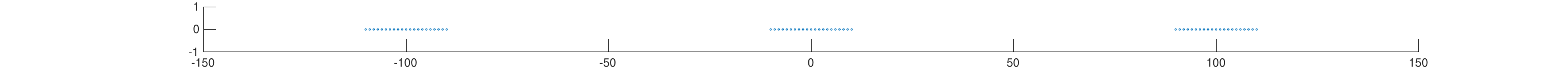}
    \caption{The measurement set of the segmented GHM with $r=1$, $m=10$, unit inner spacing, and block spacing $D=100$.}
    \label{fig:GHMmeasure}
\end{figure}

We fix the lexicographic ordering on $\Lambda^d$ and write
\[
\Lambda^d=\{\boldsymbol\alpha^{(1)},\dots,\boldsymbol\alpha^{(L^d)}\}.
\]
We then define the GHM $\mathbf {GH}(m,r,D)\in\mathbb C^{L^d\times L^d}$ by
\begin{equation}\label{eq:segmented_hankel}
\mathbf {GH}(m,r,D):=\bigl[\mathbf Y\bigl(\boldsymbol\alpha^{(p)}+\boldsymbol\alpha^{(q)}-\Omega\boldsymbol 1\bigr)\bigr]_{1\le p,q\le L^d},
\end{equation}
and the corresponding noise matrix by
\begin{equation}
\mathbf\Delta:=\bigl[\mathbf W\bigl(\boldsymbol\alpha^{(p)}+\boldsymbol\alpha^{(q)}-\Omega\boldsymbol 1\bigr)\bigr]_{1\le p,q\le L^d}.
\end{equation}

Next, define $\psi_\Lambda(t) := (t^k)_{k \in \Lambda}^\top \in \mathbb{C}^L$. For each node $\mathbf y_j=((\mathbf y_j)_1,\dots,(\mathbf y_j)_d)\in(-\pi, \pi]^d$, let
\[
(\vect V_1)_j:=\psi_\Lambda\!\left(e^{i(\mathbf y_j)_1}\right)\otimes\cdots\otimes \psi_\Lambda\!\left(e^{i(\mathbf y_j)_d}\right),
\]
and set
\begin{equation}\label{equ:segmented_vander_1}
\mathbf V_1=[(\mathbf V_1)_1,\dots,(\mathbf V_1)_n],\qquad
\mathbf\Sigma:=\operatorname{diag}\!\left(a_1 e^{-i\Omega\boldsymbol 1^\top\mathbf y_1},\dots,a_n e^{-i\Omega\boldsymbol 1^\top\mathbf y_n}\right).
\end{equation}
With the above ordering of $\Lambda^d$, the $\ell$-th entry of $(\mathbf V_1)_j$  is $e^{\,i\boldsymbol\alpha^{(\ell)}\cdot \mathbf y_j}$, respectively. The noiseless matrix
\[
\mathbf {GH}_0(m,r,D):=\mathbf V_1\mathbf\Sigma\mathbf V_1^\top
\]
admits a generalized Vandermonde decomposition in the sense of Definition \ref{def:generalized-vandermonde-decomposition}. The observed matrix is its perturbation:
\begin{equation}\label{eq:factorization_H_segmented}
\mathbf {GH}(m,r,D)=\mathbf {GH}_0(m,r,D)+\mathbf\Delta
=\mathbf V_1\mathbf\Sigma \mathbf V_1^\top+\mathbf\Delta.
\end{equation}

Based on the singular-value thresholding criterion in Theorem~\ref{thm:segmented_threshold}, we first propose a fixed-\((m,r)\) number detection procedure for the equally separated array setting.

\begin{algorithm}[H]
\caption{Segmented Singular-Value-Thresholding Number Detection Algorithm}\label{alg:fixed-segmented-svt}
\small
\KwIn{Measurements \(\{\mathbf Y(\boldsymbol\omega):\boldsymbol\omega\in\Gamma_{\mathrm{seg}}\}\), array parameters \(m, r, D\), dimension \(d\), noise level \(\sigma\).}
\KwOut{The detected number of signals \(\hat n_{m,r}\).}

\(L \gets (r+1)(m+1)\)\;
Form the segmented measurement matrix \(\mathbf {GH}(m,r,D)\) of size \(L^d \times L^d\) from these measurements\;
\(\varepsilon_{m,r} \gets \sigma L^d\)\;
\(\hat\sigma_1\ge \hat\sigma_2\ge \cdots \ge \hat\sigma_{L^d}\gets \textbf{svd}(\mathbf {GH}(m,r,D))\)\;

\For{$i=L^d$ \textbf{downto} $1$}{
    \If{$\hat\sigma_i>\varepsilon_{m,r}$}{
        \Return \(\hat n_{m,r}=i\)\;
    }
}
\Return \(\hat n_{m,r}=0\)\;
\end{algorithm}


To analyze the resolution of algorithm \ref{alg:fixed-segmented-svt}, we consider a more practical scenario where the nodes form multiple clusters in the multi-dimensional space. To rigorously characterize such a distribution, we first define the distance metric on the angular torus $\mathbb T_{2\pi}^d:= \mathbb R^d/(2\pi\mathbb Z)^d$ and the corresponding notions of local sparsity.

\begin{defi}[Periodic Metric and Minimum Separation]
    \label{defi:metric_separation}
    For any two representatives $\mathbf{u}, \mathbf{v} \in (-\pi, \pi]^d$ of points in  $\mathbb T_{2\pi}^d$ and $p \in [1, \infty)$, the periodic $\ell_p$-distance is defined as
    \[
    |\mathbf{u} - \mathbf{v}|_p := \left( \sum_{k=1}^d \min \big( |u_k - v_k|, \, 2\pi - |u_k - v_k| \big)^p \right)^{1/p},
    \]
    with the $\ell^\infty$-distance defined as $|\mathbf{u} - \mathbf{v}|_\infty := \max_{1 \le k \le d} \min \big( |u_k - v_k|, \, 2\pi - |u_k - v_k| \big)$. For a finite set of nodes $\mathcal{X} \subset (-\pi, \pi]^d$, the global minimum $\ell_p$-separation is defined as
    \[
    \Delta_p(\mathcal{X}) := \min_{\substack{\mathbf{x}, \mathbf{y} \in \mathcal{X} \\ \mathbf{x} \ne \mathbf{y}}} |\mathbf{x} - \mathbf{y}|_p.
    \]
    For a singleton set, we use the convention $\Delta_p(\mathcal X):=+\infty$.
\end{defi}

\begin{defi}[Local Sparsity]
    \label{defi:local_sparsity}
    Fix a scale parameter $\tau > 0$. For any node $\mathbf{y} \in \mathcal{X}$, its local neighborhood $\mathcal{N}_\infty(\mathbf{y}, \tau, \mathcal{X})$ is defined as 
    \[
    \mathcal{N}_\infty(\mathbf{y}, \tau, \mathcal{X}) := \left\{ \mathbf{z} \in \mathcal{X} : |\mathbf{y} - \mathbf{z}|_\infty \le \tau \right\}.
    \]
    The local sparsity $\nu_\infty(\tau, \mathcal{X})$ of the node set $\mathcal{X}$ is defined as the maximum cardinality of any such neighborhood:
    \[
    \nu_\infty(\tau, \mathcal{X}) := \max_{\mathbf{y} \in \mathcal{X}} \big| \mathcal{N}_\infty(\mathbf{y}, \tau, \mathcal{X}) \big|.
    \]
\end{defi}

In many practical scenarios, nodes are tightly concentrated within well-separated regions, naturally forming a multi-cluster geometry. To rigorously capture this phenomenon, we formally define this clumped structure.

\begin{defi}[Clump Structure]
    \label{defi:high_dim_clumps}
    Let $A,n^\star\in\mathbb N$ with $n^\star\ge2$, and let $0<\tau\le\eta$. A subset $\mathcal{X} \subset (-\pi, \pi]^d$ is said to form \textit{$(A, \infty, \tau, \eta, n^\star)$-clumps} if there exists a partition $\mathcal{X} = \bigcup_{a=1}^A \mathcal{C}_a$ into nonempty pairwise disjoint sets such that the following conditions hold:
    \begin{enumerate}
        \item $|\mathcal{C}_a| = n_a \le n^\star$ for all $a=1, \dots, A$, and $n^\star = \max_{1 \le a \le A} n_a$.
        \item $\max_{\mathbf{x}, \mathbf{y} \in \mathcal{C}_a} |\mathbf{x} - \mathbf{y}|_\infty \le \tau$ for each cluster $\mathcal{C}_a$.
        \item $\min_{\mathbf{x} \in \mathcal{C}_a, \mathbf{y} \in \mathcal{C}_b} |\mathbf{x} - \mathbf{y}|_\infty > \eta$ for all distinct indices $a \neq b$.
    \end{enumerate}
    Consequently, the local sparsity of such a set exactly satisfies $\nu_\infty(\tau, \mathcal{X}) = n^\star$.
\end{defi}

For clumped nodes, the minimum singular value of the segmented Vandermonde matrix admits the following explicit lower bound.

\begin{thm}
    \label{thm:segmented-vandermonde}
    Let $d \ge 1$, and suppose $\mathcal{X}=\{\mathbf{y}_1, \dots, \mathbf{y}_n\} \subset (-\pi, \pi]^d$ consists of $(A, \infty, \tau,  \eta, n^\star)$-clumps. 
    Let $\Lambda$ be the segmented index set with integer parameters $m\ge1$, $r\ge0$, and $D > m$.
    Fix $0 <\tau \le \frac{\pi}{2Dd}$ and $\beta > \frac{1}{2\log 2}$. Define
    \[
        m_1:=\left\lceil\frac{m}{2}\right\rceil,
        \qquad
        K_{\mathrm{loc}}:=\left\lfloor\frac{m_1}{n^\star}\right\rfloor.
    \]
    Suppose $\eta \geq \frac{4\pi \beta d}{K_{\mathrm{loc}}+1}$, $r \ge 2n^\star$, and the minimum $\ell_1$-separation satisfies $\Delta_1(\mathcal{X})\leq \frac{\pi n^\star}{rD}$. Then the smallest singular values of the segmented Vandermonde matrix $\mathbf{V}_1$ satisfies the following lower bound:
    \begin{align}
        \sigma_{\min}(\mathbf{V}_1) 
        \ge \frac{1}{\sqrt{n}} \left( 2 - e^{1/(2\beta)} \right)^{n^\star/2}\frac{\sqrt{(\frac{r}{n^\star})^d(\lfloor\frac{m}{2} \rfloor+1)^d} }{(\sqrt{2})^{n^\star-1}}  \left( \frac{rD}{\pi n^\star } \Delta_1(\mathcal{X})\right)^{n^\star-1}.
    \end{align}
\end{thm}

\begin{proof}
     By the definition of $\mathbf{V}_1$ in \eqref{equ:segmented_vander_1}, it coincides with the generalized Vandermonde matrix $\mathcal{V}_{\Lambda^d}(\mathcal{X})$ in Definition~\ref{generalized vandermonde}, hence $\sigma_{\min}(\mathbf{V}_1)=\sigma_{\min}(\mathcal{V}_{\Lambda^d}(\mathcal{X}))$. The detailed constructive proof of this lower bound for $\mathcal{V}_{\Lambda^d}(\mathcal{X})$ is deferred to Appendix~\ref{sec:appendix-segmented-Vandermonde}.
\end{proof}

We are now ready to state the multi-dimensional singular value threshold theorem for segmented frequency sets. The dimension of $\mathbf {GH}(m,r,D)$ is $L^d \times L^d$, where $L = (r+1)(m+1)$.

\begin{restatable}{thm}{segmentedthreshold}\label{thm:segmented_threshold}
    Let $d\geq 1$ and $n^\star\ge2$, and let $\mu=\sum_{j=1}^{n}a_j\delta_{\mathbf y_j}$ with $\mathcal{X}=\{\mathbf y_1, \dots, \mathbf y_n\} \subset (-\pi, \pi]^d$ forming $(A, \infty, \tau,  \eta, n^\star)$-clumps. Set $m_{\min}:=\min_{1\le j\le n}|a_j|$.
    Let $\Lambda$ be the segmented index set with integer parameters $m\ge1$, $r\ge0$, and $D > m$, and set $\Omega:=rD+m$ and $L := (r+1)(m+1)$.
    Fix $0 <\tau \le \frac{\pi}{2Dd}$ and $\beta > \frac{1}{2\log 2}$, and define
    \[
        m_1:=\left\lceil\frac{m}{2}\right\rceil,
        \qquad
        K_{\mathrm{loc}}:=\left\lfloor\frac{m_1}{n^\star}\right\rfloor.
    \]
    Suppose $\eta \geq \frac{4\pi \beta d}{K_{\mathrm{loc}}+1}$, $r \ge 2n^\star$, and the noise level satisfies $0<\sigma < m_{\min}$.
    Let $\hat{\sigma}_1\geq \hat{\sigma}_2\geq \cdots \geq \hat{\sigma}_{L^d}$ be the singular values of the segmented measurement matrix $\mathbf {GH}(m,r,D)$. Then for the noise subspace, we have
    \begin{equation}\label{eq:threshold_noise_segmented}
        \hat{\sigma}_j\leq L^d\sigma,\qquad j=n+1,\dots,L^d.
    \end{equation}
    Moreover, if the minimum separation satisfies the local geometry condition $\Delta_1(\mathcal{X}) \leq \frac{\pi n^\star}{rD}$ and is bounded from below by the threshold
\begin{equation}\label{eq:separation_condition_segmented}
     \Delta_1(\mathcal{X}) > \frac{\sqrt{2}\pi (n^\star + \frac{1}{2}) (\sqrt{5})^d}{\Omega\left( 2 - e^{1/(2\beta)} \right)} \left( \frac{2n \sigma}{m_{\min}} \right)^{\frac{1}{2n^\star-2}},
    \end{equation}
    then the signal subspace satisfies
    \begin{equation}\label{eq:threshold_signal_segmented}
        \hat{\sigma}_n > L^d\sigma.
    \end{equation}
\end{restatable}

\begin{proof}
    Similar to the proof of Theorem~\ref{liuthm5.1v2}, applying Weyl's inequality alongside the noise perturbation bound $\|\mathbf{\Delta}\|_2 \le L^d\sigma$ immediately establishes the noise subspace bound \eqref{eq:threshold_noise_segmented}.
    
    Since $\mathbf V_1$ has full column rank, the product inequality for the smallest nonzero singular value gives
    \[
        \sigma_n(\mathbf {GH}_0)
        =\sigma_n(\mathbf V_1\mathbf\Sigma\mathbf V_1^\top)
        \ge m_{\min}\sigma_{\min}(\mathbf V_1)^2.
    \]
    Thus, for the signal subspace threshold \eqref{eq:threshold_signal_segmented}, it suffices that $m_{\min} \sigma_{\min}(\mathbf{V}_1)^2 > 2L^d\sigma$. Substituting the lower bound from Theorem~\ref{thm:segmented-vandermonde} gives
    \begin{equation}\label{eq1:separation_condition_segmented}
    \Delta_1(\mathcal{X}) >
        \frac{\sqrt{2}\pi n^\star}{rD}
        \left(
        \frac{2 n (n^\star)^d \sigma}{m_{\min} (2 - e^{1/(2\beta)})^{n^\star}} \left( \frac{L}{r(\lfloor \frac{m}{2} \rfloor + 1)} \right)^d
        \right)^{\frac{1}{2n^\star-2}}.
    \end{equation}
    Specifically, for any $m \ge 1$ and $n^\star \ge 2$, utilizing $r \ge 2n^\star$, we apply three strict worst-case bounds simultaneously:
    $$
        \left( (n^\star)^d \left( \frac{L}{r(\lfloor \frac{m}{2} \rfloor + 1)} \right)^d
        \right)^{\frac{1}{2n^\star-2}} \le (\sqrt{5})^d, \quad
        \left( \frac{1}{(2 - e^{1/(2\beta)})^{n^\star}} \right)^{\frac{1}{2n^\star-2}} \le \frac{1}{2 - e^{1/(2\beta)}}.
    $$
    Furthermore, linking $rD$ with the maximum index $\Omega = rD + m$, we strictly bound the distance parameter by:
$$\frac{n^\star}{rD} = \frac{n^\star}{\Omega} \frac{rD+m}{rD} < \frac{n^\star}{\Omega}\left(1 + \frac{1}{r}\right) \le \frac{n^\star}{\Omega}\frac{2n^\star+1}{2n^\star} = \frac{n^\star + \frac{1}{2}}{\Omega}.$$
Combining these bounds yields the sufficient condition \eqref{eq:separation_condition_segmented}.
\end{proof}

Compared to the resolution limit theory in Theorem \ref{subsec:crl-number}, Theorem \ref{thm:segmented_threshold} shows that Algorithm \ref{alg:fixed-segmented-svt} achieves the optimal resolution when superresolving point sources in well-separated clumps.

\subsection{Number detection algorithm based on random GHM}\label{subsec:rand-numdetect}

In multi-dimensional super-resolution problems, classical subspace methods rely on multi-level Hankel or Toeplitz matrices. For a $d$-dimensional space, if we take $L$ measurements along each coordinate to construct the data matrix, the resulting multi-level matrix has a dimension of $L^d \times L^d$. The computational complexity of performing the Singular Value Decomposition (SVD) on this matrix scales as $\mathcal{O}(L^{3d})$. This exponential growth with respect to the spatial dimension $d$ is known as the curse of dimensionality, which makes classical methods computationally intractable even for moderate dimensions.

To reduce the matrix dimension, we consider a randomized sampling strategy based on the generalized Hankel framework. Instead of using a dense tensor grid, it constructs a matrix from two randomly selected lists of admissible integer frequencies.

Let $\mathcal A,\mathcal B\subset\mathbb Z^d$ be prescribed finite sets satisfying
$\mathcal A+\mathcal B\subset[-\Omega,\Omega]^d$, and let $M_1,M_2\in\mathbb N$ satisfy
$1\le M_1\le|\mathcal A|$ and $1\le M_2\le|\mathcal B|$. Draw
$\boldsymbol{\omega}_1,\dots,\boldsymbol{\omega}_{M_1}$ uniformly without replacement from $\mathcal A$ and, independently, draw
$\boldsymbol{\zeta}_1,\dots,\boldsymbol{\zeta}_{M_2}$ uniformly without replacement from $\mathcal B$. We then construct $\mathbf {GH}_{\mathrm{rand}}\in\mathbb C^{M_1\times M_2}$ by
\begin{equation}\label{eq:random_hankel}
    (\mathbf{GH}_{\text{rand}})_{p,q} = \mathbf{Y}(\boldsymbol{\omega}_p + \boldsymbol{\zeta}_q), \quad 1 \le p \le M_1, \ 1 \le q \le M_2,
\end{equation}
where $\mathbf{Y}(\boldsymbol{\omega})$ is the noisy measurement defined in \eqref{equ:modelsetting1}. 

Write $\mathbf {GH}_0:=\mathbf V_1\mathbf\Sigma\mathbf V_2^\top$. Since \eqref{eq:random_hankel} uses the noisy data, the random measurement matrix is a perturbation of the GHM $\mathbf{GH}_0$:
\begin{equation}\label{eq:random-hankel-noisy-factorization}
    \mathbf{GH}_{\text{rand}} = \mathbf {GH}_0+\Delta
    =\mathbf{V}_1 \mathbf{\Sigma} \mathbf{V}_2^\top+\Delta,
    \qquad
    \Delta_{pq}:=\mathbf W(\boldsymbol\omega_p+\boldsymbol\zeta_q),
\end{equation}
where $\mathbf{V}_1 \in \mathbb{C}^{M_1 \times n}$ and $\mathbf{V}_2 \in \mathbb{C}^{M_2 \times n}$ are generalized Vandermonde matrices evaluated at the randomly sampled frequencies, and $\mathbf{\Sigma} = \operatorname{diag}(a_1, \dots, a_n)$.

Based on the decomposition (\ref{eq:random-hankel-noisy-factorization}), we develop the singular value thresholding algorithm (Algorithm \ref{alg:rand_hankel_svt}) that is similar to Algorithm \ref{alg:fixed-segmented-svt}. The following theorem gives the corresponding threshold and the resolution of the algorithm, where the \emph{sampling spread} of order $n$,
\begin{equation}\label{eq:sampling_spread_def}
    \gamma_n(\Lambda) = \max_{S \subset \Lambda, |S|=n} \min_{\lambda,\mu \in S, \lambda \ne \mu} |\lambda - \mu|,
\end{equation} 
plays a crucial role. 

\begin{thm}\label{thm:resolutionrandghmnumber1}
Let $M_1, M_2> n$, then 
\[
\hat \sigma_j < \sigma\sqrt{M_1M_2}, \quad j=n+1, \cdots, M_2.
\]
Furthermore, for sufficiently closed one-dimensional sources and sufficiently small noise $\sigma$, if the sampling spread $\min(\gamma_n(\mathcal A), \gamma_n(\mathcal B))\geq C_2(n)\frac{\Omega}{n}$ and
\begin{equation}\label{equ:minisinguofrandomvander}
\min_{i\neq j}|y_i -y_j|\geq \frac{C_3}{\Omega}\left(\frac{\sigma}{m_{\min}}\right)^{\frac{1}{2n-2}},
    \end{equation}
we have 
\[
\hat \sigma_n > \sigma\sqrt{M_1M_2}.
\]
\end{thm}
\begin{proof}
The noiseless matrix $\mathbf {GH}_0$ has rank at most $n$. By Weyl's inequality, we have 
\[
\hat\sigma_j\leq \|\Delta\|_2\leq \sigma \sqrt{M_1M_2},
\qquad j\ge n+1.
\]

For other $\hat \sigma_j$'s, the estimate relies on the following estimate (\cite[Theorem~5.5]{huang2026minimum}) for the minimum singular value of non-uniformly sampled Fourier matrix.
\begin{thm}\label{thm:nonuniform_vdm_scaling}
Let $\Lambda = \{\lambda_1, \dots, \lambda_M\} \subset \mathbb{Z}$ be a one-dimensional frequency sampling set ($d=1$) with $M\geq n$ maximum frequency $W := \max_{1 \le m \le M} |\lambda_m|$. Suppose the target sources $Y = \{y_1, \dots, y_n\} \subset \mathbb{R}$ form a local cluster around a center $y_0$, satisfying the \emph{cluster condition}:
\begin{equation}\label{eq:cluster_condition_def}
    \max_{1 \le j \le n} |y_j - y_0| \le \frac{\tau \theta_{\min}}{2}, \quad \text{where} \quad \theta_{\min} := \min_{j \ne k} |y_j - y_k|.
\end{equation}
Let $\mathbf{V} \in \mathbb{C}^{M \times n}$ be the generalized Vandermonde matrix with entries $(\mathbf{V})_{m,j} = e^{i \lambda_m y_j}$. For small enough $\theta_{\min}$, there exists a constant $C_1(n) > 0$ depending only on $n$ such that:
\begin{equation}\label{eq:vdm_single_lower_bound}
    \sigma_{\min}(\mathbf{V}) \ge C_1(n) (\gamma_n(\Lambda) \theta_{\min})^{n-1}.
\end{equation}
\end{thm}

Now, since $\min(\gamma_n(\mathcal A), \gamma_n(\mathcal B))\geq C_2(n)\frac{\Omega}{n}$, by (\ref{eq:vdm_single_lower_bound}) and (\ref{equ:minisinguofrandomvander}), we have 
\[
\sigma_n(\mathbf {GH}_0)
\ge m_{\min}\sigma_{\min}(\mathbf V_1)\sigma_{\min}(\mathbf V_2)
>2\sigma\sqrt{M_1 M_2}.
\]
This yields
\[
\hat \sigma_j \geq \sigma_n(\mathbf {GH}_0)- \|\Delta\|_2 > \sigma\sqrt{M_1 M_2}.
\]
\end{proof}



\begin{algorithm}[H]
\caption{Random GHM Singular-Value-Thresholding Number Detection Algorithm}\label{alg:rand_hankel_svt}
\small
\KwIn{Measurement function \(\mathbf{Y}(\boldsymbol{\omega})\), admissible sets $\mathcal A,\mathcal B$, noise level \(\sigma\), sample sizes \(M_1, M_2\).}
\KwOut{The detected number of signals \(\hat{n}\).}

Draw the two frequency lists uniformly without replacement from $\mathcal A$ and $\mathcal B$, respectively\;
Construct \(\mathbf{GH}_{\text{rand}} \in \mathbb{C}^{M_1 \times M_2}\) with \((\mathbf{GH}_{\text{rand}})_{p,q} = \mathbf{Y}(\boldsymbol{\omega}_p + \boldsymbol{\zeta}_q)\)\;
\(\varepsilon_{\text{rand}} \gets \sigma \sqrt{M_1 M_2}\) \;
\(\hat\sigma_1\ge \hat\sigma_2\ge \cdots \ge \hat\sigma_{\min(M_1, M_2)} \gets \textbf{svd}(\mathbf{GH}_{\text{rand}})\)\;

\For{$i=\min(M_1, M_2)$ \textbf{downto} $1$}{
    \If{$\hat\sigma_i>\varepsilon_{\text{rand}}$}{
        \Return \(\hat{n}=i\)\;
    }
}
\Return \(\hat{n}=0\)\;
\end{algorithm}

The SVD of the prescribed $M_1\times M_2$ matrix costs
$O(M_1M_2\min\{M_1,M_2\})$, while forming and checking all frequency sums costs $O(dM_1M_2)$ and uses at most $M_1M_2$ distinct Fourier queries. The theoretical results presented below, together with those in \cite{huang2026minimum}, indicate that choosing $M_1,M_2=O(n)$ is sufficient for super-resolving $n$ closely spaced sources. Moreover, the randomized structure provides an additional advantage when resolving sources distributed across different clusters. Consequently, the overall computational complexity is reduced to $O(\max(n^
3,dn^2))$, which is substantially more efficient than that of classical algorithms; see next subsection for numerical comparisons.



\subsection{Numerical experiments}

In this subsection, we evaluate the three number detection constructions of Subsections~\ref{subsec:uniform array}, \ref{subsec:equally_distributed} and \ref{subsec:rand-numdetect} in dimensions $d=1$ and $d=2$. Throughout all experiments, we assume the sources have unit amplitudes and additive white noise with level $\sigma = 0.1$.

The labels in the tables are as follows and the same names are used in both dimensions.

The notation \emph{full-GHM} denotes the full grid construction of
Subsection~\ref{subsec:uniform array}, i.e., the GHM $\mathbf {GH}(s)$ in
\eqref{hankel} with $s=\Omega$ and consecutive index set
$\Lambda^d=\{0,1,\dots,\Omega\}^d$, then $\boldsymbol\omega \in \Gamma_s=\{-\Omega,-\Omega+1,\dots,\Omega\}^d$.

\emph{GHM-5} denotes the segmented construction of
Subsection~\ref{subsec:equally_distributed}, i.e., the GHM
$\mathbf {GH}(m,r,D)$ in \eqref{eq:segmented_hankel} with $m=10$, $r=2$,
$D=400$, and $\Omega=rD+m=810$. Its one-dimensional index set is
\(
\Lambda = \{0,1,\dots,10\} \cup \{400,401,\dots,410\} \cup  \{800,801,\dots,810\}
\),
then
\[
\boldsymbol\omega\in\Gamma_{\mathrm{seg}}=\{-810,\dots,-790\}\cup\{-410,\dots,-390\}\cup\{-10,\dots,10\}\cup\{390,\dots,410\}\cup\{790,\dots,810\}.
\]
In dimension $d=2$ we use
the Cartesian product $\Lambda^2$, denoted GHM-$5\times5$.

The notation \emph{randGHM-$M_1\times M_2$} denotes the random GHM
$\mathbf {GH}_{\mathrm{rand}}\in\mathbb C^{M_1\times M_2}$ of
Subsection~\ref{subsec:rand-numdetect} and \eqref{eq:random_hankel}. The $M_1$ row frequencies and
$M_2$ column frequencies are drawn uniformly without replacement from
admissible integer sets subject to the bandlimit constraint
$\boldsymbol\omega_p+\boldsymbol\zeta_q\in[-\Omega,\Omega]^d$.



We first consider a five-source signal in dimension $d=1$, consisting of two well-separated clusters $\{1.1,1.1+\Delta,2,2+\Delta,2+2\Delta\}$. The resolution and runtime are reported in Table~\ref{table:NumberVDM_1D}. We observe that our methods based on sub-samplings (GHM-5, randGHM-$20\times 20$, randGHM-$50\times 50$) achieves both the optimal order of resolution and runtime. In particular, they are $1000+$ times faster than the Hankel-based algorithm (full-GHM ($\Omega=810$) ).  


\begin{table}[htbp]
    \caption{Resolution and computational time for various methods in dimension $d=1$.}
    \centering 
    \begin{tabular}{c|ccc} 
    Method          & SVD Time (ms)                        & Total Time (ms) & Resolution \\ \hline
    full-GHM ($\Omega=10$)         &    $0.03 \pm 0.14$      & $0.07\pm 0.20$    & $0.3698$                   \\
    full-GHM ($\Omega=810$)      &    $165.6 \pm 13.2$        & $169.4\pm 13.7$ & $3.9\times 10^{-3}$      \\
    GHM-5 ($\Omega=810$)         &     $0.07 \pm 0.02$     & $0.10\pm 0.12$       & $3.1\times 10^{-3}$      \\
    randGHM-$20\times 20$ ($\Omega=810$)     &      $0.03 \pm 0.01$        & $0.06\pm 0.03 $       & $4.5\times 10^{-3}$      \\
    randGHM-$50\times 50$ ($\Omega=810$)     &      $0.19 \pm 0.06$        & $0.27\pm 0.10$       & $4.2\times 10^{-3}$      \\
    \end{tabular}
    \label{table:NumberVDM_1D}
\end{table}

We next consider five sources in dimension $d=2$, consisting of two clusters $\{(1.1, 1.1), (1.1+\Delta, 1.1), (2.0, 2.0), (2.0+\Delta, 2.0), (2.0, 2.0+\Delta)\}$. The resolution and runtime are reported in Table~\ref{table:NumberDetection_2D}. Similar to the one-dimensional case, the random GHM achieves the optimal runtime and resolution. However, due to the large size in the two-dimensional case (GHM-$5\times5$ has
size $1089\times1089$), the GHM with segmented index is no longer comparable to the random GHM, showing the efficiency of random sampling. The OOM (out of memory) indicates that
the computation could not be completed on our computer: with
$s=\Omega=810$ and $d=2$, the matrix $\mathbf {GH}(s)$ in \eqref{hankel} has
$(s+1)^d=811^2=657721$ rows and columns, i.e., about $4.3\times10^{11}$
entries, which already requires roughly $7$\,TB of memory in complex double
precision before any arithmetic is performed. 



\begin{table}[htbp]
    \caption{Resolution and computational time for number detection in dimension $d=2$.}
    \centering
    \begin{tabular}{c|ccc}
    Method & SVD Time (ms) & Total Time (ms) & Resolution \\ \hline
    full-GHM ($\Omega=30$)                   & $236.6 \pm 16.3$ & $241.6 \pm 16.5$ & 0.0612 \\
full-GHM ($\Omega=810$)                  &  OOM &  OOM  & - \\
GHM-$5\times 5$ ($\Omega=810$)            & $294.4 \pm 14.5$ & $300.5 \pm 14.7$ & $1.7\times 10^{-3}$ \\
randGHM-$20\times 20$ ($\Omega=810$)      & $0.04 \pm 0.05$ & $0.08 \pm 0.09$ & $2.9\times 10^{-3}$ \\
randGHM-$50\times 50$ ($\Omega=810$)      & $0.25 \pm 0.18$ & $0.36 \pm 0.23$ & $2.5\times 10^{-3}$ \\
    \end{tabular}
    \label{table:NumberDetection_2D}
\end{table}

\section{GHM-based location recovery}\label{sec:location-recovery}

In this section we consider the location recovery problem for the imaging model \eqref{equ:modelsetting1}: given the noisy measurements $\mathbf Y(\boldsymbol\omega)$, $\boldsymbol\omega\in[-\Omega,\Omega]^d$, recover the source locations $\mathbf y_1,\dots,\mathbf y_n\in\mathbb R^d$, where the source number $n$ is assumed to be known. Once the locations are recovered, the amplitudes $a_1,\dots,a_n$ are obtained by solving a linear least squares problem.



A direct approach to the problem is to Fourier-invert the measurements. For a uniform frequency grid this is the classical DFT, whose resolution is limited by the Rayleigh limit. For a segmented or sparse frequency set, the corresponding non-uniform DFT can have large grating-lobe sidelobes and may merge sources whose separation is below the Rayleigh limit, as illustrated in Figure~\ref{fig:nudft-vs-music}.

\begin{figure}[htb]
    \centering
    \includegraphics[width=\linewidth]{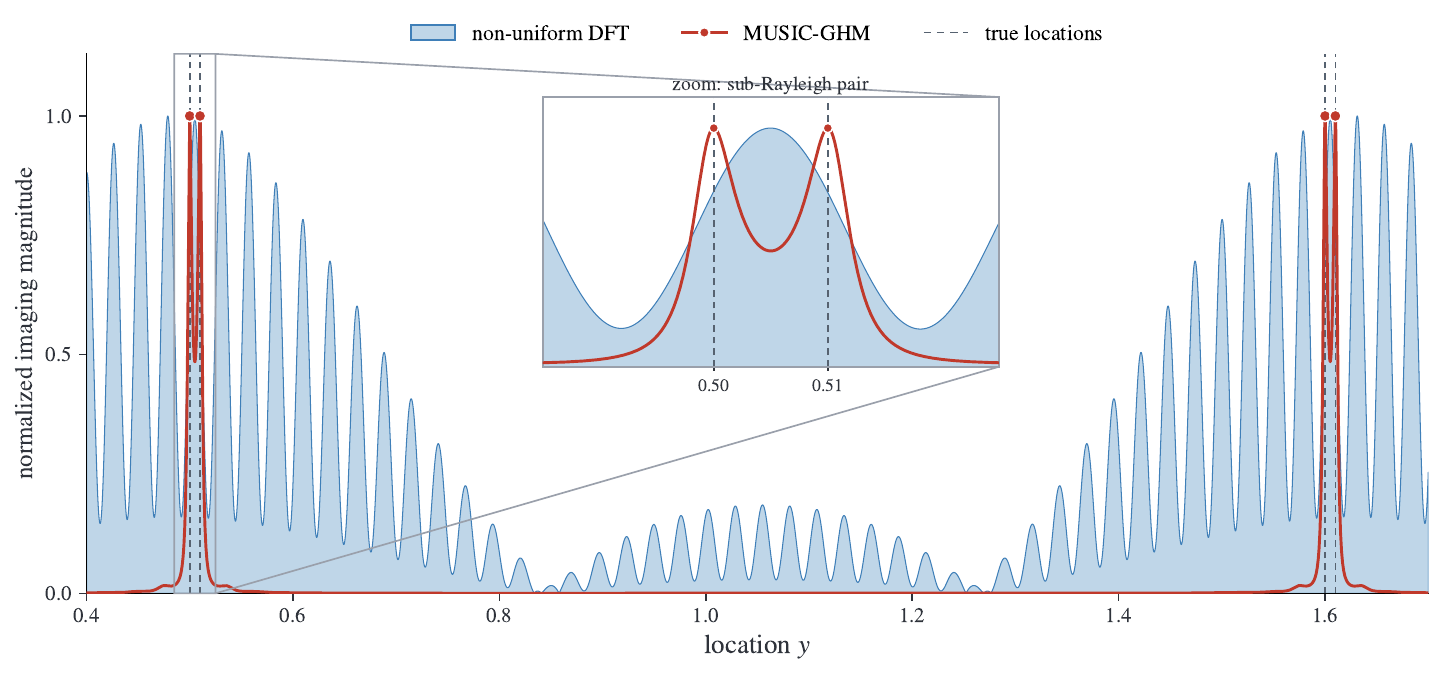}
    \caption{Non-uniform DFT versus GHM-based MUSIC on the same segmented-frequency measurements for the source locations $\{0.5,0.51,1.6,1.61\}$, with index set $\Lambda=\{0,\dots,10\}\cup\{240,\dots,250\}$. The non-uniform DFT has large grating-lobe sidelobes and merges each sub-Rayleigh pair into a single peak (inset), whereas GHM-based MUSIC resolves the sources.}
    \label{fig:nudft-vs-music}
\end{figure}

Beyond Fourier inversion, the existing methods for this problem include MUSIC~\cite{schmidt1986multiple,liao2015music}, root-MUSIC~\cite{barabell1983improving,raohari1989performance,boyer2008decoupled}, ESPRIT~\cite{roy1989esprit,zoltowski1996closed,haardt1998simultaneous} and the matrix pencil method~\cite{hua1990matrix,hua1992memp}. The cited multidimensional root-MUSIC, ESPRIT, and matrix-pencil formulations use measurements on a full tensor grid: root-MUSIC exploits the polynomial structure of the steering vector, ESPRIT the shift-invariance of a uniform array, and the matrix pencil method the multi-level Hankel structure of the data matrix. The MUSIC orthogonality criterion extends to a general GHM, but its implementation must also use the row-frequency steering map of that GHM and a localization procedure adapted to the resulting objective.

Section~\ref{subsec:music-classical} presents a unified MUSIC framework for GHMs, with classical multi-level Hankel matrices as a special case, and discusses a coarse-to-local implementation and its cost. Section~\ref{subsec:music-stability} establishes its noise stability. Random arrays are treated in Section~\ref{subsec:random-array-music}, and numerical experiments are reported in Section~\ref{subsec:music-numerics}.

\subsection{Unified MUSIC Framework}\label{subsec:music-classical}

We recall the MUSIC algorithm \cite{schmidt1986multiple,liao2016music,liao2015music} in the multi-dimensional setting. The observed matrix $\mathbf H(s)$ defined in \eqref{hankel} collects the measurements on the integer grid $\Gamma_s=\{-s,\dots,s\}^d$ and satisfies $\mathbf H(s)=\mathbf H_0(s)+\Delta$ by \eqref{eq:factorization_H}. For $\mathbf y\in\mathbb R^d$, define the steering vector
\[
\boldsymbol\phi_s(\mathbf y):=\psi_s\!\left(e^{i(\mathbf y)_1}\right)\otimes\cdots\otimes\psi_s\!\left(e^{i(\mathbf y)_d}\right)\in\mathbb C^{(s+1)^d},
\]
whose $\ell$-th entry is $e^{i\boldsymbol\alpha^{(\ell)}\cdot\mathbf y}$, so that $\mathbf V_1=[\boldsymbol\phi_s(\mathbf y_1),\dots,\boldsymbol\phi_s(\mathbf y_n)]$ in \eqref{eq:factorization_H}.

Assume that $\operatorname{rank}\mathbf H_0(s)=n$ and write its singular value decomposition as $\mathbf H_0(s)=U_0\Sigma_0V_0^*$, with $U_0=[U_{0,1},U_{0,2}]$ and $U_{0,1}$ containing the first $n$ left singular vectors. Since $\operatorname{range}\mathbf H_0(s)=\operatorname{span}\{\boldsymbol\phi_s(\mathbf y_j)\}_{j=1}^n$, one has
\[
U_{0,2}^*\boldsymbol\phi_s(\mathbf y_j)=0,\qquad j=1,\dots,n.
\]
Under the identifiability condition that no other steering vector lies in this range, the noise-free imaging function
\[
J_0(\mathbf y)=\frac{\|\boldsymbol\phi_s(\mathbf y)\|}{\|U_{0,2}^*\boldsymbol\phi_s(\mathbf y)\|}
\]
has poles exactly at the source locations. 

In the presence of noise, we run MUSIC with the noisy matrix $\mathbf H(s)$ in place of its noiseless counterpart $\mathbf H_0(s)$, as summarized in Algorithm~\ref{alg:music-grid}. Stability for the classical multi-level Hankel construction was established in \cite{liao2015music}. The corresponding analysis for GHMs is given in Section~\ref{subsec:music-stability}.

\begin{algorithm}[H]
\caption{Unified MUSIC Framework}\label{alg:music-grid}
\small
\KwIn{Measurements \(\mathcal Y\) on \(\Gamma_s\), source number \(n\), search grid \(\mathcal G\subset(-\pi,\pi]^d\).}
\KwOut{Estimated locations \(\hat{\mathbf y}_1,\dots,\hat{\mathbf y}_n\).}

Form the data matrix \(\mathbf H(s)\) in \eqref{hankel} from \(\mathcal Y\)\;
\([U,\Sigma,V]\gets\textbf{svd}(\mathbf H(s))\); set the noise space \(U_2\) to all columns of \(U\) but the first \(n\)\;
\ForEach{$\mathbf y\in\mathcal G$}{
    \(J(\mathbf y)\gets \|\boldsymbol\phi_s(\mathbf y)\|/\|U_2^*\boldsymbol\phi_s(\mathbf y)\|\)\;
}
\Return the \(n\) largest local peaks \(\hat{\mathbf y}_1,\dots,\hat{\mathbf y}_n\) of \(J\)\;
\end{algorithm}

For a GHM, the row-frequency set determines both the left Vandermonde factor and the steering map used for localization. Let $\mathcal A=\{\boldsymbol\alpha_1,\ldots,\boldsymbol\alpha_{M_1}\}$ and $\mathcal B=\{\boldsymbol\beta_1,\ldots,\boldsymbol\beta_{M_2}\}$ be its row- and column-frequency sets, respectively, and let $\mathcal X=\{\mathbf y_1,\ldots,\mathbf y_n\}$. The noiseless GHM $\mathbf{GH}_0\in\mathbb C^{M_1\times M_2}$ admits the generalized Vandermonde decomposition
\begin{equation}\label{eq:ghm_music_factorization}
\mathbf{GH}_0=\vect V_1 \mathbf\Sigma \vect V_2^\top,
\qquad
\vect V_1= \mathcal V_{\mathcal A}(\mathcal X),
\quad
\vect V_2=\mathcal V_{\mathcal B}(\mathcal X),
\end{equation}
where $\mathbf\Sigma =\operatorname{diag}(a_1,\dots,a_n)$. Its row-frequency steering vector is $\boldsymbol\phi_{\mathcal A}$ from Definition~\ref{generalized vandermonde}. For a noisy GHM $\mathbf{GH}_\sigma=\mathbf {GH}_0+\Delta$, Algorithm~\ref{alg:music-grid} applies with $\mathbf H(s)$ and $\boldsymbol\phi_s$ replaced by $\mathbf {GH}_\sigma$ and $\boldsymbol\phi_{\mathcal A}$, respectively. 


\begin{remark}\label{rem:partial_svd}
A coarse-to-local implementation first evaluates the MUSIC imaging function on a coarse grid $\mathcal G_{\mathrm c}$ and then refines the selected candidates with a continuous local solver, as in Gradient-MUSIC~\cite{fannjiang2026gradient}. In practice, one may compute the leading $n$ left singular vectors $U_c$ by $\operatorname{svds}$ and form the noise-space basis as $\operatorname{null}(U_c^*)$.
For a dense matrix, a partial SVD using $N_{\mathrm{mv}}$ matrix--vector products costs $O(N_{\mathrm{mv}}M_1M_2)$. For unstructured row frequencies, the coarse scan costs $O(|\mathcal G_{\mathrm c}|M_1(n+d))$. If $N_{\mathrm c}$ candidates are refined for $N_{\mathrm{it}}$ first-order iterations with analytic steering derivatives, the refinement costs $O(N_{\mathrm c}N_{\mathrm{it}}(M_1nd+d^3))$. The total cost is therefore
\[
O\!\left(
N_{\mathrm{mv}}M_1M_2
+|\mathcal G_{\mathrm c}|M_1(n+d)
+N_{\mathrm c}N_{\mathrm{it}}(M_1nd+d^3)
\right).
\]
\end{remark}

\subsection{Stability Analysis}\label{subsec:music-stability}

For the stability analysis, define
\[
\mathbf a(\mathbf y):=\frac{\boldsymbol\phi_{\mathcal A}(\mathbf y)}{\|\boldsymbol\phi_{\mathcal A}(\mathbf y)\|_2},
\]
and let
\[
\mathbf {GH}_0=U\Sigma V^*,
\qquad
U=[U_1,U_2],
\]
where the columns of $U_1$ span the $n$-dimensional signal subspace and the columns of $U_2$ span its orthogonal complement. Define the exact noise-space correlation and imaging function by
\begin{equation}
\label{eq:music-exact-objective}
P_N:=U_2U_2^*,
\qquad
R(\mathbf y):=\|U_2^*\mathbf a(\mathbf y)\|_2,
\qquad
J(\mathbf y):=\frac{1}{R(\mathbf y)}.
\end{equation}

For the perturbed data matrix $\mathbf {GH}_\sigma=\mathbf{GH}_0+\Delta$, write its exact singular value decomposition as
\begin{equation}
\label{eq:music-perturbed-svd}
\mathbf{GH}_{\sigma}=\widetilde U \widetilde\Sigma \widetilde V^*,
\qquad \widetilde U=[\widetilde U_1,\widetilde U_2],
\end{equation}
and define
\[
R_{\sigma}(\mathbf y):=\|\widetilde U_2^*\mathbf a(\mathbf y)\|_2,
\qquad
J_\sigma(\mathbf y):=\frac{1}{R_\sigma(\mathbf y)}.
\]
For GHMs, the following result follows by the same Wedin-type argument as Li and Liao~\cite[Proposition~4.2]{li2021stable}; the proof is omitted. Recall that $m_{\min}=\min_{j=1,\dots,n}|a_j|$.

\begin{lem}
\label{lem:stability_ghm_music}
Assume that $M_1>n$, $M_2\ge n$, $m_{\min}>0$, and that $\vect V_1, \vect V_2$ in \eqref{eq:ghm_music_factorization} have full column rank $n$. If
\[
2\|\Delta\|_2<m_{\min}\sigma_{\min}(\vect V_1)\sigma_{\min}(\vect V_2),
\]
then
\[
\|R_{\sigma}-R\|_{\infty}
\le
\frac{2\|\Delta\|_2}{m_{\min}\sigma_{\min}(\vect V_1)\sigma_{\min}(\vect V_2)}.
\]
\end{lem}

For segmented arrays under the multi-clump geometry, the general stability bound admits the following explicit form.

\begin{cor}
\label{cor:stability_multidim_segmented}
Assume the hypotheses of Theorem~\ref{thm:segmented-vandermonde} and, in addition, that $m\ge n$. For the segmented factorization \eqref{eq:factorization_H_segmented}, let $m_{\min}:=\min_{1\le j\le n}|a_j|>0$ and define
\begin{equation}
\label{eq:segmented-music-Bcl}
B_{\mathrm{cl}}
:=
\frac{1}{\sqrt n}
\left(2-e^{1/(2\beta)}\right)^{n^\star/2}
\frac{\sqrt{(\frac{r}{n^\star})^d(\lfloor\frac m2\rfloor+1)^d}}
{(\sqrt2)^{n^\star-1}}
\left(\frac{rD}{\pi n^\star}\Delta_1(\mathcal X)\right)^{n^\star-1}.
\end{equation}
Then
\[
\sigma_{\min}(\mathbf V_1)=\sigma_{\min}(\mathbf V_2)\ge B_{\mathrm{cl}}>0.
\]
If
\[
2\|\Delta\|_2<m_{\min}B_{\mathrm{cl}}^2,
\]
then the corresponding noise-space correlations satisfy
\[
\|R_{\sigma}-R\|_{\infty}
\le
\frac{2\|\Delta\|_2}{m_{\min}B_{\mathrm{cl}}^2}.
\]
\end{cor}

\begin{proof}
Theorem~\ref{thm:segmented-vandermonde} gives $\sigma_{\min}(\mathbf V_1)\sigma_{\min}(\mathbf V_2)\ge B_{\mathrm{cl}}^2$. Substituting this bound into Lemma~\ref{lem:stability_ghm_music} proves the result.
\end{proof}

For other array geometries, Lemma~\ref{lem:stability_ghm_music} can be combined with any applicable lower bound on the minimum singular value of the Vandermonde factors. Examples in their respective sampling geometries include the multivariate bounds in \cite{kunis2020smallest} and the restricted-Fourier bounds in \cite{li2021stable}.

\FloatBarrier

\subsection{Random arrays}\label{subsec:random-array-music}

For the grid-based MUSIC algorithm of Section~\ref{subsec:music-classical}, the data matrix has size $(s+1)^d\times(s+1)^d$. The random GHM $\mathbf{GH}_{\text{rand}}$ in \eqref{eq:random_hankel} provides a smaller matrix for location recovery. Its noiseless part has the factorization \eqref{eq:ghm_music_factorization}, with row frequencies $\boldsymbol\omega_p$ and column frequencies $\boldsymbol\zeta_q$. The corresponding steering vector is
\[
\boldsymbol\phi_{\mathcal A}(\mathbf y)=
[e^{i\boldsymbol\omega_1\cdot\mathbf y},\dots,e^{i\boldsymbol\omega_{M_1}\cdot\mathbf y}]^\top .
\]
The MUSIC procedure and the peak search are then applied as described in Section~\ref{subsec:music-classical}.

Unlike the segmented array, a realized random row set can exhibit manifold ambiguity even when the associated Vandermonde factors have full column rank. The experiments below evaluate realized draws, but no uniform high-probability recovery theorem over the continuous signal class is claimed.

The difference in the required measurements is illustrated in Figure~\ref{fig:full-vs-random-ghm-measurements}. In the full tensor-grid construction, the row and column frequency sets form a Cartesian grid, so all pairwise sums fill a larger Cartesian grid of measurements. In the random GHM construction, only the sums $\boldsymbol\omega_p+\boldsymbol\zeta_q$ associated with the sampled row and column frequencies are used.

\begin{figure}[htbp]
\centering
\begin{subfigure}[t]{0.48\textwidth}
\centering
\begin{tikzpicture}[x=0.36cm,y=0.36cm]
    \foreach \x in {-4,...,4} {
        \draw[gray!25] (\x,-4) -- (\x,4);
        \draw[gray!45] (\x,-4.08) -- (\x,-3.92);
    }
    \foreach \y in {-4,...,4} {
        \draw[gray!25] (-4,\y) -- (4,\y);
        \draw[gray!45] (-4.08,\y) -- (-3.92,\y);
    }
    \draw[->] (-4.45,0) -- (4.65,0) node[right] {$k_1$};
    \draw[->] (0,-4.45) -- (0,4.65) node[above] {$k_2$};
    \foreach \x in {-4,-2,0,2,4} {
        \node[font=\scriptsize,below] at (\x,-4.45) {$\x$};
    }
    \foreach \y in {-4,-2,0,2,4} {
        \node[font=\scriptsize,left] at (-4.45,\y) {$\y$};
    }
    \foreach \x in {-4,...,4} {
        \foreach \y in {-4,...,4} {
            \fill[blue!70!black] (\x,\y) circle (2pt);
        }
    }
    \node[draw,fill=white,rounded corners=1pt,align=center,font=\scriptsize] at (0,6.05)
    {full GHM\\$25\times25$ matrix\\$81$ measurements};
\end{tikzpicture}
\caption{Full tensor-grid GHM.}
\end{subfigure}
\hfill
\begin{subfigure}[t]{0.48\textwidth}
\centering
\begin{tikzpicture}[x=0.36cm,y=0.36cm]
    \foreach \x in {-4,...,4} {
        \draw[gray!25] (\x,-4) -- (\x,4);
        \draw[gray!45] (\x,-4.08) -- (\x,-3.92);
    }
    \foreach \y in {-4,...,4} {
        \draw[gray!25] (-4,\y) -- (4,\y);
        \draw[gray!45] (-4.08,\y) -- (-3.92,\y);
    }
    \draw[->] (-4.45,0) -- (4.65,0) node[right] {$k_1$};
    \draw[->] (0,-4.45) -- (0,4.65) node[above] {$k_2$};
    \foreach \x in {-4,-2,0,2,4} {
        \node[font=\scriptsize,below] at (\x,-4.45) {$\x$};
    }
    \foreach \y in {-4,-2,0,2,4} {
        \node[font=\scriptsize,left] at (-4.45,\y) {$\y$};
    }
    \foreach \x in {-4,...,4} {
        \foreach \y in {-4,...,4} {
            \fill[gray!35] (\x,\y) circle (0.8pt);
        }
    }
    \foreach \wx/\wy in {-2/-2,-1/1,1/-1,2/2} {
        \foreach \zx/\zy in {-2/1,0/-2,1/0,2/-1} {
            \pgfmathtruncatemacro{\sx}{\wx+\zx}
            \pgfmathtruncatemacro{\sy}{\wy+\zy}
            \fill[red!75!black] (\sx,\sy) circle (2.4pt);
        }
    }
    \node[draw,fill=white,rounded corners=1pt,align=center,font=\scriptsize] at (0,6.05)
    {random GHM \\$4\times4$ matrix\\$16$ measurements};
\end{tikzpicture}
\caption{Random GHM.}
\end{subfigure}
\caption{Measurement indices for full tensor-grid GHM and random GHM constructions in a two-dimensional example.}
\label{fig:full-vs-random-ghm-measurements}
\end{figure}
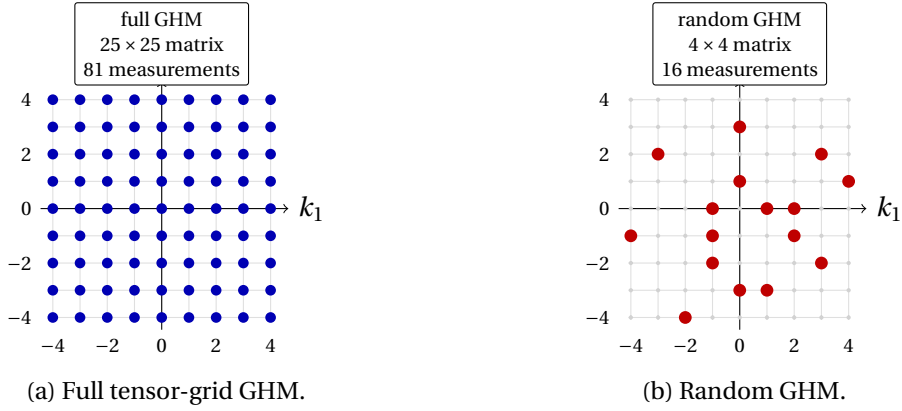

For both constructions, the method computes a rank-$n$ partial SVD. For dense matrices, one matrix--vector product costs $O(M_1M_2)$ for the random GHM and $O((s+1)^{2d})$ for the full tensor-grid matrix. Hence the corresponding subspace costs are $O(N_{\mathrm{mv}}M_1M_2)$ and $O(N_{\mathrm{mv}}(s+1)^{2d})$, respectively, with method-dependent iteration counts. The complete cost, including the location search, is given in Remark~\ref{rem:partial_svd}; reducing the matrix size affects only the subspace term unless the search budgets are also controlled. In the experiments below, $M_1$ and $M_2$ are selected empirically at scales comparable to $n+d$; this choice is not a proven sample-complexity law.

\subsection{Numerical experiments}\label{subsec:music-numerics}

We test MUSIC with segmented and random GHMs in dimensions $d=1$ and $d=2$ and compare it with MUSIC based on the classical multi-level Hankel matrix, denoted full-GHM~($\Omega$), where $\Omega$ is the cutoff frequency.

Unless otherwise stated, the source amplitudes have unit magnitude
and the additive noise is independently generated in the sense of the
noise model \eqref{equ:modelsetting1}. For a fixed source configuration
and a fixed method, the reported RMSE is computed over
$N_{\mathrm{trial}}=100$ independent Monte Carlo trials by optimal
permutation matching:
\[
\operatorname{RMSE}=
\left(
\frac{1}{N_{\mathrm{trial}}n}
\sum_{\ell=1}^{N_{\mathrm{trial}}}
\min_{\pi\in S_n}
\sum_{j=1}^n
\|\widehat{\mathbf y}^{(\ell)}_{\pi(j)}-\mathbf y_j\|_2^2
\right)^{1/2}.
\]
Resolution is the smallest tested separation $\Delta$ for which the
empirical recovery probability is at least $95\%$,
where a trial is successful if all estimated locations can be matched
to the true locations within tolerance $\Delta/4$.
For random GHMs, the row and column frequency sets are drawn
from the admissible integer sets subject to the bandlimit
constraint $\boldsymbol\omega_p+\boldsymbol\zeta_q\in[-\Omega,\Omega]^d$.
The same random frequency draw is used across the Monte Carlo trials
of a fixed experiment; no resampling criterion depending on the unknown
source locations is used by the recovery algorithm.

\subsubsection{Comparison of GHM constructions under the MUSIC framework}

We first compare several GHM constructions within the unified MUSIC framework; see the results in Table~\ref{table:RMSEandResoForGHM1D}. The results show that MUSIC based on segmented or random GHMs achieves resolution of the same order as MUSIC based on the full-GHM ($\Omega=809$), while speeding up SVD computation by factors of approximately $500$--$5000$. Since peak search dominates the total time, the overall speedup is less pronounced.

\begingroup
\renewcommand{\topfraction}{0.9}
\renewcommand{\bottomfraction}{0.8}
\renewcommand{\textfraction}{0.05}
\renewcommand{\floatpagefraction}{0.85}



\begin{table}[htbp]
\centering
\footnotesize
\setlength{\tabcolsep}{2pt}
\begin{tabular}{c|cccc}
Method & SVD Time (ms) & Total Time (ms) & Resolution & RMSE ($\Delta=0.24$) \\ \hline
full-GHM ($\Omega=9$) & $0.0267$ & $0.912$ & \bluebox{0.21} & $1.38\times 10^{-2}$ \\
full-GHM ($\Omega=809$) & \bluebox{131.82} & \bluebox{135.04} & $1.4\times 10^{-3}$ & $2.64\times 10^{-6}$ \\
GHM-3 ($\Omega=809$) & $0.0601$ & $2.40$ & $4.6\times 10^{-3}$ & $1.12\times 10^{-5}$ \\
GHM-5 ($\Omega=809$) & $0.152$ & $2.56$ & $1.55\times 10^{-3}$ & $9.80\times 10^{-6}$ \\
randGHM-$10\times5$ ($\Omega=809$) & $0.0217$ & $1.90$ & $4.5\times 10^{-3}$ & $6.38\times 10^{-5}$ \\
randGHM-$20\times10$ ($\Omega=809$) & $0.0627$ & $2.65$ & $2.6\times 10^{-3}$ & $1.38\times 10^{-5}$ \\
randGHM-$50\times30$ ($\Omega=809$) & $0.218$ & $2.62$ & $2.2\times 10^{-3}$ & $5.21\times 10^{-6}$
\end{tabular}
\caption{Resolution and RMSE for the four-source configuration
$\{0.5,0.5+\Delta,1.6,1.6+\Delta\}$ in dimension $d=1$
(noise level $\sigma=0.1$). GHM-$3$ uses $m=9$, $r=1$, and $D=800$,
and GHM-$5$ uses $m=9$, $r=2$, and $D=400$. The randGHM-$M_1\times M_2$
methods use random GHMs of size $M_1\times M_2$,
as in Section~\ref{subsec:random-array-music}. }
\label{table:RMSEandResoForGHM1D}
\end{table}

Figure~\ref{fig:MUSIC1Dspectrum} shows the MUSIC spectra of several representative GHMs. The full-GHM with $\Omega=9$ merges the neighboring spectral peaks within each clump into a single peak and therefore cannot distinguish the sources within a clump. At the same computational cost, the segmented GHM and random GHM clearly resolve the within-clump spectral peaks.

\begin{figure}[htbp]
    \centering
    \begin{subfigure}{0.48\textwidth}
        \centering
        \includegraphics[width=\linewidth]{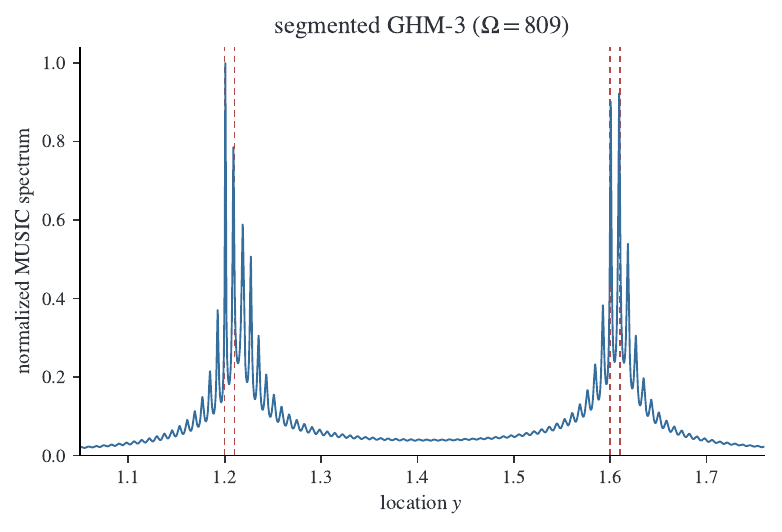}
    \end{subfigure}
    \hfill
    \begin{subfigure}{0.48\textwidth}
        \centering
        \includegraphics[width=\linewidth]{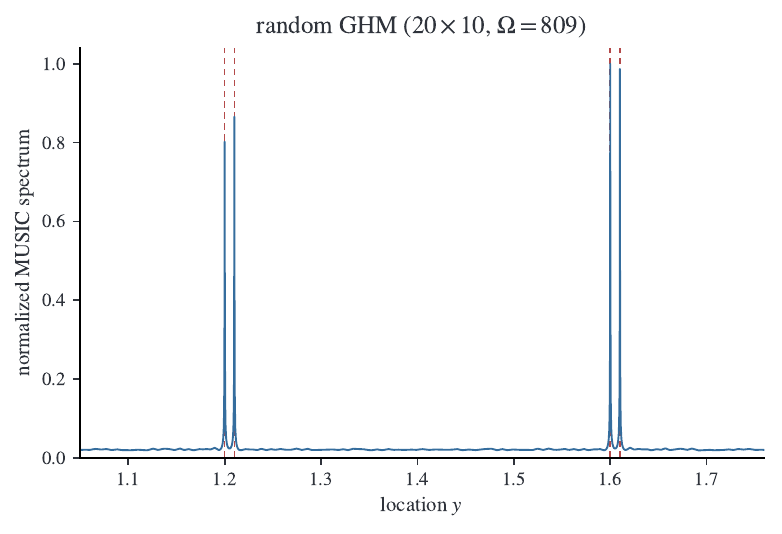}
    \end{subfigure}

    \begin{subfigure}{0.48\textwidth}
        \centering
        \includegraphics[width=\linewidth]{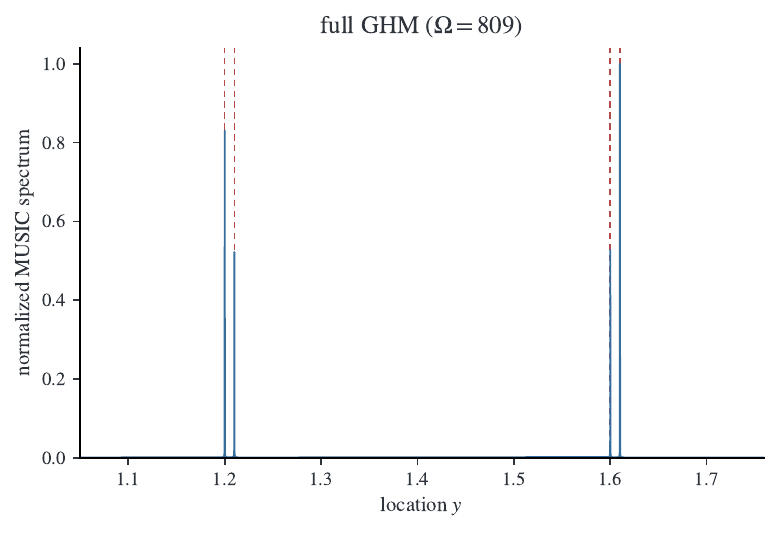}
    \end{subfigure}
    \hfill
    \begin{subfigure}{0.48\textwidth}
        \centering
        \includegraphics[width=\linewidth]{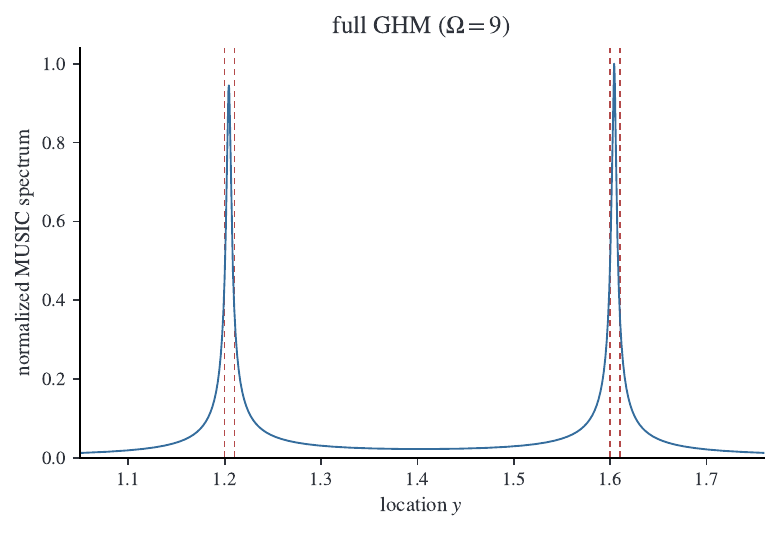}
    \end{subfigure}
    \caption{MUSIC imaging functions of four different methods in dimension $d=1$ for the source locations $\{1.2,1.21,1.6,1.61\}$. The segmented construction is GHM-$3$.}
    \label{fig:MUSIC1Dspectrum}
\end{figure}

We perform the same comparison in dimension $d=2$; see Table~\ref{table:RMSEandResoForGHM2D}. In this experiment, the noise subspace is computed using \texttt{svds} and \texttt{null} rather than a full SVD. For the full-GHM with $\Omega=809$, the matrix has size $656100\times656100$ and requires approximately $6.3$ TiB of dense complex storage. It is therefore not formed explicitly; its matrix--vector products are evaluated implicitly using FFTs.

\begin{table}[htbp]
\centering
\footnotesize
\setlength{\tabcolsep}{2pt}
\begin{tabular}{c|cccc}
Method & SVD Time (ms) & Total Time (ms) & Resolution & RMSE ($\Delta_{\min}=0.15\sqrt{2}$) \\ \hline
full-GHM ($\Omega=9$) & $0.124$ & $7.40$ & \bluebox{0.141} & $2.27\times10^{-3}$ \\
full-GHM ($\Omega=809$) & \bluebox{1667} & \bluebox{1787} & $3.54\times10^{-4}$ & $8.65\times10^{-8}$ \\
GHM-$3\times3$ ($\Omega=809$) & $0.447$ & $9.10$ & $7.07\times10^{-4}$ & $2.41\times10^{-6}$ \\
GHM-$5\times5$ ($\Omega=809$) & $2.06$ & $13.16$ & $9.19\times10^{-4}$ & $1.10\times10^{-6}$ \\
randGHM-$5\times5$ ($\Omega=809$) & $0.105$ & $8.42$ & $5.66\times10^{-3}$ & $9.23\times10^{-5}$ \\
randGHM-$50\times50$ ($\Omega=809$) & $0.107$ & $8.37$ & $1.77\times10^{-3}$ & $3.46\times10^{-6}$
\end{tabular}
\caption{Comparison of GHM constructions for the two-source configuration $\mathbf y_1=(1.1,1.1)$ and $\mathbf y_2=(1.1+\delta,1.1+\delta)$ in dimension $d=2$, where $\Delta=\sqrt2\,\delta$.}
\label{table:RMSEandResoForGHM2D}
\end{table}

\subsubsection{False peaks in MUSIC}

In the high-noise regime the segmented GHM methods with few subarrays may exhibit false peaks; GHM-$3$ is the typical case. The mechanism is visible in the noise-space correlation. In the noise-free case $U_2^*\boldsymbol\phi(\mathbf y)=0$ exactly at the true locations, so the true peaks of $J$ are infinite and dominate the spurious near-zeros of $\|U_2^*\boldsymbol\phi(\mathbf y)\|$. Under noise, $\|\widetilde U_2^*\boldsymbol\phi(\mathbf y)\|$ no longer vanishes at the true locations, while for the three-subarray geometry $\|U_2^*\boldsymbol\phi(\mathbf y)\|$ remains small at many spurious locations; the true and the spurious peaks become comparable, and false peaks appear; see Figure~\ref{fig:falsePeaksGHM3}.

However, for a random GHM, the irregular frequencies make systematic phase alignment at off-support locations unlikely. Consequently, the noise-space correlation fluctuates around a nonzero level away from the true locations, without the structured near-zeros observed for the segmented GHM; see Figure~\ref{fig:falsePeaksRandGHM}(a). As a result, false peaks are suppressed; see Figure~\ref{fig:falsePeaksRandGHM}(b). The same behavior is observed for the two-dimensional random GHM in Figure~\ref{fig:falsePeaksRandGHM2D}.

\FloatBarrier
\endgroup

\begin{figure}[H]
    \centering
    \begin{minipage}[t]{0.33\textwidth}
        \centering
        \includegraphics[width=\linewidth]{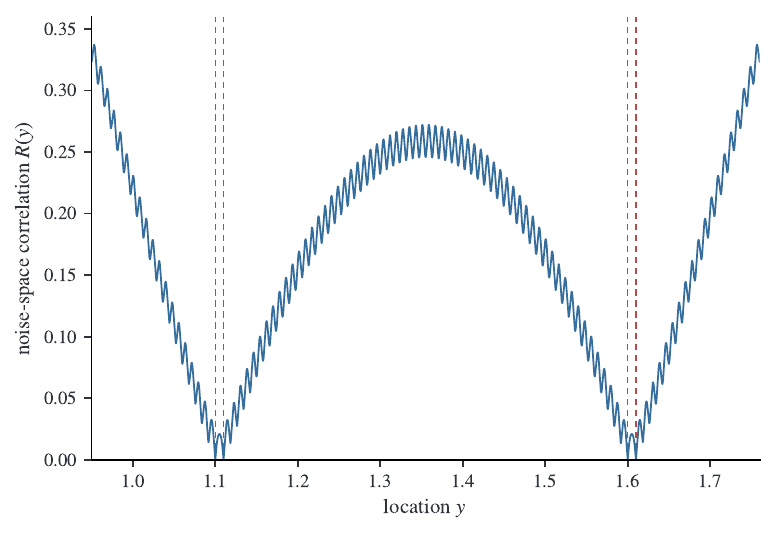}
        \caption*{(a) Noise space}
    \end{minipage}
    \hfill
    \begin{minipage}[t]{0.33\textwidth}
        \centering
        \includegraphics[width=\linewidth]{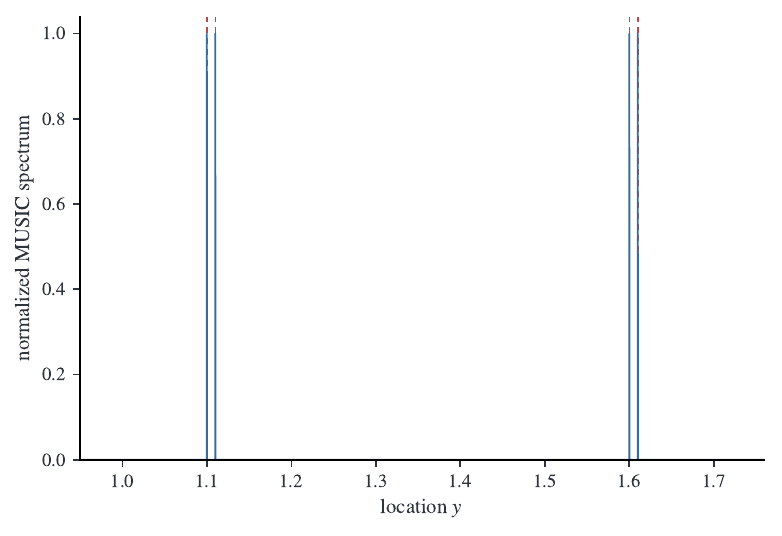}
        \caption*{(b) Spectrum (noise-free)}
    \end{minipage}
    \hfill
    \begin{minipage}[t]{0.33\textwidth}
        \centering
        \includegraphics[width=\linewidth]{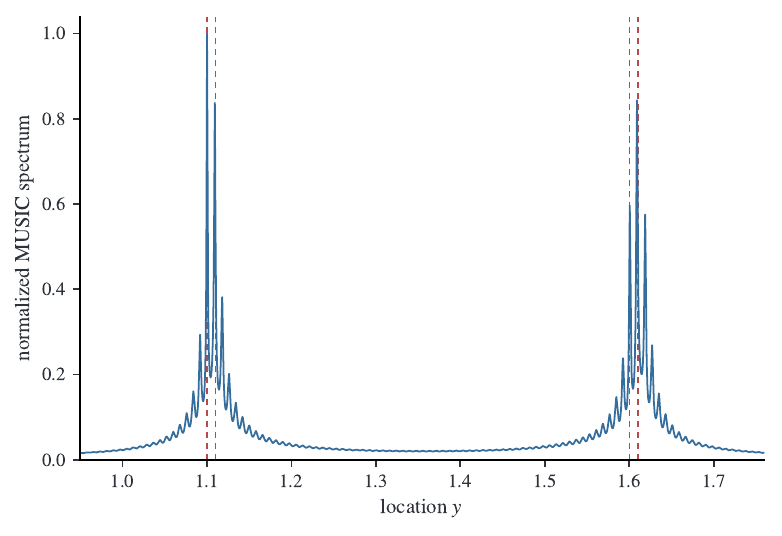}
        \caption*{(c) Spectrum ($\sigma=0.1$)}
    \end{minipage}
    \caption{MUSIC spectra and the noise space based on GHM-$3$ for the source locations $\{1.1,1.11,1.6,1.61\}$.}
    \label{fig:falsePeaksGHM3}
\end{figure}

\begin{figure}[H]
    \centering
    \begin{minipage}[t]{0.33\textwidth}
        \centering
        \includegraphics[width=\linewidth]{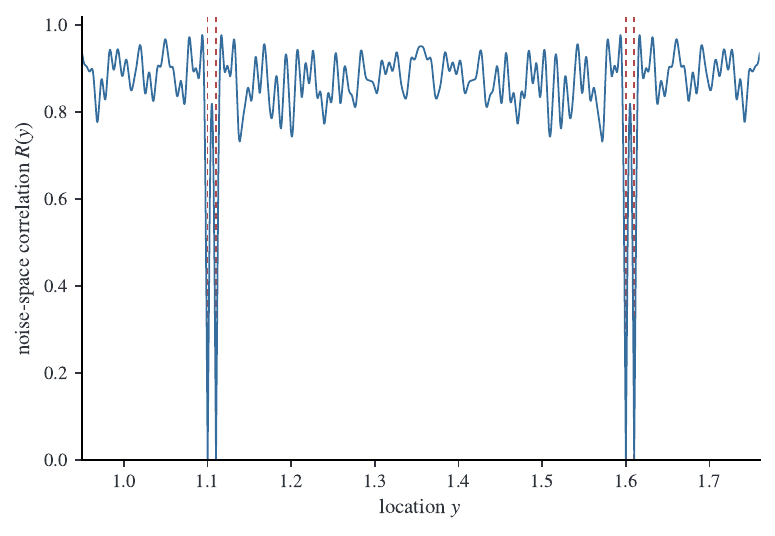}
        \caption*{(a) Noise space}
    \end{minipage}
    \hspace{3em}
    \begin{minipage}[t]{0.33\textwidth}
        \centering
        \includegraphics[width=\linewidth]{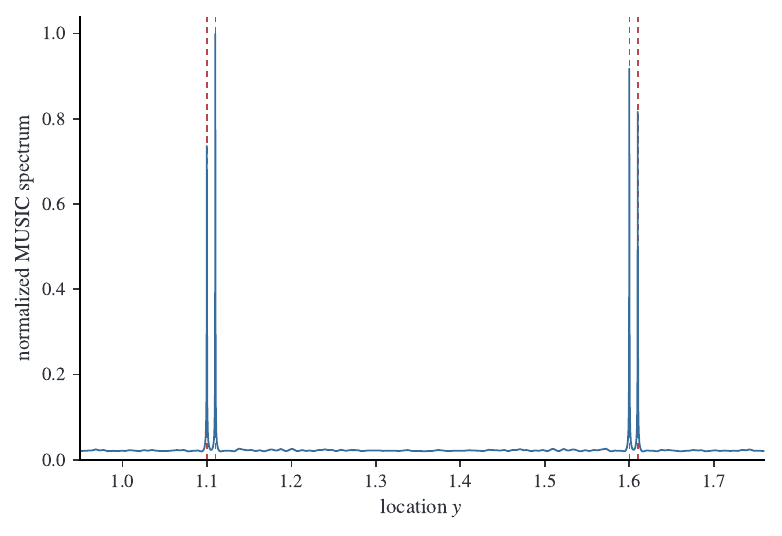}
        \caption*{(b) Spectrum ($\sigma=0.1$)}
    \end{minipage}
    \caption{MUSIC spectrum and the noise space based on a random GHM of size $20\times 10$ for the source locations $\{1.1,1.11,1.6,1.61\}$.}
    \label{fig:falsePeaksRandGHM}
\end{figure}

\begin{figure}[!htbp]
    \centering
    \begin{minipage}[t]{0.48\textwidth}
        \centering
        \includegraphics[width=\linewidth]{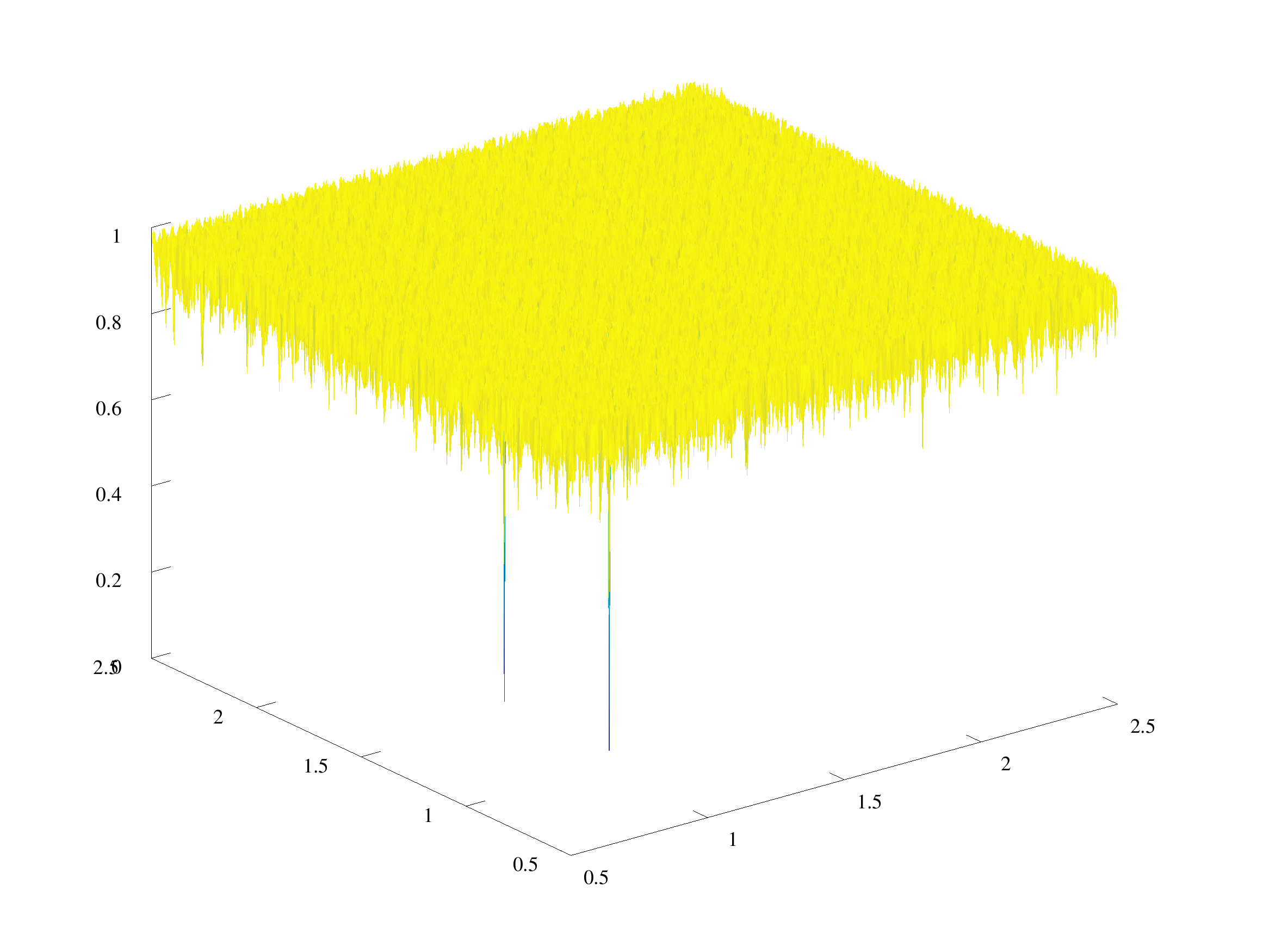}
        \caption*{(a) Noise space}
    \end{minipage}
    \hfill
    \begin{minipage}[t]{0.48\textwidth}
        \centering
        \includegraphics[width=\linewidth]{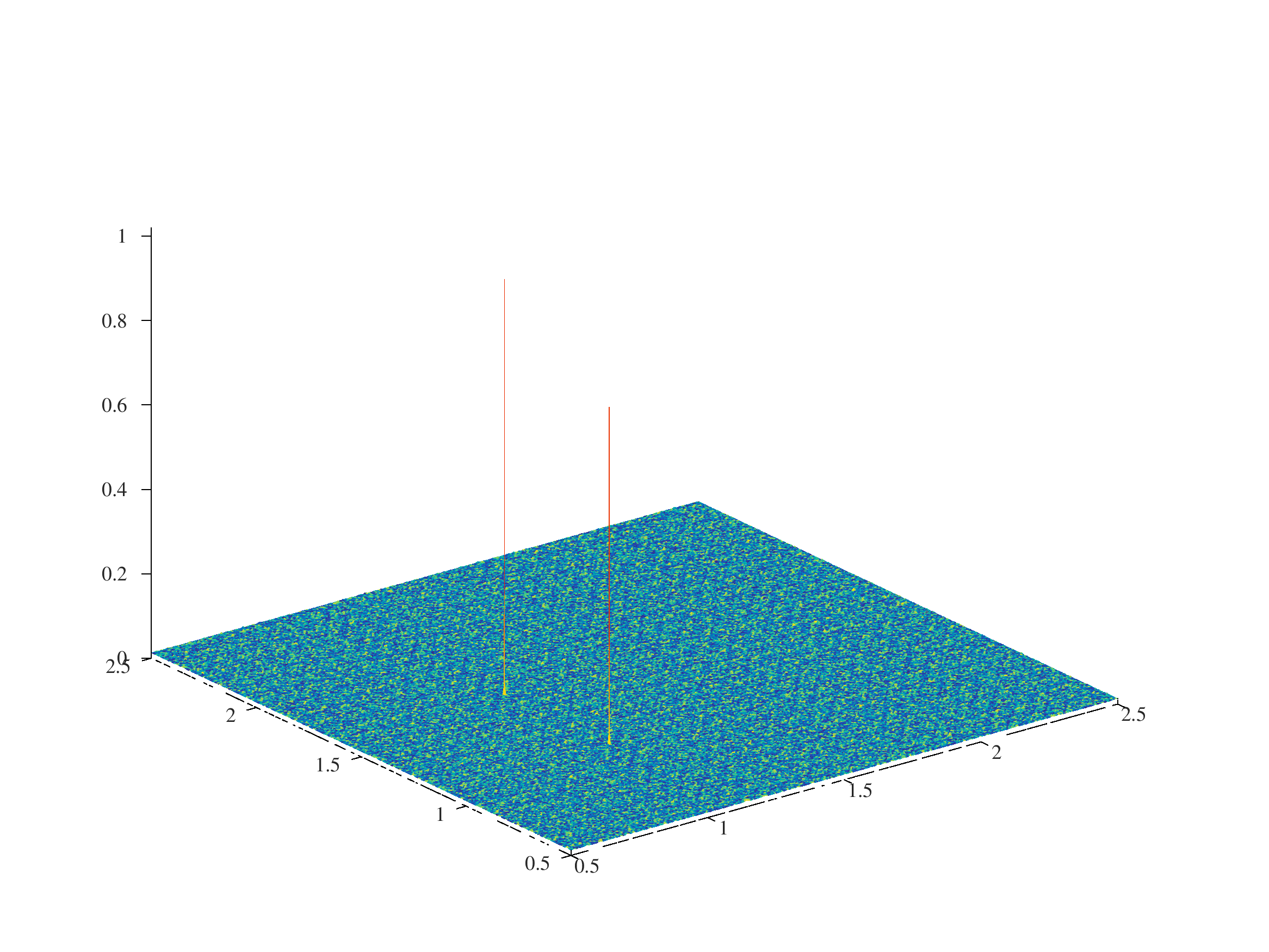}
        \caption*{(b) Spectrum ($\sigma=0.1$)}
    \end{minipage}
    \caption{Two-dimensional MUSIC spectrum and the noise space based on a random GHM of size $20\times 20$ for the source locations $\{(1.1,1.1),(1.1,1.6)\}$.}
    \label{fig:falsePeaksRandGHM2D}
\end{figure}

\subsubsection{Grid-free high-dimensional random-GHM MUSIC localization}

Although a random GHM avoids the exponential growth of the data-matrix dimension, direct MUSIC peak search still requires evaluating the spectrum on a high-dimensional tensor-product grid. The number of grid points therefore remains exponential in $d$. To remove this remaining dimensional bottleneck, we introduce a grid-free localization algorithm. The algorithm first obtains suitable initial estimates from the intersection of the estimated signal subspace with the constant-modulus steering manifold, followed by full-dimensional phase decoding. Starting from these estimates, Newton-type local iterations locate the peaks of the MUSIC spectrum.

Write $\mathcal W=\{\boldsymbol\omega_1,\ldots,\boldsymbol\omega_{M_1}\}$ and $\mathcal Z=\{\boldsymbol\zeta_1,\ldots,\boldsymbol\zeta_{M_2}\}$ for the realized row- and column-frequency sets. Using the signal-subspace matrix $\widetilde U_1$ in \eqref{eq:music-perturbed-svd}, constant-modulus representatives of the signal subspace are obtained from multiple random starts of
\[
F_{\mathrm{CM}}(\mathbf q)
:=\sum_{p=1}^{M_1}
\left(M_1\left|(\widetilde U_1\mathbf q)_p\right|^2-1\right)^2,
\qquad \|\mathbf q\|_2=1.
\]
For phase decoding, let $D_{\mathcal W}\in\mathbb Z^{(M_1-1)\times d}$ have rows $(\boldsymbol\omega_p-\boldsymbol\omega_1)^\top$, $p=2,\ldots,M_1$. Lattice reduction yields a short integer combination matrix $C\in\mathbb Z^{d\times(M_1-1)}$ such that $B:=CD_{\mathcal W}\in\mathbb Z^{d\times d}$ is nonsingular and well-conditioned. For a normalized subspace representative $\mathbf z$, set
\[
\boldsymbol\theta(\mathbf z):=
\bigl(\arg(z_p\overline{z_1})\bigr)_{p=2}^{M_1},
\qquad
\boldsymbol\eta(\mathbf z):=
\arg\!\left(\exp\!\left(iC\boldsymbol\theta(\mathbf z)\right)\right),
\]
where the argument and exponential in the definition of $\boldsymbol\eta$ are applied entrywise. Its aliases in the prescribed search region $\mathcal D$ are
\[
\mathcal L(\mathbf z):=
\left\{B^{-1}\bigl(\boldsymbol\eta(\mathbf z)+2\pi\mathbf k\bigr)\in\mathcal D:
\mathbf k\in\mathbb Z^d\right\}.
\]
For $r=1,\ldots,N_{\mathrm{CM}}$, draw $\mathbf q_r^{(0)}$ uniformly from $\mathbb S^{2n-1}$ and apply constrained local optimization to $F_{\mathrm{CM}}$ on the unit sphere. Denote the resulting point by $\mathbf q_r$ and set
\[
\mathbf z_r
=\frac{\widetilde U_1\mathbf q_r}{\|\widetilde U_1\mathbf q_r\|_2}.
\]
Starting from each point in $\mathcal L(\mathbf z_r)$, $r=1,\ldots,N_{\mathrm{CM}}$, Newton-type local iterations are applied to $\frac{1}{J_\sigma^2(\mathbf y)}$ over $\mathcal D$. Among the refined locations, those corresponding to the $n$ largest MUSIC spectral peaks form the recovered set $\widehat{\mathcal X}$.

We evaluate the method for $d=1,\ldots,8$ with $n=4$, cutoff frequency $\Omega=809$, and $\sigma=10^{-3}$. For each $d$, two centers $\mathbf c_1\in[-0.9,-0.4]^d$ and $\mathbf c_2\in[0.4,0.9]^d$ and two unit vectors $\mathbf v_1,\mathbf v_2\in\mathbb S^{d-1}$ are drawn once, giving the two-clump configuration
\begin{equation}
\label{eq:two-clump-configuration}
\mathcal X_\Delta
=\bigcup_{a=1}^2
\left\{\mathbf c_a-\frac{\Delta}{2}\mathbf v_a,
\mathbf c_a+\frac{\Delta}{2}\mathbf v_a\right\}.
\end{equation}
For the constant-modulus initialization, $N_{\mathrm{CM}}=64$ random starting points are optimized.

\begin{table}[htbp]
\centering
\small
\setlength{\tabcolsep}{7pt}
\begin{tabular}{c|c|ccc}
$d$ & randGHM $M$ & Total time (ms) & Resolution & RMSE ($\Delta=0.1$) \\ \hline
$1$ & $20$ & $47.53$  & $3.8\times10^{-3}$ & $1.01\times10^{-7}$ \\
$2$ & $24$ & $55.78$  & $3.3\times10^{-3}$ & $1.26\times10^{-7}$ \\
$3$ & $28$ & $77.08$  & $3.4\times10^{-3}$ & $1.18\times10^{-7}$ \\
$4$ & $32$ & $113.52$ & $4.7\times10^{-3}$ & $1.34\times10^{-7}$ \\
$5$ & $36$ & $175.72$ & $3.9\times10^{-3}$ & $1.32\times10^{-7}$ \\
$6$ & $40$ & $255.44$ & $4.5\times10^{-3}$ & $1.27\times10^{-7}$ \\
$7$ & $44$ & $396.11$ & $6.3\times10^{-3}$ & $1.26\times10^{-7}$ \\
$8$ & $64$ & $894.70$ & $4.5\times10^{-3}$ & $8.56\times10^{-8}$ \\
\end{tabular}
\caption{Total time, resolution, and recovery accuracy of the proposed grid-free random-GHM MUSIC localization method for the two-clump configuration in~\eqref{eq:two-clump-configuration}. Total time is averaged over $100$ independent trials.}
\label{table:RandomGHMLocationDimension}
\end{table}

The full-GHM MUSIC method is also compared in Table~\ref{table:FullGHMLocationDimension} using the same grid-free peak-localization procedure. The results illustrate the curse of dimensionality: even with FFT-based implicit matrix-vector products in place of explicit matrix construction, the memory requirement for $d\geq 3$ is prohibitive on a personal computer.

\begin{table}[!htbp]
\centering
\small
\setlength{\tabcolsep}{4pt}
\begin{tabular}{c|cc|ccc}
$d$ & $N_d$ & Memory & Total time (ms) & Resolution & RMSE ($\Delta=0.1$) \\ \hline
$1$ & $810$ & $25.3\,\mathrm{KiB}$ & $361.89$ & $3.7\times10^{-3}$ & $3.65\times10^{-8}$ \\
$2$ & $656100$ & $40.0\,\mathrm{MiB}$ & $2755.13$ & $3.1\times10^{-3}$ & $1.81\times10^{-9}$ \\
$3$ & $5.3\times10^{8}$ & $63\,\mathrm{GiB}$ & OOM & $-$ & $-$ \\
$4$ & $4.3\times10^{11}$ & $100\,\mathrm{TiB}$ & OOM & $-$ & $-$ \\
$5$ & $3.5\times10^{14}$ & $160\,\mathrm{PiB}$ & OOM & $-$ & $-$ \\
$6$ & $2.8\times10^{17}$ & $250\,\mathrm{EiB}$ & OOM & $-$ & $-$ \\
$7$ & $2.3\times10^{20}$ & $400\,\mathrm{ZiB}$ & OOM & $-$ & $-$ \\
$8$ & $1.9\times10^{23}$ & $620\,\mathrm{YiB}$ & OOM & $-$ & $-$ \\
\end{tabular}
\caption{Grid-free full-GHM MUSIC localization with $\Omega=809$ for the two-clump configuration in~\eqref{eq:two-clump-configuration}, where the full GHM is of size $N_d\times N_d$ with $N_d=810^d$. Memory is the storage of one double-precision complex FFT convolution grid, namely $16(2\Omega+1)^d$ bytes. Total time is averaged over $100$ independent trials.}
\label{table:FullGHMLocationDimension}
\end{table}

\FloatBarrier
\section{Conclusion and future work}\label{sec:conclusionandfuture}

In this work, we developed a generalized Hankel/Toeplitz framework for multi-dimensional super-resolution from nonuniform Fourier measurements by taking the underlying Vandermonde decomposition, rather than the rigid Hankel/Toeplitz pattern, as the fundamental algebraic structure. This framework unifies classical Hankel, Toeplitz, and multi-level constructions while allowing substantially more flexible sampling geometries and matrix dimensions that scale with the intrinsic degrees of freedom rather than with a full tensor-product grid. On the theoretical side, we established explicit multi-dimensional sufficient separation conditions and computational resolution bounds for source-number detection in the multi-dimensional super-resolution problem. For segmented frequency sets, which model sparse and distributed sampling geometries, we derived deterministic lower bounds for the minimum singular values of the associated generalized Vandermonde matrices and corresponding thresholding guarantees for both single-clump and multi-clump configurations. To overcome the dimensional and computational bottlenecks of conventional multi-level Hankel matrices, we further introduced randomized GHM constructions whose computational cost scales with the effective information content, together with conditional recovery guarantees expressed directly in terms of the minimum singular values of the realized Vandermonde factors. We also extended the generalized framework from source-number detection to source localization by developing a GHM-based MUSIC method for nonuniform Fourier measurements, whose stability under noise is controlled by Vandermonde conditioning and whose computational burden can be reduced through partial singular-value decomposition and efficient peak-search strategies. Taken together, these results establish a unified connection among sampling geometry, Vandermonde conditioning, computational resolution, number detection, and subspace localization, and demonstrate that the generalized Vandermonde decomposition provides a flexible and scalable foundation for extending classical Hankel/Toeplitz-based methods to sparse, segmented, randomized, and high-dimensional sensing scenarios.

This work also opens the door to many new problems that have important applications. Firstly, the numerical experiments suggest that the randomized generalized Hankel
matrix (GHM) performs remarkably well not only for a cluster of closely spaced
sources, but also for sources distributed over a relatively large imaging
region. While the super-resolution capability for a single cluster of
closely spaced sources has been elucidated in
Theorem~\ref{thm:resolutionrandghmnumber1}, the mechanism underlying the favorable performance
for well-separated or multi-cluster sources is still not completely
understood. This observation motivates understanding why randomized GHMs can reliably
distinguish sources distributed over a large region. This can be transformed into the minimum singular value estimation of a generalized Vandermonde matrix (random Fourier matrix), which is related to the compressive sensing theory and random matrix theory. Our second objective is to establish a more complete theory for the
resolution limit of super-resolution from random measurements. Such a theory should provide both lower and upper bounds for the resolution
limit, as well as quantitative stability estimates for the recovered
locations and amplitudes. It should also reveal how randomization affects these different regimes and
whether random sampling can achieve nearly the same resolution as complete
Fourier measurements with substantially fewer measurements. Beyond information-theoretic resolution limits, it is equally important to
characterize the performance of concrete reconstruction algorithms. Such results would make it possible to compare
different generalized Hankel/Toeplitz constructions from the perspectives
of resolution, robustness, sample complexity, and computational complexity. Our third direction is motivated by the observation that the proposed
generalized framework is not specific to one super-resolution algorithm.
A large family of classical methods in spectral estimation, array signal
processing, system identification, and inverse problems are based on
Hankel or Toeplitz matrices. Representative examples include Prony-type
methods, matrix-pencil methods, MUSIC, ESPRIT, and structured low-rank
approximation. The fourth direction is to consider the scattering problem and bistatic array signal processing, which can also be analyzed based on the generalized Hankel/Toeplitz matrix framework, where new theoretical and algorithmic problems arise. 

\section*{Acknowledgments}
This work was partially supported by the Fundamental and Interdisciplinary Disciplines Breakthrough Plan of the Ministry of Education
of China grant number JYB2025XDXM103 and the National Key R\&D Program of China grant number 2024YFA1016000.

\section*{Data Availability Statement}
Data and codes supporting the findings of this work are available upon request.

\section*{Conflict of interest} 
The authors have no conflicts of interest to declare. 

\appendix

\section{Minimum singular value of multivariate Vandermonde matrices}\label{sec:appendix-li}

This appendix is devoted to the proof of the minimum singular value estimate for contiguous (single-clump) multivariate Vandermonde matrices stated in Lemma~\ref{lem:uniform-Vandermonde}, together with the auxiliary results needed for the threshold theorem Theorem~\ref{liuthm5.1v2} and hence for the resolution limit in Theorem~\ref{thm:li-resolution}. In \cite{li2025nonharmonic}, Li studied the conditioning of nonharmonic multivariate Fourier matrices under general geometric assumptions on the node set; his estimates are governed by the maximum local cluster size and by hyperplane geometry, and they hold uniformly over multi-clump configurations. Here we specialize to the single-clump setting, in which all sources are confined to the ball $B_{\frac{\pi n}{\Omega},1}^d(\mathbf0)$ and the frequency set is a contiguous cube, and we construct explicit vanishing trigonometric polynomials to obtain a sharper, fully explicit lower bound. This improved conditioning is exactly what yields the better resolution, and the corresponding results are formally given in Theorem~\ref{thm:li-resolution}. The multi-clump case is treated separately in Appendix~\ref{sec:appendix-segmented-Vandermonde}.

Throughout this appendix we follow the notation of Subsection~\ref{subsec:uniform array}: $\mathbf V_1$ and $\mathbf \Sigma$ are defined there (see \eqref{eq:factorization_H}), and we assume, as in Lemma~\ref{lem:uniform-Vandermonde}, that $s$ is even. For $h\ge0$, let
\[
\overline Q_h^{\,d}:=[-h,h]^d
\]
denote the closed frequency cube, and let $\mathcal P(\overline Q_h^{\,d})$ be the trigonometric polynomials on the unit torus $\mathbb T^d$ whose Fourier coefficients are supported in $\overline Q_h^{\,d}\cap\mathbb Z^d$. All $L^p(\mathbb T^d)$ norms below use normalized Haar measure; this frequency-set notation is distinct from the open spatial neighborhood $Q_h^d(\mathbf x)=B_{h,\infty}^d(\mathbf x)$ defined in Section~\ref{sec:number-detection}. We define the centered vector
\[
\varphi_s(t):=\bigl(t^{-s/2},t^{-s/2+1},\dots,t^{s/2}\bigr)^\top\in\mathbb C^{s+1}.
\]
Then $\psi_s(t)=t^{s/2}\varphi_s(t)$. Accordingly, for each node $\mathbf y_j=((\mathbf y_j)_1,\dots,(\mathbf y_j)_d)\in\mathbb R^d$, we introduce the centered Vandermonde columns
\(
(\widetilde{\mathbf V}_1)_j:=\varphi_s\!\left(e^{i(\mathbf y_j)_1\Omega/s}\right)\otimes\cdots\otimes \varphi_s\!\left(e^{i(\mathbf y_j)_d\Omega/s}\right)
\), and write $\widetilde{\mathbf V}_1=[(\widetilde{\mathbf V}_1)_1,\dots,(\widetilde{\mathbf V}_1)_n]$.

With the rescaled nodes
\[
\mathbf t_j:=\frac{\Omega}{2\pi s}\mathbf y_j,\qquad j=1,\dots,n,
\]
the matrix $\widetilde{\mathbf V}_1$ is precisely the generalized Vandermonde matrix associated with the frequency set $\overline Q_{s/2}^{\,d}\cap\mathbb Z^{d}$ and the nodes $\{\mathbf t_j\}_{j=1}^{n}\subset\mathbb T^d$; under the lexicographic ordering of this frequency set,
\[
\widetilde{\mathbf V}_1=\bigl[e^{2\pi i\,\boldsymbol\omega\cdot \mathbf t_j}\bigr]_{\boldsymbol\omega\in \overline Q_{s/2}^{\,d}\cap\mathbb Z^{d},\,1\le j\le n},
\qquad
\Phi:=\bigl[e^{-2\pi i\,\boldsymbol\omega\cdot \mathbf t_j}\bigr]_{\boldsymbol\omega\in \overline Q_{s/2}^{\,d}\cap\mathbb Z^{d},\,1\le j\le n},
\]
so that $\widetilde{\mathbf V}_1=\overline{\Phi}$.

\begin{lem}\label{lem:nonnegative_to_centered}
Recall that \(\mathbf V_1\) is the Vandermonde matrix in \eqref{eq:factorization_H}. There exists a diagonal unitary matrix
\[
\mathbf U:=\operatorname{diag}\!\left(
e^{\frac{i \Omega}{2}\sum_{q=1}^{d}(\mathbf y_1)_q},
\dots,
e^{\frac{i \Omega}{2}\sum_{q=1}^{d}(\mathbf y_n)_q}
\right)
\]
such that $\mathbf V_1=\widetilde{\mathbf V}_1\mathbf U$. Consequently,
\[
\sigma_k(\mathbf V_1)=\sigma_k(\widetilde{\mathbf V}_1),
\qquad k=1,\dots,n.
\]
Moreover,  since $\boldsymbol 1^\top\mathbf y_j=\sum_{q=1}^{d}(\mathbf y_j)_q$, we have
\[
\mathbf V_1\mathbf\Sigma \mathbf V_1^\top
=
\widetilde{\mathbf V}_1\operatorname{diag}(a_1,\dots,a_n) \widetilde{\mathbf V}_1^\top.
\]
\end{lem}

\begin{proof}
Since $\psi_s(t)=t^{s/2}\varphi_s(t)$, for each $j$ we have
\(
(\mathbf V_1)_j
=
e^{\frac{i\Omega}{2}\sum_{q=1}^{d}(\mathbf y_j)_q}(\widetilde{\mathbf V}_1)_j
\). This gives $\mathbf V_1=\widetilde{\mathbf V}_1\mathbf U$. Since $\mathbf U$ is diagonal unitary, the singular values are preserved. The last identity follows from the fact that $\mathbf U\mathbf\Sigma\mathbf U=\operatorname{diag}(a_1,\dots,a_n)$.
\end{proof}

The following lemma is the specialization of \cite[Lemma 3.8]{li2025nonharmonic} to the $\ell^\infty$ frequency cube and the dual $\ell^1$ spatial metric.

\begin{lem}\label{lem2:uniform-Vandermonde}
Let \( \mathcal{U} \subseteq [-\frac{1}{2}, \frac{1}{2})^d \) be a finite set of at most \( r \) elements such that \( \mathbf 0 \in \mathcal{U} \) and \( \|\mathbf u\|_{1} \leq \frac{1}{4} \) for each \( \mathbf u \in \mathcal{U} \). For any real \( h \geq 2r \), there exists \( f \in \mathcal{P}(\overline{Q}_{h(r-1)/r}^{\,d}) \) such that \( f(\mathbf 0)=1 \), \( f \) vanishes on \( \mathcal{U}\setminus\{\mathbf 0\} \), and
\[
\|f\|_{L^\infty(\mathbb T^d)}
\leq
\sqrt{2^{|\mathcal U|-1}}
\prod_{\substack{\mathbf u\in\mathcal U\\0<\|\mathbf u\|_{1}\leq \frac{r}{2h}}}
\frac{r}{2h\|\mathbf u\|_{1}}.
\]
\end{lem}

With this auxiliary polynomial construction in hand, we now prove Lemma~\ref{lem:uniform-Vandermonde}.

\begin{proof}[Proof of Lemma \ref{lem:uniform-Vandermonde}]
Together with Lemma~\ref{lem:nonnegative_to_centered}, this gives
\(
\sigma_{\min}(\mathbf V_1)
=\sigma_{\min}(\widetilde{\mathbf V}_1)
\). It therefore suffices to estimate
$\sigma_{\min}(\widetilde{\mathbf V}_1)$. 
If $\theta_{\min}(\Omega,n)=0$, \eqref{eq:uniform-Vandermonde} holds trivially. Hence assume $\theta_{\min}(\Omega,n)>0$; in particular, the nodes $\mathbf t_1,\dots,\mathbf t_n$ are pairwise distinct.
Since $\mathbf y_j\in B_{\frac{\pi n}{\Omega},1}^d(\mathbf{0})$, we have
\[
\mathbf t_j\in B_{n/(2s),\,1}^{d}(\mathbf 0)\subseteq Q_{\frac{n}{2s}}^d(\mathbf 0) \subseteq Q_{1/8}^{d}(\mathbf 0),
\]
so every coordinate difference belongs to $(-1/4,1/4)$ and the chosen
representatives realize the periodic $\ell^1$-distance on the unit torus.

For each $k$, set $\mathcal U_k:=\{\mathbf t_j-\mathbf t_k:1\le j\le n\}$ and apply Lemma \ref{lem2:uniform-Vandermonde} with $\mathcal U=\mathcal U_k$, $h=s/2$, and $r=n$. Since
\[
\|\mathbf t_j-\mathbf t_k\|_1\le \frac{n}{s}\le \frac14
\qquad\text{and}\qquad
\frac{s}{2}\ge 2n,
\]
the assumptions are satisfied. The cited lemma gives $g_k\in\mathcal P(\overline Q_{s(n-1)/(2n)}^{\,d})$ such that $g_k(\mathbf0)=1$ and $g_k(\mathbf t_j-\mathbf t_k)=0$ for $j\ne k$. Define
\[
b_k(\mathbf x):=g_k(\mathbf x-\mathbf t_k).
\]
Translation preserves the Fourier support and the $L^\infty(\mathbb T^d)$ norm. Hence $b_k(\mathbf t_j)=\delta_{jk}$ and
\[
\|b_k\|_{L^\infty(\mathbb T^d)}\leq
\sqrt{2^{n-1}}
\prod_{\substack{1\le j\le n\\j\ne k}}
\frac{n}{s\|\mathbf t_j-\mathbf t_k\|_1}.
\]
Since $
\|\mathbf t_j-\mathbf t_k\|_1
=
\frac{\Omega}{2\pi s}\|\mathbf y_j-\mathbf y_k\|_1
=
\frac{n}{s}\cdot \frac{\Omega}{2n\pi}\|\mathbf y_j-\mathbf y_k\|_1$, define
\[
\theta_{jk}:=\frac{\Omega}{2n\pi}\|\mathbf y_j-\mathbf y_k\|_1.
\]
Every factor in the preceding product equals $\theta_{jk}^{-1}$, and $\theta_{jk}\ge\theta_{\min}(\Omega,n)$. Therefore
\[
\|b_k\|_{L^\infty}
\le
\sqrt{2^{n-1}}
\prod_{\substack{1\le j\le n\\j\ne k}}\theta_{jk}^{-1}
\le
\sqrt{2^{n-1}}\bigl(\theta_{\min}(\Omega,n)\bigr)^{-(n-1)}.
\]
Now let
\[
h_0(\mathbf x):=\frac{1}{\left(2\left\lfloor\frac{s}{2n}\right\rfloor+1\right)^d}
\sum_{\boldsymbol\omega\in \overline Q_{s/(2n)}^{\,d}\cap\mathbb Z^d} e^{2\pi i \boldsymbol\omega\cdot \mathbf x}.
\]
Then $h_0\in\mathcal P(\overline Q_{s/(2n)}^{\,d})$, $h_0(\mathbf 0)=1$, and
\[
\|h_0\|_{L^2}=\left(2\left\lfloor\frac{s}{2n}\right\rfloor+1\right)^{-d/2}.
\]

For each \(k\in\{1,\dots,n\}\), set \(f_k:=h_0(\cdot-\mathbf t_k)b_k\). Since translation does not change the Fourier support and the Fourier transform of a product is the convolution of the Fourier transforms, we have
\[
\operatorname{supp}(\widehat{f_k})\subseteq
\operatorname{supp}\!\bigl(\widehat{h_0}\bigr)
+
\operatorname{supp}(\widehat{b_k})
\subseteq
\overline Q_{s/(2n)}^{\,d}+\overline Q_{s(n-1)/(2n)}^{\,d}
=
\overline Q_{s/2}^{\,d}.
\]
Hence \(f_k\in \mathcal P(\overline Q_{s/2}^{\,d})\) for all \(k=1,\dots,n\). Moreover, \(f_k(\mathbf t_j)=\delta_{jk}\), and
\[
\|f_k\|_{L^2}\le \|h_0\|_{L^2}\|b_k\|_{L^\infty}
\le
\frac{\sqrt{2^{n-1}}}{\sqrt{\left(2\left\lfloor\frac{s}{2n}\right\rfloor+1\right)^d}}
\bigl(\theta_{\min}(\Omega,n)\bigr)^{-(n-1)}.
\]
The conjugated polynomials $\overline{f_k}$ remain supported in $\overline Q_{s/2}^{\,d}$, satisfy $\overline{f_k}(\mathbf t_j)=\delta_{jk}$, and have the same $L^2$ norms. Applying the trigonometric-polynomial duality principle \cite[Lemma~3.3]{li2025nonharmonic} to \(\Phi\) with these conjugated interpolants gives
\[
\sigma_{\min}(\Phi)
\ge
\frac{1}{\sqrt n}\min_{1\le k\le n}\|f_k\|_{L^2}^{-1},
\]
which yields \eqref{eq:uniform-Vandermonde}, since  $\widetilde{\mathbf V}_1=\overline\Phi$. Lemma~\ref{lem:nonnegative_to_centered} transfers the estimate to $\mathbf V_1$.
\end{proof}


\section{Minimum singular value of multivariate segmented Vandermonde matrices}\label{sec:appendix-segmented-Vandermonde}

In this appendix, we provide the detailed constructive proof for the minimum singular value lower bound presented in Theorem~\ref{thm:segmented-vandermonde}. The analytical techniques utilized here draw inspiration from \cite{li2025nonharmonic}, adapted to the multi-cluster equally distributed array geometry. We follow the notation of Subsection~\ref{subsec:equally_distributed} and Definitions~\ref{defi:metric_separation}--\ref{defi:high_dim_clumps}. We also recall the generalized Vandermonde matrix $\mathcal{V}_{\Gamma}(\mathcal{Z})$ from Definition~\ref{generalized vandermonde} and write all Vandermonde matrices in this appendix in this unified form; in particular, $\mathcal{V}_{\Lambda^d}(\mathcal{X})$ is often denoted simply as $\mathcal{V}$ for brevity. With this notation established, we now introduce several key definitions and technical lemmas required to complete the proof.

\begin{defi}[Multivariate Segmented Trigonometric Polynomials]
    \label{defi:high_dim_uniform_poly}
    Let $d,D\in\mathbb N$ and $m,r\in\mathbb Z_{\ge0}$ satisfy $D>m$. We consider the uniform segmented index set $\Lambda$ defined as
    \[
    \Lambda = \bigcup_{k=0}^r \{kD, kD+1, \dots, kD+m\}.
    \]
    The high-dimensional sampling set is the Cartesian product $\Lambda^d \subset \mathbb{Z}^d$.
    
    We define $\mathcal{P}(m, r, D, d)$ as the space of trigonometric polynomials supported on $\Lambda^d$. Any function $f \in \mathcal{P}(m, r, D, d)$ takes the form
    \[
    f(\bm{\omega}) = \sum_{\mathbf{s} \in \{0,\dots,r\}^d} \quad \sum_{\mathbf{h} \in \{0,\dots,m\}^d} c_{\mathbf{s}, \mathbf{h}} \, e^{2\pi i (D\mathbf{s} + \mathbf{h}) \cdot \bm{\omega}},
    \]
    where $\mathbf{s} = (s_1, \dots, s_d)$, $\mathbf{h} = (h_1, \dots, h_d)$ and the term $(D\mathbf{s} + \mathbf{h})$ represents the frequency vector $(D s_1 + h_1, \dots, D s_d + h_d)$.
\end{defi}

\begin{defi}
    \label{defi:high_dim_lagrange}
    Let $\mathcal{X}=\{\mathbf{y}_1, \dots, \mathbf{y}_n\} \subset (-\pi, \pi]^d$ be the set of nodes. We say that $\{f_k\}_{k=1}^n \subset \mathcal{P}(m, r, D, d)$ is a family of \textit{Lagrange interpolants} for $\mathcal{X}$ if for all $1 \le k, \ell \le n$, the following condition holds:
    \begin{equation}
        f_k\left(\frac{\mathbf{y}_\ell}{2\pi}\right) = \delta_{k,\ell} = 
        \begin{cases} 
            1 & \text{if } k = \ell, \\
            0 & \text{if } k \neq \ell.
        \end{cases}
    \end{equation}
\end{defi}

\begin{remark}
    The scaling factor $\frac{1}{2\pi}$ in the argument of $f_k$ is necessary to match the frequency definitions. The polynomial $f_k$ is defined with the kernel $e^{2\pi i \mathbf{k} \cdot \bm{\omega}}$, while the segmented Vandermonde matrix $\mathcal{V}$ is constructed using $e^{i \mathbf{k} \cdot \mathbf{y}}$. Evaluating at $\bm{\omega} = \frac{\mathbf{y}}{2\pi}$ yields $e^{2\pi i \mathbf{k} \cdot (\mathbf{y}/2\pi)} = e^{i \mathbf{k} \cdot \mathbf{y}}$, which aligns the polynomial evaluation with the matrix-vector multiplication.
\end{remark}

\begin{lem}
    \label{lem:minsvd_bound_by_lagInterp_high_dim}
    Let $\mathcal{V}=\mathcal{V}_{\Lambda^d}(\mathcal{X})$ be the Vandermonde matrix defined on the frequency set $\Lambda^d$ and nodes $\mathcal{X}$. If there exists a family of Lagrange interpolants $\{f_k\}_{k=1}^n \subset \mathcal{P}(m, r, D, d)$ for $\mathcal{X}$, then
    \begin{equation}
        \frac{1}{\sigma_{\min}(\mathcal{V})} \leq \left(\sum_{k=1}^n \|f_k\|_{L^2(\mathbb T^d)}^2\right)^{1/2}.
    \end{equation}
    Here, the $L^2$ norm is defined on the unit torus $\mathbb T^d \cong [0, 1)^d$ as $\|f\|_{L^2}^2 = \int_{\mathbb T^d} |f(\bm{\omega})|^2 d\bm{\omega}$.
\end{lem}

\begin{proof}
    Let $\mathbf{c}_k\in\mathbb{C}^{|\Lambda^d|}$ be the coefficient vector of $f_k$ and set $C=[\mathbf{c}_1,\ldots,\mathbf{c}_n]$. The interpolation identities are equivalent to $C^T \mathcal{V}=I_n$. Hence $\mathcal{V}$ has full column rank. For every $\mathbf{x}\in\mathbb{C}^n$,
    \[
        \|\mathbf{x}\|_2
        =\|C^T\mathcal{V}\mathbf{x}\|_2
        \le \|C^T\|_2\|\mathcal{V}\mathbf{x}\|_2
        =\|C\|_2\|\mathcal{V}\mathbf{x}\|_2.
    \]
    Taking the infimum over unit vectors and using Parseval's identity yields
    \[
        \frac{1}{\sigma_{\min}(\mathcal{V})}
        \le \|C\|_2
        \le \|C\|_F
        =\left(\sum_{k=1}^n\|\mathbf{c}_k\|_2^2\right)^{1/2}
        =\left(\sum_{k=1}^n\|f_k\|_{L^2(\mathbb T^d)}^2\right)^{1/2}.
    \]
\end{proof}

The following lemma gives a minimum-norm interpolant whenever the sampling Vandermonde matrix has full column rank.

\begin{lem}
    \label{lem:interpolation_via_svd}
 Let $m,r,D,d$ and $\Lambda$ be as in Definition~\ref{defi:high_dim_uniform_poly}. Let $\mathcal{X}=\{\mathbf{y}_1, \dots, \mathbf{y}_n\} \subset (-\pi, \pi]^d$, and assume that the associated Vandermonde matrix $\mathcal{V}_{\Lambda^d}(\mathcal{X})$ has full column rank. Then for any vector $\mathbf w \in \mathbb{C}^n$, there exists a polynomial $f \in \mathcal{P}(m, r, D, d)$ satisfying the interpolation condition
    \begin{equation}
        f\left(\frac{\mathbf{y}_j}{2\pi}\right) = w_j, \quad \text{for all } j=1, \dots, n,
    \end{equation}
    and the following norm bounds
    \begin{equation}
        \|f\|_{L^2(\mathbb T^d)} \le \frac{\|\mathbf w\|_2}{\sigma_{\min}(\mathcal{V}_{\Lambda^d}(\mathcal{X}))}
        \quad \text{and} \quad 
        \|f\|_{L^\infty(\mathbb T^d)} \le \frac{\sqrt{|\Lambda^d|} \, \|\mathbf w\|_2}{\sigma_{\min}(\mathcal V_{\Lambda^d}(\mathcal{X}))},
    \end{equation}
    where $|\Lambda^d| = [(r+1)(m+1)]^d$.
\end{lem}

\begin{proof}
    Set
    \[
        A:=\mathcal{V}_{\Lambda^d}(\mathcal{X})^T
        \in\mathbb{C}^{n\times|\Lambda^d|},
        \qquad
        \mathbf{c}:=A^\dagger\mathbf{w}.
    \]
    The full-column-rank assumption on $\mathcal{V}_{\Lambda^d}(\mathcal{X})$ implies that $A$ has full row rank, and hence $AA^\dagger=I_n$. Using $\mathbf{c}$ as the coefficient vector of $f$ gives $A\mathbf{c}=\mathbf{w}$ and thus the stated interpolation values. Moreover, Parseval's identity gives
    \[
        \|f\|_{L^2(\mathbb T^d)}=\|\mathbf{c}\|_2
        \le \|A^\dagger\|_2\|\mathbf{w}\|_2
        =\frac{\|\mathbf{w}\|_2}
        {\sigma_{\min}(\mathcal{V}_{\Lambda^d}(\mathcal{X}))}.
    \]
    Here $A=\mathcal V_{\Lambda^d}(\mathcal X)^T$ has the same singular values as $\mathcal V_{\Lambda^d}(\mathcal X)$.
    Finally,
    \[
        \|f\|_{L^\infty(\mathbb T^d)}
        \le \|\mathbf{c}\|_1
        \le \sqrt{|\Lambda^d|}\,\|\mathbf{c}\|_2,
    \]
    which proves the second estimate.
\end{proof}

\begin{thm}
    \label{thm:well_separated_segmented}
    Let $d\ge1$, $m\ge1$, $r\ge0$, and $D\ge m+1$ be integers, and let
    \[
        \Lambda=\bigcup_{s=0}^r\{sD,sD+1,\ldots,sD+m\}.
    \]
    Let $\beta\ge 1/(2\log2)$ and let
    $\mathcal X=\{\mathbf y_1,\ldots,\mathbf y_n\}\subset(-\pi,\pi]^d$
    be a finite non-empty set. If $n\ge2$, assume
    \begin{equation}
        \label{eq:separation_condition_global}
        \Delta_\infty(\mathcal X)\ge\frac{4\pi\beta d}{m+1}.
    \end{equation}
    Then, with
    \[
        a_\beta:=2-e^{1/(2\beta)},
        \qquad
        b_\beta:=e^{1/(2\beta)},
    \]
    the Vandermonde matrix satisfies
    \begin{equation}
        \label{eq:well_separated_segmented_frame}
        a_\beta|\Lambda^d|
        \le\sigma_{\min}^2(\mathcal V_{\Lambda^d}(\mathcal X))
        \le\sigma_{\max}^2(\mathcal V_{\Lambda^d}(\mathcal X))
        \le b_\beta|\Lambda^d|,
        \qquad
        |\Lambda^d|=[(r+1)(m+1)]^d.
    \end{equation}
\end{thm}

\begin{proof}
    Set \(N:=m+1\) and \(\mathcal B_0:=\{0,\ldots,m\}^d\). If \(n=1\),
    the result follows from the column norm. We therefore assume \(n\ge2\).

    Following the argument in \cite[Theorem~2.3]{li2025nonharmonic},
    for every \(\mathbf c\in\mathbb C^n\), we have
    \begin{equation}
        \label{eq:base_block_bound}
        a_\beta N^d\|\mathbf c\|_2^2
        \le
        \|\mathcal V_{\mathcal B_0}(\mathcal X)\mathbf c\|_2^2
        \le
        b_\beta N^d\|\mathbf c\|_2^2.
    \end{equation}

    Since
    \[
        \Lambda^d
        =\bigsqcup_{\mathbf s\in\{0,\ldots,r\}^d}
        (D\mathbf s+\mathcal B_0),
    \]
    the block indexed by \(\mathbf s\) replaces \(c_j\) by
    \(c_j e^{iD\mathbf s\cdot\mathbf y_j}\), without changing
    \(\|\mathbf c\|_2\). Summing \eqref{eq:base_block_bound} over the
    \((r+1)^d\) blocks gives
    \[
        a_\beta|\Lambda^d|\|\mathbf c\|_2^2
        \le
        \|\mathcal V_{\Lambda^d}(\mathcal X)\mathbf c\|_2^2
        \le
        b_\beta|\Lambda^d|\|\mathbf c\|_2^2.
    \]
    Taking the infimum and supremum over unit vectors proves
    \eqref{eq:well_separated_segmented_frame}.
\end{proof}

For clumped nodes, the complement of a cluster can be decomposed into at most $n^\star$ well-separated classes. The following construction also applies to arbitrary subsets of the node set.

\begin{prop}
    \label{prop:decomposition}
    Let $\mathcal X=\bigcup_{a=1}^A\mathcal C_a$ form $(A,\infty,\tau,\eta,n^\star)$-clumps. Every non-empty subset $\mathcal Y\subseteq\mathcal X$ admits a partition
    \[
        \mathcal Y=\bigcup_{\ell=1}^{\nu}\mathcal Y_\ell,
        \qquad
        1\le\nu\le n^\star,
    \]
    into non-empty disjoint sets such that $|\mathbf y-\mathbf y'|_\infty>\eta$ for every pair of distinct nodes $\mathbf y,\mathbf y'\in\mathcal Y_\ell$.
\end{prop}

\begin{proof}
    Enumerate each non-empty set $\mathcal Y\cap\mathcal C_a$ as
    $\{\mathbf z_{a,1},\ldots,\mathbf z_{a,q_a}\}$ and let $\nu:=\max_a q_a$. For $1\le\ell\le\nu$, set
    \[
        \mathcal Y_\ell:=\{\mathbf z_{a,\ell}:q_a\ge\ell\}.
    \]
    Each $\mathcal Y_\ell$ contains at most one node from each clump. Hence any two of its nodes belong to distinct clumps and have periodic $\ell^\infty$-distance greater than $\eta$. Since $q_a\le n^\star$, we have $\nu\le n^\star$.
\end{proof}

\begin{lem}
    \label{lem:localization}
    Let $\mathcal X\subset(-\pi,\pi]^d$ form $(A,\infty,\tau,\eta,n^\star)$-clumps, and let $m,D\in\mathbb N$ satisfy $D>m$. Set
    \[
        K:=\left\lfloor\frac{m}{n^\star}\right\rfloor.
    \]
    If $\beta>1/(2\log2)$ and
    \begin{equation}
        \label{eq:bandwidth_condition_final}
        \eta\ge\frac{4\pi\beta d}{K+1},
    \end{equation}
    then, for every $\mathbf y_k\in\mathcal X$, there exists $g_k\in\mathcal P(m,0,D,d)$ such that
    \[
        g_k\left(\frac{\mathbf y_k}{2\pi}\right)=1,
        \qquad
        g_k\left(\frac{\mathbf y_j}{2\pi}\right)=0
        \quad\text{for }\mathbf y_j\notin\mathcal N_\infty(\mathbf y_k,\tau,\mathcal X),
    \]
    and
    \[
        \|g_k\|_{L^\infty(\mathbb T^d)}
        \le\left(2-e^{1/(2\beta)}\right)^{-n^\star/2}.
    \]
\end{lem}

\begin{proof}
    Let $\mathcal C_{a(k)}$ be the clump containing $\mathbf y_k$. The clump assumptions imply
    \[
        \mathcal N_\infty(\mathbf y_k,\tau,\mathcal X)=\mathcal C_{a(k)}.
    \]
    Set $\mathcal G_k:=\mathcal X\setminus\mathcal C_{a(k)}$. If $\mathcal G_k=\varnothing$, take $g_k=1$.

    Assume $\mathcal G_k\ne\varnothing$. Then $K\ge1$. Indeed, if $K=0$, condition \eqref{eq:bandwidth_condition_final} gives $\eta\ge4\pi\beta d>\pi$, whereas a node in a different clump would satisfy
    \[
        \eta<|\mathbf y-\mathbf y_k|_\infty\le\pi.
    \]

    By Proposition~\ref{prop:decomposition}, write $\mathcal G_k$ as the disjoint union of $\nu_k\le n^\star$ non-empty $\eta$-separated sets $\mathcal G_{k,1},\ldots,\mathcal G_{k,\nu_k}$. Since $\mathbf y_k$ and every node of $\mathcal G_k$ belong to distinct clumps, each set
    \[
        \mathcal W_{k,\ell}:=\mathcal G_{k,\ell}\cup\{\mathbf y_k\}
    \]
    is also $\eta$-separated.

    Let $\Omega_K:=\{0,\ldots,K\}^d$ and $a_\beta:=2-e^{1/(2\beta)}>0$. Theorem~\ref{thm:well_separated_segmented}, applied with $m=K$, $r=0$, and $D=K+1$, gives
    \[
        \sigma_{\min}^2\!\left(\mathcal V_{\Omega_K}(\mathcal W_{k,\ell})\right)
        \ge a_\beta(K+1)^d>0.
    \]
    Hence the Vandermonde matrix has full column rank. Lemma~\ref{lem:interpolation_via_svd} provides a polynomial $h_{k,\ell}$ supported on $\Omega_K$ such that
    \[
        h_{k,\ell}\left(\frac{\mathbf y_k}{2\pi}\right)=1,
        \qquad
        h_{k,\ell}\left(\frac{\mathbf y}{2\pi}\right)=0
        \quad(\mathbf y\in\mathcal G_{k,\ell}),
    \]
    and
    \[
        \|h_{k,\ell}\|_\infty
        \le\frac{\sqrt{|\Omega_K|}}
        {\sigma_{\min}(\mathcal V_{\Omega_K}(\mathcal W_{k,\ell}))}
        \le a_\beta^{-1/2}.
    \]

    Set $g_k:=\prod_{\ell=1}^{\nu_k}h_{k,\ell}$. Since $0<a_\beta<1$,
    \[
        \|g_k\|_\infty\le a_\beta^{-\nu_k/2}\le a_\beta^{-n^\star/2}.
    \]
    Moreover,
    \[
        \operatorname{supp}\widehat g_k
        \subseteq\{0,\ldots,K\nu_k\}^d
        \subseteq\{0,\ldots,Kn^\star\}^d
        \subseteq\{0,\ldots,m\}^d.
    \]
    Thus $g_k\in\mathcal P(m,0,D,d)$ and has the stated interpolation values.
\end{proof}

\begin{lem}
    \label{lem:freq_quantization}
    Let $d,D\in\mathbb N$, $p \in [1, \infty]$, and let $p'$ be the H\"older conjugate of $p$, with $1'=\infty$ and $\infty'=1$. Let $\mathbf{u} \in (-\pi, \pi]^d$ be the coordinatewise shortest representative of a nonzero difference in $\mathbb T_{2\pi}^d$. Throughout this lemma, $\|\cdot\|_q$ denotes the ordinary $\ell^q$-norm of a vector representative. Assume
    \[
    0 < \|\mathbf{u}\|_{p'} \le \frac{\pi}{2 Dd^{1/p}}.
    \]
    For any scale parameter $\alpha > 0$ such that $\|\mathbf{u}\|_{p'} \le 2\pi\alpha \le \frac{\pi}{2D d^{1/p}}$, there exists an integer frequency vector $\mathbf{k} \in \mathbb{Z}^d$ such that
    \begin{equation}
        \label{eq:freq_quant_bounds}
        \|\mathbf{k}\|_p \le \frac{1}{2D\alpha}, \quad
        \frac{1}{4\alpha} \|\mathbf{u}\|_{p'} \le D|\mathbf{k} \cdot \mathbf{u}| \le \pi, \quad\text{and}\quad
        |1 - e^{i D\mathbf{k} \cdot \mathbf{u}}| \ge \frac{\sqrt{2}}{2\pi\alpha} \|\mathbf{u}\|_{p'}.
    \end{equation}
\end{lem}

\begin{proof}
For a scalar $a \in \mathbb{R}$, let $[a]$ denote truncation toward zero,
namely the unique integer satisfying $|a - [a]| < 1$ and
$|[a]| \le |a|$. We extend this operation component-wise to vectors
$\mathbf{a} \in \mathbb{R}^d$. It satisfies
\begin{equation}
    \label{eq:round_props}
    \|\mathbf{a} - [\mathbf{a}]\|_\infty < 1
    \quad\text{and}\quad
    \|[\mathbf{a}]\|_p \le \|\mathbf{a}\|_p.
\end{equation}

Let $\mathbf{v} \in \mathbb{R}^d$ be an $\ell^p$-unit vector satisfying $\mathbf{v} \cdot \mathbf{u} = \|\mathbf{u}\|_{p'}$. For $1<p<\infty$, take
\[
    v_j=\frac{|u_j|^{p'-1}\operatorname{sgn}(u_j)}{\|\mathbf u\|_{p'}^{p'-1}}.
\]
For $p=1$, choose an index $j_0$ attaining $|u_{j_0}|=\|\mathbf u\|_\infty$ and take $\mathbf v=\operatorname{sgn}(u_{j_0})\mathbf e_{j_0}$. For $p=\infty$, take $v_j=\operatorname{sgn}(u_j)$.
We define the ideal frequency $\mathbf{a}$ and the quantized integer frequency $\mathbf{k}$ as
\[
\mathbf{a} := \frac{1}{2D\alpha} \mathbf{v} \quad\text{and}\quad \mathbf{k} := [\mathbf{a}].
\]

The first inequality in \eqref{eq:freq_quant_bounds} follows immediately from \eqref{eq:round_props}:
\[
\|\mathbf{k}\|_p = \|[\mathbf{a}]\|_p
\le \|\mathbf{a}\|_p
= \frac{1}{2D\alpha}\|\mathbf{v}\|_p
= \frac{1}{2D\alpha}.
\]

We estimate the phase $\mathbf{k} \cdot \mathbf{u}$. By the Triangle Inequality and Hölder's Inequality,
\[
|\mathbf{k} \cdot \mathbf{u}| \ge |\mathbf{a} \cdot \mathbf{u}| - |(\mathbf{k} - \mathbf{a}) \cdot \mathbf{u}| 
= \frac{1}{2D\alpha} \|\mathbf{u}\|_{p'} - |(\mathbf{k} - \mathbf{a}) \cdot \mathbf{u}|.
\]
Using $\|\mathbf{u}\|_1 \le d^{1/p} \|\mathbf{u}\|_{p'}$, $\|[\mathbf{a}] - \mathbf{a}\|_\infty < 1$, and the assumption $\alpha \le \frac{1}{4D d^{1/p}}$, we have
\[
|(\mathbf{k} - \mathbf{a}) \cdot \mathbf{u}|
\le \|[\mathbf{a}] - \mathbf{a}\|_\infty\|\mathbf{u}\|_1
\le d^{1/p}\|\mathbf{u}\|_{p'}
\le \frac{1}{4D\alpha}\|\mathbf{u}\|_{p'}.
\]
We obtain both lower and upper bounds:
\[
|\mathbf{k} \cdot \mathbf{u}|
\ge \left( \frac{1}{2D\alpha} - \frac{1}{4D\alpha} \right)
\|\mathbf{u}\|_{p'}
= \frac{1}{4D\alpha}\|\mathbf{u}\|_{p'},
\]
\[
|\mathbf{k} \cdot \mathbf{u}|
\le \|\mathbf{k}\|_{p}\|\mathbf{u}\|_{p'}
\le \frac{1}{2D\alpha}\|\mathbf{u}\|_{p'}
\le \frac{\pi}{D}.
\]
 This proves the second part of \eqref{eq:freq_quant_bounds}.

 It remains to prove the final inequality.
On the interval $[0, \frac{\pi}{2}]$, the normalized sinc function $g(t) = \frac{\sin(t)}{t}$  is decreasing.
Using this inequality and the lower bound for $|\mathbf{k} \cdot \mathbf{u}|$, we have
\[
|1 - e^{i D\mathbf{k} \cdot \mathbf{u}}|
= 2\sin\left(\frac{D|\mathbf{k} \cdot \mathbf{u}|}{2}\right)
\ge 2\sin\left(\frac{\|\mathbf{u}\|_{p'}}{8\alpha}\right)
= \frac{\|\mathbf{u}\|_{p'}}{4\alpha}
g\left(\frac{\|\mathbf{u}\|_{p'}}{8\alpha}\right)
\ge \frac{\|\mathbf{u}\|_{p'}}{4\alpha}g\left(\frac{\pi}{4}\right)
= \frac{\sqrt{2}}{2\pi\alpha}\|\mathbf{u}\|_{p'}.
\]
This completes the proof.
\end{proof}

\begin{lem}
    \label{lem:neighborset_segmented}
    Let $d,M,D\in\mathbb N$ and $m\in\mathbb Z_{\ge0}$ satisfy $D>m$. Let $p \in [1, \infty]$ and let $\mathcal{U} \subset (-\pi, \pi]^d$ be a finite set of coordinatewise shortest representatives of torus differences, containing $\mathbf{0}$ and having cardinality $v := |\mathcal{U}|$.
    Suppose 
    \(
    \|\mathbf{u}\|_{p'} \le \frac{\pi}{2 D d^{1/p}}\) for all \(\mathbf{u} \in \mathcal{U}\)
     and $M \ge 2 d^{1/p} v$,
    then there exists a polynomial $f \in \mathcal{P}(m, M, D, d)$ such that $f(\mathbf{0})=1$, $f\left(\frac{\mathbf{u}}{2\pi}\right) = 0$ for all $\mathbf{u} \in \mathcal{U} \setminus \{\mathbf{0}\}$, and
    \begin{equation}
        \label{eq:neighbor_segmented}
        \|f\|_{L^2(\mathbb T^d)}
        \leq \frac{(\sqrt{2})^{v-1} }{\sqrt{(M/v)^d(m+1)^d}}
        \prod_{\substack{\mathbf u\in\mathcal U\\0<\|\mathbf{u}\|_{p'} \leq \frac{\pi v}{MD}}}
        \frac{\pi v}{MD \|\mathbf{u}\|_{p'}}.
    \end{equation}
   
\end{lem}

\begin{proof}
Set $z:=\lfloor M/v\rfloor$ and partition the nonzero nodes as
\[
    \mathcal{I}:=\left\{\mathbf{u}\in\mathcal{U}\setminus\{\mathbf{0}\}: \|\mathbf{u}\|_{p'}\le\frac{\pi v}{MD}\right\},
    \qquad
    \mathcal{J}:=\mathcal{U}\setminus(\mathcal{I}\cup\{\mathbf{0}\}).
\]

For each $\mathbf{u}\in\mathcal I$, apply Lemma~\ref{lem:freq_quantization} with $\alpha=v/(2MD)$. This gives $\mathbf{k}(\mathbf{u})\in\mathbb Z^d$ such that
\begin{equation}
    \label{eq:quan_I}
    \|\mathbf{k}(\mathbf{u})\|_p\le\frac{M}{v},
    \qquad
    |1-e^{iD\mathbf{k}(\mathbf{u})\cdot\mathbf{u}}|
    \ge\frac{\sqrt{2}MD}{\pi v}\|\mathbf{u}\|_{p'}.
\end{equation}
For every quantized vector used in this proof, write
\[
    \mathbf{k}=\mathbf{k}^+-\mathbf{k}^-,
    \qquad
    k_j^+:=\max\{k_j,0\},
    \qquad
    k_j^-:=\max\{-k_j,0\},
\]
and define the recentered two-point interpolant
\begin{equation}
    \label{eq:recentered_quantized_factor}
    L_{\mathbf{u}}(\bm{\omega})
    :=e^{2\pi iD\mathbf{k}^-(\mathbf{u})\cdot\bm{\omega}}
    \frac{e^{2\pi iD\mathbf{k}(\mathbf{u})\cdot\bm{\omega}}-e^{iD\mathbf{k}(\mathbf{u})\cdot\mathbf{u}}}
    {1-e^{iD\mathbf{k}(\mathbf{u})\cdot\mathbf{u}}}.
\end{equation}
The two frequencies in $L_{\mathbf u}$ are $D\mathbf{k}^+(\mathbf u)$ and $D\mathbf{k}^-(\mathbf u)$. Since $\mathbf{k}(\mathbf u)$ is integer valued and $\|\mathbf{k}(\mathbf u)\|_p\le M/v$, both vectors belong to $D\{0,\ldots,z\}^d$. Moreover,
\[
    L_{\mathbf u}(\mathbf 0)=1,
    \qquad
    L_{\mathbf u}\left(\frac{\mathbf u}{2\pi}\right)=0,
\]
and the prefactor in \eqref{eq:recentered_quantized_factor} has unit modulus, so all pointwise bounds for the original factor remain unchanged.

Define the empty product to be one and set
\[
    h_0(\bm{\omega}):=\prod_{\mathbf u\in\mathcal I}L_{\mathbf u}(\bm{\omega})
    \in\mathcal P(0,z|\mathcal I|,D,d).
\]
Let
\[
    A_{z,m}(\bm{\omega})
    :=\frac{1}{(z+1)^d(m+1)^d}
    \sum_{\mathbf l\in\{0,\ldots,z\}^d}
    \sum_{\mathbf j\in\{0,\ldots,m\}^d}
    e^{2\pi i(D\mathbf l+\mathbf j)\cdot\bm{\omega}},
    \qquad
    h:=A_{z,m}h_0.
\]
Then $h\in\mathcal P(m,z(|\mathcal I|+1),D,d)$, $h(\mathbf0)=1$, and $h(\mathbf u/(2\pi))=0$ for $\mathbf u\in\mathcal I$. Because $D>m$, all $(z+1)^d(m+1)^d$ frequencies of $A_{z,m}$ are distinct. Parseval's identity and $z+1\ge M/v$ therefore give
\begin{equation}
    \label{equ:h_L2_bound_by_h0}
    \begin{aligned}
    \|h\|_{L^2(\mathbb T^d)}
    &\le \|A_{z,m}\|_{L^2(\mathbb T^d)}\|h_0\|_{L^\infty(\mathbb T^d)}\\
    &=\frac{\|h_0\|_{L^\infty(\mathbb T^d)}}{\sqrt{(z+1)^d(m+1)^d}}
    \le\frac{\|h_0\|_{L^\infty(\mathbb T^d)}}{\sqrt{(M/v)^d(m+1)^d}}.
    \end{aligned}
\end{equation}

Using \eqref{eq:quan_I}, we obtain
\begin{equation}
    \label{eq:h_bound}
    \|h_0\|_{L^\infty(\mathbb T^d)}
    \le(\sqrt2)^{|\mathcal I|}
    \prod_{\mathbf u\in\mathcal I}\frac{\pi v}{MD\|\mathbf u\|_{p'}}.
\end{equation}

For each $\mathbf u\in\mathcal J$, apply Lemma~\ref{lem:freq_quantization} with $\alpha=\|\mathbf u\|_{p'}/(2\pi)$. It gives a vector $\mathbf k(\mathbf u)\in\mathbb Z^d$ satisfying
\begin{equation}
    \label{eq:quan_J}
    \|\mathbf k(\mathbf u)\|_p
    \le\frac{\pi}{D\|\mathbf u\|_{p'}}<\frac{M}{v},
    \qquad
    |1-e^{iD\mathbf k(\mathbf u)\cdot\mathbf u}|\ge\sqrt2.
\end{equation}
Define $L_{\mathbf u}$ by \eqref{eq:recentered_quantized_factor} using this vector and set
\[
    g(\bm{\omega}):=\prod_{\mathbf u\in\mathcal J}L_{\mathbf u}(\bm{\omega})
    \in\mathcal P(0,z|\mathcal J|,D,d).
\]
Then
\begin{equation}
    \label{eq:g_bound}
    \|g\|_{L^\infty(\mathbb T^d)}\le(\sqrt2)^{|\mathcal J|}.
\end{equation}

Finally, let $f:=hg$. Since $|\mathcal I|+|\mathcal J|=v-1$ and $zv\le M$,
\[
    f\in\mathcal P(m,zv,D,d)\subseteq\mathcal P(m,M,D,d).
\]
The polynomial $f$ equals one at $\mathbf0$ and vanishes at $\mathbf u/(2\pi)$ for every nonzero $\mathbf u\in\mathcal U$. Combining \eqref{equ:h_L2_bound_by_h0}, \eqref{eq:h_bound}, and \eqref{eq:g_bound} gives
\[
    \|f\|_{L^2(\mathbb T^d)}
    \le\frac{(\sqrt2)^{v-1}}{\sqrt{(M/v)^d(m+1)^d}}
    \prod_{\mathbf u\in\mathcal I}\frac{\pi v}{MD\|\mathbf u\|_{p'}},
\]
which is \eqref{eq:neighbor_segmented}.
\end{proof}


Having established the lemmas above, we complete the proof by constructing polynomials $\{F_k\}_{k=1}^n$.
\begin{proof}[Proof of Theorem~\ref{thm:segmented-vandermonde}]

We fix  $k \in \{1, \dots, n\}$. Let $\mathcal{C}_k := \mathcal{N}_\infty(\mathbf{y}_k, \tau, \mathcal{X})$  with cardinality $\nu_k \le n^\star$.
We will construct a test polynomial $F_k \in \mathcal{P}(m, r, D, d)$ that interpolates 1 at $\frac{\mathbf{y}_k}{2\pi}$ and 0 at $\frac{\mathbf{y}_j}{2\pi}$ for $j \neq k$.

Let
\[
    m_{1} = \left\lceil\frac m2\right\rceil,
    \qquad
    m_{2} = m-m_{1}=\left\lfloor\frac m2\right\rfloor,
    \qquad
    K_{\mathrm{loc}}=\left\lfloor\frac{m_1}{n^\star}\right\rfloor.
\]

The hypothesis $\eta\ge4\pi\beta d/(K_{\mathrm{loc}}+1)$ permits Lemma~\ref{lem:localization} to be applied with bandwidth $m_1$. It gives a polynomial $G_k \in \mathcal{P}(m_{1}, 0, D, d)$ such that
\[
G_k\left(\frac{\mathbf{y}_k}{2\pi}\right) = 1,
\qquad
G_k\left(\frac{\mathbf{y}_j}{2\pi}\right) = 0
\quad(\mathbf y_j\in\mathcal X\setminus\mathcal C_k),
\qquad
\|G_k\|_{L^\infty} \le \left( 2 - e^{1/(2\beta)} \right)^{-n^{\star}/2}.
\]

We handle the remaining nodes within the cluster $\mathcal{C}_k$. Let $\mathcal{U} := \{ \mathbf{y}_j - \mathbf{y}_k : \mathbf{y}_j \in \mathcal{C}_k \}$, where each difference is identified with its coordinatewise representative in $(-\pi,\pi]^d$. Note that $\mathbf{0} \in \mathcal{U}$.
For these representatives,
\[
\|\mathbf y_j-\mathbf y_k\|_1
=|\mathbf y_j-\mathbf y_k|_1,
\qquad \mathbf y_j\in\mathcal C_k,
\]
where the right-hand side is the periodic distance of Definition~\ref{defi:metric_separation}.
We apply Lemma \ref{lem:neighborset_segmented} to $\mathcal{U}$ with $p=\infty$, $m_2$ and $M=r$. 
The size condition holds because $r\ge2n^\star\ge2\nu_k=2d^{1/\infty}\nu_k$.
Using the assumption $\tau \leq \frac {\pi}{2Dd}$, for any $\mathbf{y}_j\in \mathcal{N}_\infty(\mathbf{y}_k,\tau,\mathcal{X})$, we have
$$
|\mathbf{y}_j-\mathbf{y}_k|_1 \leq d \, |\mathbf{y}_j-\mathbf{y}_k|_\infty \leq d\tau \leq \frac {\pi}{2D}.
$$
There exists a polynomial $f \in \mathcal{P}(m_2, r, D, d)$ such that $f(\mathbf{0})=1$ and $f(\frac{\mathbf{u}}{2\pi})=0$ for non-zero $\mathbf{u} \in \mathcal{U}$.
We define the shifted polynomial $B_k(\bm{\omega})$ as
\(
B_k(\bm{\omega}) := f\left(\bm{\omega} - \frac{\mathbf{y}_k}{2\pi}\right), 
\) then we have

\[
B_k(\frac{\mathbf{y}_k}{2\pi}) = 1, \quad B_k(\frac{\mathbf{y}_j}{2\pi}) = 0 \text{ for } \mathbf{y}_j \in \mathcal{C}_k \setminus \{\mathbf{y}_k\},
\]
Lemma~\ref{lem:neighborset_segmented} first gives the following estimate with $\nu_k$ in place of $n^\star$. Replacing $\nu_k$ by $n^\star$ enlarges the factor $(\sqrt2)^{\nu_k-1}$, decreases the denominator, and adds only product factors that are at least one. Hence
\[
\|B_k\|_{L^2} = \|f\|_{L^2} \le \frac{(\sqrt{2})^{n^\star-1} }{\sqrt{(\frac{r}{n^\star})^d(m_2+1)^d}}
\prod_{\substack{1\le j\le n\\0<|\mathbf{y}_j-\mathbf{y}_k|_{1} \leq \frac{\pi n^\star}{rD}}}
\frac{\pi n^\star}{rD |\mathbf{y}_j-\mathbf{y}_k|_{1}}.
\]

Define $F_k(\bm{\omega}) := G_k(\bm{\omega}) B_k(\bm{\omega})$. By construction, $F_k$ satisfies the Lagrange interpolation property at the nodes $\frac{\mathbf{y}_j}{2\pi}, 1\leq j\leq n$. Besides,  $F_k \in \mathcal{P}(m, r, D, d)$.
We estimate the $L^2$ norm of $F_k$, i.e.,
\[
\|F_k\|_{L^2} \le  \|G_k\|_{L^\infty} \|B_k\|_{L^2} \le   \frac{\left( 2 - e^{1/(2\beta)} \right)^{-n^\star/2}(\sqrt{2})^{n^\star-1} }{\sqrt{(\frac{r}{n^\star})^d(m_2+1)^d}}
\prod_{\substack{1\le j\le n\\0<|\mathbf{y}_j-\mathbf{y}_k|_{1} \leq \frac{\pi n^\star}{rD}}}
\frac{\pi n^\star}{rD |\mathbf{y}_j-\mathbf{y}_k|_{1}}.
\]
Using the duality principle $\sigma_{\min} \ge \frac{1}{\sqrt{n} \max_k \|F_k\|_{L^2}}$ in Lemma~\ref{lem:minsvd_bound_by_lagInterp_high_dim}, we obtain
\[
    \sigma_{\min} \ge \frac{1}{\sqrt{n}} \left( 2 - e^{1/(2\beta)} \right)^{n^\star/2}\frac{\sqrt{(\frac{r}{n^\star})^d(m_2+1)^d} }{(\sqrt{2})^{n^\star-1}}\min_{1\leq k\leq n}
    \left\{
    \prod_{\substack{1\le j\le n\\0<|\mathbf{y}_j-\mathbf{y}_k|_{1} \leq \frac{\pi n^\star}{rD}}}
    \frac{rD}{\pi n^\star } |\mathbf{y}_j-\mathbf{y}_k|_{1}
    \right\}.
\]
    By the local geometry assumption, we have
    \begin{align*}
        \prod_{\substack{1\le j\le n\\0<|\mathbf{y}_j-\mathbf{y}_k|_{1} \leq \frac{\pi n^\star}{rD}}}
        \frac{rD}{\pi n^\star } |\mathbf{y}_j-\mathbf{y}_k|_{1}
        \ge \left(\frac{rD\delta}{\pi n^\star}\right)^{q_k}
        \ge \left(\frac{rD\delta}{\pi n^\star}\right)^{n^\star-1},
    \end{align*}
    which completes the proof.
\end{proof}

\bibliographystyle{plain}
\bibliography{reference_final}
\end{document}